\documentclass[preprint,11pt,authoryear]{elsarticle}

\usepackage{amsmath,amssymb,amsfonts,amsthm,mathtools}
\usepackage{mathrsfs}
\DeclareMathAlphabet{\pazocal}{OMS}{zplm}{m}{n}

\usepackage{algorithm}
\usepackage{algorithmic}
\usepackage{graphicx}
\usepackage{booktabs}
\usepackage{adjustbox}
\usepackage{pifont}
\usepackage{subcaption}
\usepackage{comment}
\usepackage{makecell}
\newcommand{\cmark}{\ding{51}}
\newcommand{\xmark}{\ding{55}}
\newcommand{\pmark}{$\triangle$}

\usepackage{setspace}
\usepackage[table]{xcolor}

\usepackage[normalem]{ulem}
\usepackage{xparse}

\newif\ifshowchanges
\showchangesfalse

\DeclareRobustCommand{\added}[1]{%
  \ifshowchanges
    \textcolor{blue}{#1}%
  \else
    #1%
  \fi
}

\AtBeginDocument{%
\setlength{\abovedisplayskip}{5pt plus 1pt minus 1pt}
\setlength{\belowdisplayskip}{5pt plus 1pt minus 1pt}
\setlength{\abovedisplayshortskip}{5pt plus 1pt minus 1pt}
\setlength{\belowdisplayshortskip}{5pt plus 1pt minus 1pt}
\setlength{\jot}{2pt}
}

\AtBeginDocument{%
\setlength{\textfloatsep}{5pt plus 1pt minus 2pt} 
\setlength{\floatsep}{5pt plus 1pt minus 2pt}     
\setlength{\intextsep}{5pt plus 1pt minus 2pt}    
\setlength{\abovecaptionskip}{2pt}
\setlength{\belowcaptionskip}{2pt}
}

\theoremstyle{plain}
\newtheorem{theorem}{Theorem}
\newtheorem{proposition}[theorem]{Proposition}
\newtheorem{lemma}[theorem]{Lemma}
\newtheorem{definition}[theorem]{Definition}
\newtheorem{assumption}[theorem]{Assumption}

\journal{European Journal of Operational Research}
\usepackage[
    colorlinks=true,
    linkcolor=blue,
    urlcolor=blue,
    citecolor=blue
]{hyperref}
\begin{document}
\begin{frontmatter}
\title{Shared Infrastructure Investment and Pricing: Stackelberg Equilibria in Risk-Aware Take-or-Pay Contracts}
\author[a,b]{Amal Sakr\corref{cor1}}
\ead{amalsakr@illinois.edu}
\author[a]{Andrea Araldo}
\author[b]{Tamer Başar}
\author[a]{Tijani Chahed}
\cortext[cor1]{Corresponding author}
\affiliation[a]{
    organization={Institut Polytechnique de Paris},
    city={Palaiseau},
    postcode={91120},
    country={France}
}
\affiliation[b]{
    organization={University of Illinois Urbana-Champaign},
    city={Urbana},
    postcode={61801},
    state={IL},
    country={USA}
}
\begin{abstract}
\noindent We study a shared infrastructure deployed by an Infrastructure Provider (InP) and used by multiple firms generating revenues through resource usage. We focus on a challenging setting where (i) infrastructure deployment requires substantial upfront investment, which the InP recovers via payments by firms that depend on their uncertain future revenues; (ii) firms' resource usage is jointly influenced by exogenous factors, infrastructure pricing, operational costs, and resource congestion; and (iii) firms exhibit heterogeneous risk aversion. These aspects are typical of emerging technologies, such as Mobile Edge Computing (MEC). \added{Yet, their joint effect on the InP’s capacity dimensioning and pricing decisions and the firms' usage commitments remains poorly understood.}

\added{We establish conditions for the uniqueness of the equilibrium among the firms.} We do so by introducing a Stackelberg game with risk-aware take-or-pay contracting and firm-side operational and congestion costs, in which the InP acts as the leader, while firms act as followers that share the infrastructure and commit upfront to future resource usage under uncertain revenues. Followers' heterogeneous risk aversion is modeled through Conditional Value-at-Risk (CVaR). We prove the existence of a Stackelberg equilibrium (SE), in which the followers' decisions constitute a generalized Nash equilibrium, and develop a polynomial-time algorithm that boundedly approximates the SE. We derive a lower bound on the followers' Probability of Profit (PoP). 
Simulations in a realistic MEC scenario show that higher followers' risk aversion reduces capacity, pricing, and leader profit, while increasing the lower bound on PoP.
\end{abstract}
\end{frontmatter}
\section{Introduction}
\noindent Several emerging technologies are technically ready for broad deployment, yet their adoption depends on costly shared infrastructure that must be installed before future utilization is known. Mobile edge computing (MEC), for example, requires distributed edge capacity to support latency-sensitive services~\citep{cruz2022edge}. 
Similar investment challenges arise in other infrastructure systems, such as electric vehicle (EV) charging networks~\citep{dimanchev2023accelerating} and electricity storage \citep{huang2025strategic}. 

In these settings, the Infrastructure Provider (InP), who deploys and maintains shared infrastructure, is faced with the crucial decisions of how much capacity to install and how to price it to the firms that will use such an infrastructure. Overbuilding leaves costly capacity underutilized, whereas underbuilding creates congestion, limits service quality, and may slow adoption. The InP is typically reluctant to take the risk of such investments, which prevents critical infrastructure from emerging \citep[p.~4]{Ehlers2014Infrastructure}, \citep{dimanchev2023accelerating}.
The reasons for this reluctance are twofold: (i) uncertainty regarding future infrastructure utilization and, consequently, the revenues that the InP may eventually collect, which materialize long after the upfront deployment costs are incurred; and (ii) the fact that revenues generated through the infrastructure are primarily captured by the firms using it, making it difficult for the InP to obtain a fair share of the value created during operation~\citep{GSMA2023FairShare, Telefonica2022OTTtraffic}.

\emph{Take-or-pay} contracts are designed to mitigate this barrier by requiring firms to commit in advance to the resources they will use (and pay for it) \cite[p.~60]{WorldBank2014PPA}. In doing so, risk is transferred from the InP to the firms: while a larger commitment provides the firms with access to greater infrastructure resources, which in principle allow them to generate higher future revenues, it also increases exposure to losses if realized revenues fall short of expectations.

\added{Existing research in Operations Research~\citep{zahur2025long,jin2007capacity, boyaci2010information} overlooks two important features of shared-infrastructure investment: (i)~congestion arising from the joint use of shared resources and (ii)~firms’ risk awareness. Congestion affects service quality and, consequently, investment profitability, whereas risk awareness shapes firms' commitment decisions.
 Therefore, omitting these features substantially limits the applicability of existing models to real-world take-or-pay contracts. A comprehensive method for analyzing such contracts is therefore still lacking. We fill this gap via an original Stackelberg game formulation, with risk-aware firms on a congested shared resource. The InP acts as the leader and determines capacity and access price, while the firms act as followers and choose commitments that determine infrastructure utilization and the InP’s revenue.
 Via this formulation, we establish the existence and uniqueness of an equilibrium among firms with heterogeneous levels of risk-aversion.}

The main contributions of this paper are as follows.
\vspace{-2mm}
\begin{itemize}
    \item  We derive a sufficient condition on the structure of firm-side operational and congestion costs that guarantees the uniqueness of the followers’ variational equilibrium (VE), which is a generalized Nash equilibrium. Followers are heterogeneous in their revenue profiles and risk aversion, modeled through Conditional Value-at-Risk (CVaR) (Theorem \ref{thm:gne_uniqueness}).
    \item We prove the existence of a Stackelberg equilibrium, which includes the decisions of the leader and the followers’ VE (Theorem~\ref{theor:stackelberg_exists}).
    \item We establish an ex-post profit guarantee for ex-ante risk-aware decisions by deriving a lower bound on the followers’ \emph{Probability of Profit} (PoP), \added{quantifying how CVaR-based commitments mitigate downside-profit risk} (Theorem~\ref{thm:global_bound_pi}). 
    \item Under realistic functional forms, we develop a polynomial-time procedure for computing an approximate Stackelberg equilibrium with a bounded optimality gap (Proposition \ref{prop:computational_tractability}).
\end{itemize}
\added{We evaluate the model through a realistic MEC case study calibrated with public cloud/edge cost data and Monte Carlo simulations of uncertain revenues. }
We further validate the model on a real-world dataset and compare it with representative state-of-the-art benchmarks, showing high InP profit, high utilization, and strong lower bounds on followers’ PoP.

The paper is organized as follows. \S~\ref{sec:related} reviews related work; \S~\ref{sec:model} presents the system model and Stackelberg game, establishes the followers’ VE existence and uniqueness, proves SE existence, and derives a PoP guarantee; \S~\ref{sec:analytical} gives the computation procedure; \S~\ref{sec:numerical} reports MEC results; and \S~\ref{sec:conclusion} concludes.


\section{Related Work and Positioning}\label{sec:related}

Our paper contributes to four streams of operations research literature: capacity--pricing decisions, capacity investment under uncertainty, take-or-pay contracts, and access pricing and infrastructure sharing. We review each stream and position our contributions.
\vspace{-3mm}
\paragraph{\textbf{Joint Capacity and Pricing Decisions}}
\added{\citet{maglaras2003pricing} and \citet{maglaras2005pricing} study shared-resource service systems where prices and congestion induce aggregate demand, considering single-class and differentiated guaranteed/best-effort services, respectively. Recent work examines stochastic or unknown price-dependent demand in joint pricing--inventory \citep{chen2024optimal} and capacity-adjustment/pricing models \citep{chen2026capacity}. In contrast, our demand arises endogenously from strategic, risk-aware firms' equilibrium resource commitments to shared capacity. Revenue uncertainty affects the InP through these commitments: firms choose resources under mean--CVaR preferences, thereby determining shared infrastructure utilization, access revenue, capacity, and price.
Relatedly,} \citet{basar2002revenue} study congestion-based bandwidth pricing with capacity scaling proportionally to users, and \citet{shen2007optimal} extend this setting to nonlinear pricing and incentive design. We instead endogenize capacity dimensioning. \citet{acemoglu2009price} analyze deterministic capacity--price competition with private capacities, while \citet{harks2024stackelberg} study a multi-leader Stackelberg game in which each leader chooses the capacity and price of a separate resource. By contrast, we model uncertainty and a single InP that jointly dimensions and prices one shared infrastructure.
The closest stream focuses on shared-storage models. \citet{zhang2024shared} study deterministic capacity dimensioning and leasing prices, while \citet{huang2025strategic} propose a Stackelberg model where a leader chooses storage investment and service fees and followers choose resources under peak-demand uncertainty using prospect theory. Unlike their setting, which assumes no mutual impact among followers, we consider the more realistic and challenging case of congestion in the shared resource, requiring characterization of a generalized Nash equilibrium (GNE). Moreover, none of these works provides an ex-post profit guarantee for ex-ante risk-aware decisions.
\vspace{-3mm}
\paragraph{\textbf{Capacity Investment under Uncertainty}}
\citet{huang2009value} study multistage capacity expansion, while 
\citet{yu2025value} examine risk-averse capacity planning, and \citet{wang2023production} consider CVaR-based production planning. When risk is modeled, it is borne by the capacity planner; in our setting, risk belongs to followers and shapes their resource commitments. Related studies address cloud server deployment \citep{liu2025efficient}, private-cloud capacity planning \citep{furman2025optimal}, electricity-market \citep{garcia2025strategic} and transmission investment \citep{lavrutich2023transmission}, and infrastructure co-investment \citep{sakr2026coanor}. These works are closely related to ours because they emphasize the value of accounting for uncertain future demand when making capacity decisions. We instead consider a more general setting in which the price is also optimized jointly with the capacity while anticipating the equilibrium 
commitments.
\vspace{-3mm}
\paragraph{\textbf{Take-or-Pay Contracts}}
\added{\citet{jin2007capacity} analyze capacity-reservation and take-or-pay contracts under stochastic demand, while \citet{boyaci2010information} study advance-selling for demand learning and capacity timing, and \citet{ozer2006strategic} use reservation and advance-purchase commitments for credible forecast sharing. \citet{keutz2025assessing} examine take-or-pay hydrogen import contracts and infrastructure planning under weather variability. We instead model take-or-pay contracts as endogenous resource-usage decisions of strategic firms: firms choose committed utilization under revenue downside risk, congestion, and shared infrastructure, and these risk-aware commitments feed back into the InP’s joint capacity--price decision.}
\vspace{-8mm}
\paragraph{\textbf{Access Pricing and Infrastructure Sharing}}
\citet{datar2022strategic} study Stackelberg pricing with generalized Nash games over 5G slicing resources while \citet{ardagna2012generalized} formulate cloud provisioning as a generalized Nash game. Related uncertainty-based pricing studies include \citet{jiang2020multi}, on miner-population uncertainty, and \citet{dhamal2025admission}, on stochastic demand in network-slicing pricing and admission control. In contrast, we study long-term capacity investment jointly with access pricing, where firms’ future revenue uncertainty enters the InP’s problem through followers’ equilibrium commitments.
\added{Distributionally robust GNE models, such as \citet{fabiani2023distributionally} and \citet{wen2025distributionally}, also study uncertainty in games, but their guarantees focus on constraint satisfaction, whereas ours focus on the ex-post probability that followers earn positive realized profit after making risk-aware ex-ante commitments.}

\added{Taken together, these differences position our paper as an operations-research contribution to shared-infrastructure design under uncertainty.
Table~\ref{tab:related_positioning} summarizes how our paper compares with representative studies across the main methodological dimensions.}
\begin{table}[t]
\centering
\caption{Positioning relative to the state of the art}
\label{tab:related_positioning}
\begin{adjustbox}{width=\textwidth}
\begin{tabular}{c c c c c c c}
\toprule
\textbf{Paper}
& \makecell{\textbf{Capacity}\\\textbf{dimensioning}}
& \makecell{\textbf{Pricing}\\\textbf{decision}}
& \makecell{\textbf{Uncertainty}\\\textbf{/ risk}}
& \makecell{\textbf{Take-or-pay}\\\textbf{or advance commitment}}
& \makecell{\textbf{Strategic downstream}\\\textbf{decisions}}
& \makecell{\textbf{Shared-capacity}\\\textbf{GNE/VE}} \\
\midrule

\citet{maglaras2003pricing}
& \cmark
& \cmark
& \cmark
& \xmark
& \xmark
& \xmark \\

\citet{chen2026capacity}
& \cmark
& \cmark
& \cmark
& \xmark
& \xmark
& \xmark \\

\citet{harks2024stackelberg}
& \cmark
& \cmark
& \xmark
& \xmark
& \cmark
& \xmark \\

\citet{huang2025strategic}
& \cmark
& \cmark
& \cmark
& \xmark
& \cmark
& \pmark \\

\citet{furman2025optimal}
& \cmark
& \xmark
& \cmark
& \xmark
& \xmark
& \xmark \\

\citet{liu2025efficient}
& \cmark
& \xmark
& \cmark
& \xmark
& \xmark
& \xmark \\

\citet{lavrutich2023transmission}
& \cmark
& \xmark
& \cmark
& \xmark
& \cmark
& \xmark \\

\citet{jin2007capacity}
& \cmark
& \pmark
& \cmark
& \cmark
& \cmark
& \pmark \\

\citet{boyaci2010information}
& \cmark
& \cmark
& \cmark
& \cmark
& \xmark
& \xmark \\

\citet{datar2022strategic}
& \xmark
& \cmark
& \xmark
& \xmark
& \cmark
& \cmark \\

\citet{dhamal2025admission}
& \xmark
& \cmark
& \cmark
& \xmark
& \cmark
& \xmark \\

\textcolor{blue}{\textbf{Our paper}}
& \textcolor{blue}{\cmark}
& \textcolor{blue}{\cmark}
& \textcolor{blue}{\cmark}
& \textcolor{blue}{\cmark}
& \textcolor{blue}{\cmark}
& \textcolor{blue}{\cmark} \\
\bottomrule
\end{tabular}
\end{adjustbox}
\scriptsize{\cmark: central feature; \pmark: partial/indirect feature; \xmark: not modeled.}
\end{table}
\vspace{-6mm}

\section{System Model and Stackelberg Game}\label{sec:model}
\subsection{System Overview and Notation}\label{sec:overview}
We consider an investment period \(I\) divided into time slots \(t \in \pazocal T=\{1,\dots,T\}\). An infrastructure provider (InP) deploys an infrastructure resource with capacity \(C\ge 0\) and sets an access price \(\theta\ge 0\). After observing \((C,\theta)\), each firm \(i\in\pazocal N=\{1,\dots,N\}\) chooses resources \(h_i^t\ge 0\) at each time slot \(t\in\pazocal T\) to serve end-user load and generate revenue.
This setting raises three main questions: (i) How should the InP choose \(C\) and \(\theta\) when future resource usage by firms is uncertain?
(ii) How much resource \(h_i^t\) should each firm \(i\in\pazocal N\) choose at each time slot~\(t\), under revenue uncertainty, congestion, and risk sensitivity?
(iii) What probabilistic profit guarantees can be established for the firms under uncertainty?
To answer these questions, we resort to a game-theoretic framework. Since capacity must be deployed before it can be used, it is natural to model the InP as the first decision maker, choosing \(C\) and \(\theta\) in anticipation of the firms' responses. After observing the InP decisions, and under take-or-pay contracting, firms commit upfront to future resource levels. This leads to a Stackelberg game in which the InP acts as the leader and the firms act as followers.
At the follower level, all firms share the same infrastructure capacity. Their decisions are mutually dependent: the resource commitment of one firm affects the congestion costs borne by the others, and all commitments are coupled through the common capacity constraint.
This leads to a generalized Nash equilibrium (GNE) problem among the followers, whose decisions are made under revenue uncertainty. \ref{appendix:notation} in the Supplementary Material summarizes the main notation.
\subsection{Stackelberg Game Formulation}\label{sec:stackelberg_game}

We formalize the Stackelberg game by first defining the follower game, then the leader problem and the resulting Stackelberg equilibrium.

\subsubsection{Follower Game}\label{sec:follower_problem}
For a given leader decision $(C,\theta)$, each follower $i\in\pazocal N$ enters into a take-or-pay contract by committing upfront to a resource reservation $h_i^t\in\mathbb R_+$ at each time slot $t\in\pazocal T$, where
$\mathbb{R}_+$ is the set of nonnegative real numbers. This commitment determines the access payment $p_\theta(h_i^t)$, where \(\theta\) is the access price per unit of resource $h_i^t$ in each time slot \(t\). This payment is due regardless of the realized revenue or the actual use of the reserved resources. We denote by $\mathbf h^t=(h_1^t,\ldots,h_N^t)\in\mathbb R_+^N$ the vector of resource commitments at time slot~$t$.
Revenue uncertainty is modeled on a probability space $(\Omega,\pazocal F,\mathbb P)$, where \(\Omega\) is the sample space, \(\pazocal F\) is the set of events, and \(\mathbb P\) is the probability measure.  For each realization $\omega\in\Omega$, the revenue of follower $i$ at time slot $t$ is denoted by $r_{i,\omega}^t(h_i^t)$. A larger commitment $h_i^t$ may increase revenue by allowing follower $i$ to serve more user load, but it also increases costs. First, it raises the access payment \(p_\theta(h_i^t)\) transferred to the InP.
Second, it may generate firm-side costs \(\psi_i^t(\mathbf h^t)\), borne by firm \(i\) apart from the payments transferred to the InP. These costs include the operational cost of managing the committed resource level and the congestion cost arising from shared infrastructure.
Thus, committing to a large $h_i^t$ exposes follower \(i\) to unfavorable outcomes when access payments and firm-side costs remain high while realized revenue is low. Conversely, committing to a small $h_i^t$ may limit the ability to serve demand and reduce realized revenue. Hence, each follower balances profitability against downside risk when choosing its resource commitment.
To do so, we first define the random profit~\(\Pi_i^t(\mathbf h^t;\theta)\) of follower \(i\) at time slot \(t\). 
Its realization in $\omega\in\Omega$ is defined as
\begin{equation}
\Pi_{i,\omega}^t(\mathbf h^t;\theta)
:=
r_{i,\omega}^t(h_i^t)
-
p_\theta(h_i^t)
-
\psi_i^t(\mathbf h^t)
\label{eq:profit_sp_realization}
\end{equation}
\label{page:assumption_continuity}We use \(h\in\mathbb R_+\) for a scalar resource commitment
and \(\mathbf h\in\mathbb R_+^N\) for a vector of resource commitments of all followers. 
We assume that $h \mapsto r_{i,\omega}^t(h)$ is a continuous function on $\mathbb R_+$; $p_\theta(h):\mathbb R_+\to\mathbb R_+$ with $(h,\theta)\mapsto p_\theta(h)$ continuous on \(\mathbb R_+^2\); and~$\mathbf h \mapsto \psi_i^t(\mathbf h)$ is continuous on $\mathbb R_+^N$. Define the random loss \(L_i^t(\mathbf h^t;\theta)\), with realization
\begin{equation}
L_{i,\omega}^t(\mathbf h^t;\theta)
:=
-\Pi_{i,\omega}^t(\mathbf h^t;\theta)
\label{eq:loss_i_t}
\end{equation}
Due to the continuity of the aforementioned functions, $\mathbf h\in\mathbb R_+^N \mapsto \Pi_{i,\omega}^t(\mathbf h;\theta)$ and $\mathbf h\mapsto L_{i,\omega}^t(\mathbf h;\theta)$ are continuous, for any $\omega\in\Omega$, $\theta$, and~$t\in\pazocal T$.

We next specify how downside risk is measured. For a confidence level \(\alpha_i\in(0,1)\), the Conditional Value-at-Risk (CVaR) of the loss is the expected loss in the worst \((1-\alpha_i)\) fraction of realizations, capturing how severe losses can be~\citep{rockafellar2000optimization}.
Following~\citep[Eqs.~(28) and (29)]{li2022risk},
\(\mathrm{CVaR}_{\alpha_i}(L_i^t(\mathbf h^t;\theta))
=
\mathbb E_\omega\!\left[
L_i^t(\mathbf h^t;\theta)
\,\middle|\,
L_i^t(\mathbf h^t;\theta)\ge
\mathrm{VaR}_{\alpha_i}(L_i^t(\mathbf h^t;\theta))
\right],
\)
where 
\(\mathrm{VaR}_{\alpha_i}(L_i^t(\mathbf h^t;\theta))\) denotes the Value-at-Risk of the loss \(L_i^t(\mathbf h^t;\theta)\) at confidence level \(\alpha_i\), defined as the smallest threshold not exceeded with probability at least \(\alpha_i\), i.e.,
\(
\mathrm{VaR}_{\alpha_i}\!\left(L_i^t(\mathbf{h}^t;\theta)\right)
=
\min \left\{ z \mid \mathbb P\!\left(L_i^t(\mathbf{h}^t;\theta) \le z\right) \ge \alpha_i \right\}
\).

We define the mean-CVaR utility of follower $i$ as follows. 
\begin{equation}
U_i(\mathbf h_i,\mathbf h_{-i};\theta)
:=
\sum_{t\in\pazocal T}
\left(
\mathbb E_\omega[\Pi_i^t(\mathbf h^t;\theta)]
-
\beta_i\,\mathrm{CVaR}_{\alpha_i}(L_i^t(\mathbf h^t;\theta))
\right),
\label{eq:u_follower}
\end{equation}
where \(
\mathbf h_i := (h_i^t)_{t\in\pazocal T}\in\mathbb R_+^T
\)
denotes the resource commitment vector of follower~\(i\) over the time horizon,
\(
\mathbf h^t := (h_1^t,\dots,h_N^t)\in\mathbb R_+^N
\)
denotes the vector of resource commitments across all followers at time slot~\(t\), and
\(
\mathbf h_{-i} := (\mathbf h_j)_{j\in\pazocal N,\; j\neq i}
\)
denotes the collection of resource commitment vectors of all followers except~\(i\).
Parameter $\beta_i\ge 0$ measures the firm's sensitivity to downside risk. When $\beta_i=0$, follower $i$ is \emph{risk-neutral} and maximizes expected profit only. When $\beta_i>0$, the follower is \emph{risk-averse}, placing additional weight on adverse outcomes. 
Allowing $(\beta_i,\alpha_i)$ to differ across followers enables the model to capture heterogeneity in followers' attitudes toward risk and uncertainty.
In~\eqref{eq:u_follower}, each time slot includes a CVaR-based downside-risk term, providing robustness against adverse revenue realizations within that slot~\citep{wang2025bi,yu2020robust}. Since CVaR is applied before the summation, adverse outcomes in one time slot are not canceled out by favorable outcomes in another slot.

For fixed \((C,\theta)\), the lower-level interaction among the followers defines a generalized Nash game with shared constraints~\citep{facchinei2007generalizedproblems}. 
A GNE is a strategy profile
\(
\mathbf H^*
=
(\mathbf h_1^*,\dots,\mathbf h_N^*)
\in \mathbb R_+^{NT},
\)
where \(\mathbf h_i^*=(h_i^{t*})_{t\in\pazocal T}\), such that, for every \(\mathbf h_i\in\mathbb R_+^T\) and every $(C, \theta)$,
\begin{equation}
    U_i(\mathbf h_i^*,\mathbf h_{-i}^*;\theta)
\ge
U_i(\mathbf h_i,\mathbf h_{-i}^*;\theta), \qquad i\in\pazocal N
\label{eq:utility_gne}
\end{equation} satisfying
\begin{align}
h_i^t+\sum_{j\in\pazocal N\setminus\{i\}} h_j^{t*}
&\le C,\qquad
h_i^t\ge 0,
\qquad \forall i\in\pazocal N,\; \forall t\in\pazocal T,
\label{eq:constraints}
\end{align}
where \(\mathbf h_{-i}^*:=(\mathbf h_j^*)_{j\in\pazocal N,\;j\neq i}\).
This GNE problem is consistent with pricing and sharing models \citep{jiang2020multi,dhamal2025admission,datar2022strategic}.

\noindent At a GNE, no follower can improve its utility 
\(U_i(\mathbf h_i^*,\mathbf h_{-i}^*;\theta)\)~\eqref{eq:u_follower}  
by changing its resource commitment, given the other followers' decisions and the shared capacity constraint. Since the follower game may admit multiple GNEs, an equilibrium-selection criterion is needed. 
We focus on the variational equilibrium (VE)~\citep{facchinei2007generalizedproblems} as the selected GNE. In a VE, the common capacity constraint is associated with the same shadow price for all followers.
In our setting, the VE has a clear economic interpretation: the shared capacity constraint is enforced uniformly across all followers, consistent with a leader who announces a common capacity and a common price, without giving any follower priority in infrastructure use. 
Existence of such a VE and sufficient conditions for its uniqueness are established later in \S~\ref{sec:ve_existence} and \S~\ref{sec:ve_uniqueness}. Under these conditions, the selected follower response is single-valued, so the leader problem in the next subsection is well defined.

\subsubsection{Leader Problem}\label{sec:leader_problem}
At the upper level, the leader chooses capacity dimensioning and access price while anticipating the followers' VE response. Its objective is to maximize profit by balancing access revenue against the investment cost \(Cost(C)\), which includes deployment, maintenance, and technical intervention costs and is continuous in \(C\).
For each \((C,\theta)\), let
\(
\mathbf H^*(C,\theta)
:=
\bigl(h_i^{t*}(C,\theta)\bigr)_{i\in\pazocal N,\;t\in\pazocal T}
\)
denote the selected VE response. The leader solves
\begin{equation}
(C^*,\theta^*)
\in
\arg\max_{(C,\theta)\in\pazocal Y}
\left\{
\sum_{t\in\pazocal T}\sum_{i\in\pazocal N}
p_\theta\!\left(h_i^{t*}(C,\theta)\right)
-
\mathrm{Cost}(C)
\right\}
\label{eq:leader_problem}
\end{equation}
where \(\pazocal Y\) is the leader feasible set, defined as
\begin{equation}
\pazocal Y
:=
\left\{
(C,\theta)\in\mathbb R_+^2
\;\middle|\;
0\le C\le \overline C,\;
0\le \theta\le \overline\theta
\right\}
\label{eq:leader_set}
\end{equation}
Here, $\overline C>0$ and $\overline\theta>0$ denote the maximum admissible capacity and price.

\subsubsection{Stackelberg Equilibrium}
We now introduce the Stackelberg equilibrium (SE). 
\begin{definition}\label{def:stack}
A Stackelberg equilibrium is a triple \((C^*,\theta^*,\mathbf H^*)\) such that: (i) \(\mathbf H^*=\mathbf H^*(C^*,\theta^*)\) is the selected VE, hence a GNE, of the follower game associated with \((C^*,\theta^*)\);
(ii) \((C^*,\theta^*)\) solves the leader problem~\eqref{eq:leader_problem}.
\end{definition}

\subsection{Stackelberg Equilibrium Analysis}\label{sec:equilibrium}
We now establish the existence of an SE by proving existence of the selected VE, providing sufficient conditions for its uniqueness, and then proving existence of an optimal solution to the leader problem.

\subsubsection{Existence of the Selected Variational Equilibrium}\label{sec:ve_existence}
We now analyze the selected VE of the follower game. 
The following standard assumption ensures that each follower's utility is differentiable with respect to its resource commitment.
\begin{assumption}
\label{ass:fct_c1}
For every \(i\in\pazocal N\), \(t\in\pazocal T\), and \(\omega\in\Omega\), 
\(h \mapsto r_{i,\omega}^t(h)\) is continuously differentiable on \(\mathbb R_+\).
For every \(\theta\in[0,\overline \theta]\), 
\(h \mapsto p_\theta(h)\) is continuously differentiable on \(\mathbb R_+\), and $(h,\theta)\mapsto \frac{\partial p_\theta(h)}{\partial h}$ is continuous on
$\mathbb R_+\times[0,\overline \theta]$.
For every \(i\in\pazocal N\) and \(t\in\pazocal T\),
\(\mathbf h \mapsto \psi_i^t(\mathbf h)\) is continuously differentiable on \(\mathbb R_+^N\).
\end{assumption}
We now ensure that VaR varies continuously with resource commitments. Hence, changes in capacity or price do not cause discontinuous jumps in the loss quantile through their effect on resource commitments. We further ensure that CVaR depends smoothly on the followers’ resource commitments.
\begin{assumption}\label{ass:var_continuity}
For every \(i\in\pazocal N\), let
\(\pazocal O\subset\mathbb R^N\) be an open set containing
\(
\mathcal K:=\Bigl\{\mathbf h\in\mathbb R_+^N:\sum_{j\in\pazocal N}h_j\le \overline C\Bigr\}
\).
For every fixed \(\theta\in[0,\overline \theta]\), and \(t\in\pazocal T\), let
\(
F_i^t(z,\mathbf h;\theta)
:=
\mathbb P\!\left(L_i^t(\mathbf h;\theta)\le z\right), (z,\mathbf h)\in\mathbb R\times \pazocal O,
\)
denote the cumulative distribution function (CDF) of the loss
\(L_i^t(\mathbf h;\theta)\). We assume that:
\begin{enumerate}
    \item the CDF \(F_i^t\) is jointly continuous in \((z,\mathbf h)\) on
    \(\mathbb R\times \pazocal O\);

    \item for every \(\mathbf h\in\pazocal O\) the random variable
    \(L_i^t(\mathbf h;\theta)\) admits a density \(f_i^t(z;\mathbf h,\theta)\), \(z\in\mathbb R\),
    which is strictly positive in a neighborhood of
    \(
    \mathrm{VaR}_{\alpha_i}(L_i^t(\mathbf h;\theta)).
    \)
\end{enumerate}
\end{assumption}

\begin{assumption}\label{ass:cvar}
For every \(i\in\pazocal N\), $t \in \pazocal T$ and fixed \(\theta\in[0,\overline \theta]\), assume that
\begin{enumerate}
\item for every \(\mathbf h\in \pazocal O\), random loss \(L_i^t(\mathbf h;\theta)\) is
integrable;
\item the derivative of the random loss term with respect to the resource vector
\(\mathbf h\) is
uniformly bounded on \(\pazocal O\); that is, there exists an integrable random variable \(Z_i^t\) such that
\(
\sup_{\mathbf h \in \pazocal O}
\left\|
\nabla_{\mathbf h}L_{i,\omega}^t(\mathbf h;\theta)
\right\|
\le Z_{i,\omega}^t, \text{for a.e. }\omega\in\Omega.
\)
\end{enumerate}
\end{assumption}

\begin{lemma}\label{lemma:cvar}
For every \(i\in\pazocal N\), $t \in \pazocal T$ and fixed \(\theta\in[0,\overline \theta]\), \(
\mathbf h \mapsto
\mathrm{VaR}_{\alpha_i}\!\left(L_i^t(\mathbf h;\theta)\right)
\)
is continuous on~\(\pazocal O\) and
\(\mathbf h\mapsto \mathrm{CVaR}_{\alpha_i}(L_i^t(\mathbf h;\theta))\)
is continuously differentiable on \(\pazocal O\), and thus on~\(\mathcal K\).
\end{lemma}
The proof of Lemma~\ref{lemma:cvar} follows from \citet{sakr2026continuityvarcontinuousdifferentiability}.

We impose curvature assumptions on the revenue, access-payment, and firm-side cost functions. 
These assumptions are standard in resource allocation models and are commonly used in network applications: 
concave revenue captures diminishing returns from additional demand~\citep{basar2002revenue,iiduka2018distributed}; 
convex access payments are consistent with pay-as-you-go pricing~\citep{basar2002revenue,chi2015fairness}; 
and convex firm-side costs capture increasing operational burden and congestion-induced performance degradation under shared-resource use~\citep{zhang2019price, johari2005efficiency, chen2012design}.


\begin{assumption}\label{ass:convex_follower}
For every \(i\in\pazocal N\), \(t\in\pazocal T\), and
\(\omega\in\Omega\),
\(h\mapsto r_{i,\omega}^t(h)\) is concave on \(\mathbb R_+\).
For every \(\theta\in[0,\overline\theta]\), \(h\mapsto p_\theta(h)\) is convex on \(\mathbb R_+\).
Moreover, for every fixed \(\mathbf h_{-i}^t\), \(
h\mapsto \psi_i^t(h,\mathbf h_{-i}^t)
\)
is convex on \(\mathbb R_+\).
\end{assumption}
\begin{proposition}\label{prop:ve_gne_existence}
For any fixed \((C,\theta)\in\pazocal Y\), the follower game admits at least one VE, and therefore at least one GNE.
\end{proposition}

\begin{proof}[Proof sketch]
Fix \((C,\theta)\in\pazocal Y\), and define
\(
X(C):=\{\mathbf H\in\mathbb R_+^{NT}:
\sum_{j\in\pazocal N}h_j^t\le C,\ \forall t\in\pazocal T\}.
\)
Let
\(
P_i(\mathbf h_i,\mathbf h_{-i};\theta)
:=-U_i(\mathbf h_i,\mathbf h_{-i};\theta)
\)
be follower \(i\)'s disutility, defined in \eqref{eq:u_follower}, and let
\(
\mathcal G(\mathbf H;\theta)
:=
(\nabla_{\mathbf h_i}P_i(\mathbf h_i,\mathbf h_{-i};\theta))_{i\in\pazocal N}.
\)
By definition \citep[\S~5.2]{facchinei2007generalized}, a VE solves the variational inequality
\(\mathrm{VI}(X(C),\mathcal G(\cdot;\theta))\), namely,
\(
\left\langle
\mathcal G(\mathbf H^*;\theta),
\mathbf H-\mathbf H^*
\right\rangle
\ge 0,\ \forall \mathbf H\in X(C).
\)
\citep[Proposition~2.2]{facchinei2007generalized} establishes the existence of a solution to this VI when
\(\mathbf H\mapsto\mathcal G(\mathbf H;\theta)\) is continuous and
\(X(C)\) is nonempty, compact, and convex. These properties of \(X(C)\) are easy to verify, while Assumptions~\ref{ass:fct_c1}--\ref{ass:cvar} and Lemma~\ref{lemma:cvar} imply the continuity of
\(\mathbf H\mapsto\mathcal G(\mathbf H;\theta)\).
Moreover, according to
\citep[Theorem~5]{facchinei2007generalizedproblems}, every VE is a GNE if \(X(C)\) is closed and convex, each follower's feasible set is a section of \(X(C)\), and, for every fixed \(\mathbf h_{-i}\), the mapping
\(\mathbf h_i\mapsto P_i(\mathbf h_i,\mathbf h_{-i};\theta)\)
is continuously differentiable and convex. Here, follower \(i\)'s feasible set is
\(
\{\mathbf h_i\in\mathbb R_+^T:
(\mathbf h_i,\mathbf h_{-i})\in X(C)\},
\)
and is therefore a section of \(X(C)\). Assumptions~\ref{ass:fct_c1}--\ref{ass:cvar} and Lemma~\ref{lemma:cvar} imply continuous differentiability, while Assumption~\ref{ass:convex_follower}, the monotonicity and convexity of CVaR
\citep[Prop.~2(iv),(v)]{pflug2000some}, and the preservation of convexity under expectations and finite sums
\citep[\S~3.2.1]{boyd2004convex}
imply convexity. Thus, every VE is a GNE.
See Supplementary Material, \ref{appendix:ve_gne_existence}, for the full proof.
\end{proof}

\subsubsection{Uniqueness of the Variational Equilibrium}\label{sec:ve_uniqueness}
We now establish the uniqueness of the VE. This ensures that each leader's decision \((C,\theta)\) induces a well-defined, single-valued follower response.
Function $\psi_i^t(\mathbf h^t)$ represents costs borne by firm $i$, apart from the payments transferred to the InP. We decompose it into an operational cost
$q_i^t(h_i^t)$ and a congestion cost $g_i^t(\mathbf h^t)$.
The operational cost $q_i^t(h_i^t)$ captures the firm-specific effort associated
with a larger committed resource level. A larger~$h_i^t$ may require additional
configuration, coordination, monitoring, and internal resource management. In
MEC, this may correspond to serving a larger workload over a reserved
compute slice, managing a larger amount of reserved capacity, orchestrating more containers, and increasing the firm's workflow and organizational costs. In EV charging, it may correspond to
scheduling and coordinating a larger set of charging commitments across stations
and time slots.
The congestion cost $g_i^t(\mathbf h^t)$ captures the loss borne by firm $i$
when other firms increase their commitments and thereby raise the load on
infrastructure components that remain shared across firms. In MEC,
even if CPU and memory are reserved per firm, cross-firm effects may still arise
through contention for storage I/O, data paths, network links, or backhaul
resources, consistent with~\citep{pu2012your}. 
Similar effects arise when multiple EVs share a distribution transformer: each charging decision adds to aggregate load, which in turn affects every EV’s charging cost~\citep{beaude2012charging}.

The following theorem states that, if firms are more sensitive to their own operational cost than to the congestion effect, then the equilibrium is unique. Sensitivities to operational cost and congestion cost are measured via bounds on second-order sensitivity coefficients of $q_i^t(\cdot)$ and $g_i^t(\cdot)$, namely 
\begin{align}
\label{eq:Q-and-G}
Q:=\inf_{i,t,h_i^t}(q_i^t)''(h_i^t)
\text{ and }
G:=\sup_{i,j,t,\mathbf h^t}\left|\partial^2 g_i^t(\mathbf h^t)/\partial h_i^t\partial h_j^t\right|
\text{, respectively}.
\end{align}

\begin{theorem}\label{thm:gne_uniqueness}
For any \((C,\theta)\in\pazocal Y\), a sufficient condition for the uniqueness of the variational equilibrium is that
(i)~the firm-side cost admits the decomposition
\(
\psi_i^t(\mathbf h^t)
=
q_i^t(h_i^t)
+
g_i^t(\mathbf h^t)
\);
(ii)~functions \(q_i^t\) and \(g_i^t\) are twice continuously differentiable;
(iii) the following condition holds
\begin{equation}
Q>\max_{i\in\pazocal N}\xi_i\,G,
\qquad
\xi_i
=
\frac{
2|\pazocal N|+|\pazocal N|\beta_i+\sum_{j\in\pazocal N}\beta_j
}{
2(1+\beta_i)
}.
\label{eq:ass_q_g}
\end{equation}
\end{theorem}
\begin{proof}[Proof sketch]
Fix \((C,\theta)\in\pazocal Y\). The VE solves
\(\mathrm{VI}(X(C),\mathcal G(\cdot;\theta))\).
According to \citep[Theorem~2.3.3(b)]{facchinei2003finitebook}, this VI has at most one solution if \(\mathcal G(\cdot;\theta)\) is strongly monotone on \(X(C)\). Define \(
\phi_i^t(h):=
-\mathbb E_\omega[r_{i,\omega}^t(h)]
+p_\theta(h)
+\beta_i\mathrm{CVaR}_{\alpha_i}(-r_i^t(h)+p_\theta(h)).
\)
Hence, it is easy to show that, each \((i,t)\)-component of \(\mathcal G\) can be written as
\( \mathcal G_i^t(\mathbf H;\theta) = \Phi_i^t(h_i^t) + M_i^t(\mathbf h^t), \) where
\(
\Phi_i^t (h):=(\phi_i^t)'(h),
\text{ and }
M_i^t:=(1+\beta_i)
\frac{\partial\psi_i^t(\mathbf h^t)}{\partial h_i^t}.
\) We now show that \(\Phi_i^t\) is monotone and that the mapping formed by the \(M_i^t\) terms is strongly monotone, which implies that \(\mathcal G\) is strongly monotone. It is easy to show that Assumptions \ref{ass:fct_c1} and \ref{ass:convex_follower}, \citep[Prop.~2(iv),(v)]{pflug2000some}, and Lemma~\ref{lemma:cvar}, imply that \(h\mapsto\phi_i^t(h)\) is convex and differentiable. Therefore, \(\Phi_i^t\) is monotone \citep[Prop.~3.17]{penot2013calculus}. Recall that \( \psi_i^t(\mathbf h^t) = q_i^t(h_i^t)+g_i^t(\mathbf h^t) \), thus $M_i^t$ becomes \( M_i^t(\mathbf h^t) := (1+\beta_i) \left( (q_i^t)'(h_i^t) + \frac{\partial g_i^t(\mathbf h^t)}{\partial h_i^t} \right). \) For each \(t\in\pazocal T\), define \( M^t(\mathbf h^t):=(M_i^t(\mathbf h^t))_{i\in\pazocal N}. \) We now show that \(M^t\) is strongly monotone. 
Let the Jacobian \(J^t(\mathbf h^t):=\nabla_{\mathbf h^t}M^t(\mathbf h^t)\), and let
\(S^t(\mathbf h^t):=\frac12\bigl(J^t(\mathbf h^t)+J^t(\mathbf h^t)^\top\bigr)\) be the symmetric part of $J^t(\mathbf h^t)$.  
It can be shown that the diagonal and off-diagonal entries of \(S^t\) satisfy
\(S_{ii}^t(\mathbf h^t)\ge(1+\beta_i)(Q-G)\) and, for \(i\ne j\),
\(|S_{ij}^t(\mathbf h^t)|\le\frac12\bigl((1+\beta_i)+(1+\beta_j)\bigr)G\). 
Therefore,
\(
S_{ii}^t(\mathbf h^t)-\sum_{j\ne i}|S_{ij}^t(\mathbf h^t)| \ge (1+\beta_i)(Q-\xi_iG)\overset{\eqref{eq:ass_q_g}}{>0}.
\)
Hence \(S^t(\mathbf h^t)\) is strictly row diagonally dominant with positive diagonal entries, and therefore positive definite; see \citep[Theorem~8.6]{gallier2022applications} and \citep[Theorem 2.4]{agosti2005theoretical}. Define \(m:=\min_{i\in\pazocal N}(1+\beta_i)(Q-\xi_iG)\), which is positive by \eqref{eq:ass_q_g}. 
For every \(i\), \(t\), and \(\mathbf h^t\), \(
S_{ii}^t(\mathbf h^t)-\sum_{j\ne i}|S_{ij}^t(\mathbf h^t)|
\ge
(1+\beta_i)(Q-\xi_iG)
\ge m
\)
holds.  
By \citep[Lemma~2]{luo2022directional}, every eigenvalue of \(S^t(\mathbf h^t)\) is therefore bounded below by \(m\); equivalently, \(S^t(\mathbf h^t)\succeq m I\), uniformly over \(t\in\pazocal T\) and \(\mathbf H\in X(C)\). 
Hence, the Jacobian criterion for strong monotonicity \citep[Prop.~3(a)]{parise2019variational}\citep[Prop.~2.3.2(c)]{facchinei2003finitebook} gives, summing over \(t\),
\(
\sum_t\bigl(M^t(\mathbf h^t)-M^t(\mathbf h^{t\prime})\bigr)^\top(\mathbf h^t-\mathbf h^{t\prime})
\ge m\|\mathbf H-\mathbf H'\|^2, \forall \mathbf H, \mathbf H' \in X(C).
\)
Combining this inequality with the monotonicity of the \(\Phi_i^t\) terms yields the strong monotonicity of \(\mathcal G\) on \(X(C)\). Thus, the VE is unique. 
See Supplementary Material,~\ref{appendix:gne_uniqueness} for the full proof.
\end{proof}
In condition~\eqref{eq:ass_q_g}, parameter \(\xi_i\) captures how the number of followers and
their risk attitudes enter~\eqref{eq:ass_q_g}; a larger \(\xi_i\) makes the
inequality harder to satisfy. When \(G=0\), condition~\eqref{eq:ass_q_g} reduces to the simple requirement \(Q>0\). 
\ref{appendix:condition_unique}, in the Supplementary Material,
discusses how the model can be interpreted when \eqref{eq:ass_q_g} is not
satisfied.

\subsubsection{Leader Problem Analysis and Stackelberg Equilibrium Existence}
By Theorem~\ref{thm:gne_uniqueness}, the follower VE reaction map
\(
\mathbf H^*:\pazocal Y\to\mathbb R_+^{N|\pazocal T|},
(C,\theta)\mapsto \mathbf H^*(C,\theta)
\)
is single-valued, which makes the leader problem~\eqref{eq:leader_problem} well-defined. 
We now show also that this map is continuous, which ensures that the leader can find an optimal $(C,\theta)$.

\begin{proposition}
\label{prop:continuity}
Reaction map
\(
(C,\theta)\mapsto \mathbf{H}^*(C,\theta)
\)
is continuous on $\pazocal Y$.
\end{proposition}
\begin{proof}[Proof sketch]
For $n\ge 1$. Let \((C_n,\theta_n)\to(C,\theta)\) in \(\pazocal Y\), and set
\(\mathbf H_n:=\mathbf H^*(C_n,\theta_n)\). 
Since \(C_n\to C\), the sequence \(\{C_n\}\) is bounded \citep[Theorem~3.2]{rudin1976principles}. 
Because \(\mathbf H_n\in X(C_n)\), each component satisfies
\(
0\le h_{i,n}^t\le \sum_{j\in\pazocal N}h_{j,n}^t\le C_n
\);
hence \(\{\mathbf H_n\}\) is bounded. Since bounded sequences admit convergent
subsequences~\citep[Lemma 1.4]{sasane2017friendly}, there exists a subsequence $\{\mathbf{H}_{n_k}\}_{k\ge1}$ and a
vector $\overline{\mathbf{H}}\in\mathbb R^{NT}$ such that
\(\mathbf H_{n_k}\to\overline{\mathbf H} \quad \text{as} \quad k \mapsto \infty\). 
It is easy to show that \(\overline{\mathbf H}\in X(C)\).
Let $\mathbf Y\in X(C)$ be arbitrary. We show that there exists a sequence
$\mathbf Y_k\in X(C_{n_k})$ such that
\(
\mathbf Y_k\to \mathbf Y
\). If \(C=0\), then \(X(C)=\{\mathbf0\}\), so take \(\mathbf Y_k=\mathbf0\). 
If \(C>0\), take
\(
\mathbf Y_k=(C_{n_k}/C)\mathbf Y
\).
Then \(\mathbf Y_k\in X(C_{n_k})\) and \(\mathbf Y_k\to\mathbf Y\). 
By the definition of \(\mathbf H_n\), for every \(k\ge1\),
\(\mathbf H_{n_k}=\mathbf H^*(C_{n_k},\theta_{n_k})\) solves
\(\mathrm{VI}(X(C_{n_k}),
\mathcal G(\cdot;\theta_{n_k}))\). Hence,
\(
\left\langle
\mathcal G(\mathbf H_{n_k};\theta_{n_k}),
\mathbf Y_k-\mathbf H_{n_k}
\right\rangle\ge0.
\)
As in the proof of Proposition~\ref{prop:ve_gne_existence}, \((\mathbf H,\theta)\mapsto \mathcal G(\mathbf H;\theta)\) is continuous. Hence
\(
\mathcal G(\mathbf H_{n_k};\theta_{n_k})
\to
\mathcal G(\overline{\mathbf H};\theta)
\).
By continuity of the inner product \citep[Lemma~3.2.2]{kreyszig1991introductory}, taking the limit yields
\(
\left\langle
\mathcal G(\overline{\mathbf H};\theta),
\mathbf Y-\overline{\mathbf H}
\right\rangle\ge0, \forall \mathbf Y\in X(C).
\)
Thus \(\overline{\mathbf H}\) solves \(\mathrm{VI}(X(C),\mathcal G(\cdot;\theta))\). 
By Theorem~\ref{thm:gne_uniqueness}, this VI has the unique solution
\(\mathbf H^*(C,\theta)\), so \(\overline{\mathbf H}=\mathbf H^*(C,\theta)\). 
Therefore every convergent subsequence of \(\{\mathbf H_n\}\) has the same limit \(\mathbf H^*(C,\theta)\), and since \(\{\mathbf H_n\}\) is bounded, the whole sequence satisfies
\(\mathbf H_n\to\mathbf H^*(C,\theta)\). 
Because \((C_n,\theta_n)\to(C,\theta)\) was arbitrary, the map
\((C,\theta)\mapsto\mathbf H^*(C,\theta)\) is continuous on~\(\pazocal Y\). 
See Supplementary Material,~\ref{appendix:continuity} for the full proof.
\end{proof}
\begin{proposition}\label{prop:leader_exists}
Leader problem~\eqref{eq:leader_problem} admits at least one optimal solution.
\end{proposition}

\begin{proof}
Let \(\Phi(C,\theta)\) denote the leader objective in~\eqref{eq:leader_problem}. 
By Proposition~\ref{prop:continuity}, \((C,\theta)\mapsto \mathbf H^*(C,\theta)\) is continuous on \(\pazocal Y\). 
By the assumptions in \S~\ref{sec:follower_problem} and \S~\ref{sec:leader_problem}, the mappings \((h,\theta)\mapsto p_\theta(h)\) and \(C\mapsto \mathrm{Cost}(C)\) are continuous. Hence, \(\Phi (C,\theta)\) is continuous on \(\pazocal Y\) \eqref{eq:leader_set}. 
Moreover, \(\pazocal Y=[0,\overline C]\times[0,\overline\theta]\) is compact. 
Thus, by \citep[Theorem~4.4.2]{abbott2016understanding},~\eqref{eq:leader_problem} admits a solution.
\end{proof}

We show that the Stackelberg game admits an equilibrium \((C^*,\theta^*,\mathbf{H}^*)\).

\begin{theorem}\label{theor:stackelberg_exists}
A Stackelberg equilibrium (Def.~\ref{def:stack}) exists.
\end{theorem}
\begin{proof}
By Proposition~\ref{prop:leader_exists}, the leader problem \eqref{eq:leader_problem} has an optimal solution \((C^*,\theta^*)\in\pazocal Y\). By Theorem~\ref{thm:gne_uniqueness}, the corresponding follower game admits a unique VE \(\mathbf{H}^*(C^*,\theta^*)\). Hence,
\(
\bigl(C^*,\theta^*,\mathbf{H}^*(C^*,\theta^*)\bigr)
\)
is a SE.
\end{proof}

\subsection{Probabilistic Profit Guarantee}\label{sec:profit_guarantee}
While the SE is defined through the ex ante utility~\eqref{eq:u_follower}, we are interested in the Probability of Profit (PoP), i.e., the probability that a firm obtains positive realized profit at equilibrium. For follower \(i\), the realized profit is
\begin{equation}
\Pi_{i,\omega}^*
:=
\sum_{t\in\pazocal T} r_{i,\omega}^t(h_i^{t*})
-
\sum_{t\in\pazocal T}
\Bigl(p_{\theta^*}(h_i^{t*})+\psi_i^t(\mathbf h^{t*})\Bigr),
\label{eq:pi_star}
\end{equation}
where \(R_{i,\omega}^*:=\sum_{t\in\pazocal T} r_{i,\omega}^t(h_i^{t*})\) is the realized total revenue, with corresponding random variable \(R_i^*\), and \(K_i^*:=\sum_{t\in\pazocal T}
\Bigl(p_{\theta^*}(h_i^{t*})+\psi_i^t(\mathbf h^{t*})\Bigr)\) is the cost. 
We first exclude two trivial cases. If \(K_i^*\ge \mathbb E_\omega[R_i^*]\), then positive profit is not achieved even in expectation, so we cannot certify positive realized profit. If \(K_i^*=0\), follower~\(i\) chooses no resources and is therefore inactive. Apart from these cases, we next derive a lower bound on the PoP,  $\mathbb{P}(\Pi_{i,\omega}^*>0).$
\begin{theorem}\label{thm:global_bound_pi}
For any Stackelberg equilibrium \((C^*,\theta^*,\mathbf H^*)\), and for any follower \(i\in\pazocal N\), assume that
\(
0<K_i^*<\mathbb E_\omega[R_i^*],
\text{and } \mathrm{Var}_\omega(R_i^*)<\infty .
\) Then
\begin{equation}
\mathbb P\bigl(\Pi_{i,\omega}^*>0\bigr)\ge \nu_i \quad \text{where } \nu_i
:=
\frac{1}{1+
\dfrac{\mathrm{Var}_\omega(R_i^*)}
{\bigl(\mathbb E_\omega[R_i^*]-K_i^*\bigr)^2}}
\label{eq:bound_nu}
\end{equation}
Moreover, if the confidence level \(\alpha_i\) required by follower \(i\) is not too large, then the Probability of Profit is guaranteed to be no less than \(\alpha_i\); in fact,
\begin{equation}
\mathbb P\bigl(\Pi_{i,\omega}^*>0\bigr)
\ge \max\{\nu_i,\alpha_i\}:=\hat{\nu_i},
\qquad \text{if } \alpha_i < \bar\alpha_i,
\label{eq:profit_bound_final}
\end{equation}
where
\(
\bar\alpha_i
:=
\sup\left\{
\alpha_i\in(0,1) \text{ such that } 
\sum_{t\in\pazocal T}
\mathrm{CVaR}_{\alpha_i}
\bigl(L_i^t(\mathbf h^{t*};\theta^*)\bigr)<0
\right\}.
\)
\end{theorem}
\begin{proof}[Proof sketch]
Fix \(i\in \pazocal N\). By~\eqref{eq:pi_star}, $\Pi_{i,\omega}^*>0$ \(\iff\)
\(R_{i,\omega}^*>K_i^*\). Since \(R_i^*\ge0\), \(K_i^*<\mathbb E_\omega[R_i^*]\), and
\(\mathrm{Var}_\omega(R_i^*)<\infty\), we apply the Paley--Zygmund inequality
~\citep[Page~8]{kahane1985some}(see also \citep[Section 3]{ghosh2002probability}) to \(R_i^*\) with
\(\vartheta_i:=K_i^*/\mathbb E_\omega[R_i^*]\in(0,1)\). This gives
\begin{equation}
    \mathbb P(R_{i,\omega}^*>K_i^*)
\ge
\frac{(1-\vartheta_i)^2(\mathbb E_\omega[R_i^*])^2}
{(1-\vartheta_i)^2(\mathbb E_\omega[R_i^*])^2+\mathrm{Var}_\omega(R_i^*)}
\label{eq:paley}
\end{equation} 
Substituting \(\vartheta_i\) into \eqref{eq:paley}, noting that
\(
\Bigl(1-\frac{K_i^*}{\mathbb E_{\omega}[R_i^*]}\Bigr)^2
\bigl(\mathbb E_{\omega}[R_i^*]\bigr)^2
=
\bigl(\mathbb E_{\omega}[R_i^*]-K_i^*\bigr)^2,
\) and dividing the numerator and denominator by
\(
\bigl(\mathbb E_{\omega}[R_i^*]-K_i^*\bigr)^2,
\)
we obtain \(\mathbb P(R_{i,\omega}^*>K_i^*)\ge \nu_i\) \eqref{eq:bound_nu}.
Assume now that \(\alpha_i<\bar\alpha_i\). By the definition of \(\bar\alpha_i\),
\(
\sum_{t\in\pazocal T}
\mathrm{CVaR}_{\alpha_i}
\bigl(L_i^t(\mathbf h^{t*};\theta^*)\bigr)<~0.
\)
Let \(L_i^*:=-\Pi_i^* =\sum_{t\in\pazocal T}L_i^t(\mathbf h^{t*};\theta^*)\) be the total equilibrium loss. Subadditivity of CVaR
\citep[p.~15]{pflug2000some} gives
\(
\mathrm{CVaR}_{\alpha_i}(L_i^*)
\le
\sum_{t\in\pazocal T}
\mathrm{CVaR}_{\alpha_i}
\bigl(L_i^t(\mathbf h^{t*};\theta^*)\bigr)
<0 .
\)
Moreover, by the definition of VaR \citep[Eq.~(28),(29)]{li2022risk},
\(\mathbb P(L_i^*\le \mathrm{VaR}_{\alpha_i}(L_i^*))\ge\alpha_i\), and
\(\mathrm{VaR}_{\alpha_i}(L_i^*)\le\mathrm{CVaR}_{\alpha_i}(L_i^*)\) 
\citep[Prop.~4]{pflug2000some}. Hence \(\mathrm{VaR}_{\alpha_i}(L_i^*)<0\), and therefore
\(
\mathbb P(\Pi_{i,\omega}^*>0)
=
\mathbb P(L_i^*<0)
\ge
\mathbb P(L_i^*\le \mathrm{VaR}_{\alpha_i}(L_i^*))
\ge
\alpha_i .
\)
Combining this bound with \(\mathbb P(\Pi_{i,\omega}^*>0)\ge\nu_i\) proves~\eqref{eq:profit_bound_final}. See Supplementary Material,~\ref{appendix:global_bound_pi} for the full proof.
\end{proof}
The theorem shows that the PoP is protected by two complementary effects. First, the bound \(\nu_i\) in~\eqref{eq:bound_nu} measures how safe the firm's profit is against revenue uncertainty, as captured by the variance \(\mathrm{Var}_\omega(R_i^*)\). The expected profit is positive, since
\(\mathbb E_\omega[\Pi_i^*]=\mathbb E_\omega[R_i^*]-K_i^*>0\).
This positive margin acts as a buffer against revenue fluctuations: the larger the margin and the smaller the variance of \(R_i^*\), the stronger is the guarantee.
Second, the CVaR term captures how the follower's risk sensitivity affects the PoP. Since the loss is the negative of profit, a negative adverse-tail loss means that the follower remains profitable even under unfavorable revenue realizations. Thus, when \(\alpha_i<\bar\alpha_i\), the CVaR guarantees~\(
\mathbb P(\Pi_{i,\omega}^*>0)\ge \alpha_i.
\)
This shows that the follower can improve its guaranteed PoP by choosing a higher confidence level \(\alpha_i\), as long as the corresponding adverse-tail loss remains negative. 
\section{Equilibrium Characterization and Computation}\label{sec:analytical}
We now specify standard functional forms for the revenue, payment, and firm-side cost functions to characterize and compute the SE.
\subsection{Revenue, Pricing, and Firm-Side Cost Models}
The revenue of follower \(i\) at time slot \(t\) under realization \(\omega\) is modeled as~\citep[Eq. (1)]{basar2002revenue} \citep[Eq. (4)]{iiduka2018distributed}:
\begin{align}
r_{i,\omega}^t(h_i^t)
=
a_{i,\omega}^t \ln(1+h_i^t),
\qquad
\forall i\in\pazocal N,\; t\in\pazocal T,\; \omega\in\Omega 
\label{eq:revenue}
\end{align}
where \(a_{i,\omega}^t\) is a realization of the random revenue coefficient \(a_i^t\) which captures exogenous variation in the revenue level. The logarithmic form captures diminishing marginal returns: increasing \(h_i^t\) improves revenue, but each additional unit of resource generates a smaller incremental gain.

Following usage-based pricing models~\citep[Eq. (1)]{basar2002revenue} \citep[Eq. (2)]{chi2015fairness}, the payment of follower \(i\) at time  slot \(t\) is
\begin{align}
p_\theta(h_i^t)
=
\theta h_i^t,
\qquad
\forall i\in\pazocal N,\; t\in\pazocal T
\label{eq:price}
\end{align}
where \(\theta\) is the access price per unit of reserved resource $h_i^t$ in time slot $t$.

The firm-side cost incurred by follower \(i\) at time slot \(t\) is modeled as
\begin{align}
\psi_i^t(\mathbf h^t)
=
\frac{\delta}{2}(h_i^t)^2
+
\gamma\, h_i^t\sum_{j\in\pazocal N\setminus\{i\}} h_j^t, \qquad
\forall i\in\pazocal N,\; t\in\pazocal T
\label{eq:psi}
\end{align}
where \( \delta>0 \) is the operational cost parameter, and \(\gamma\ge 0\) is the congestion cost parameter. The first term captures the increasing operational cost of managing a larger own resource commitment. The second term captures congestion from shared infrastructure: it increases with follower $i$'s own commitment and with the aggregate commitment of the other followers, and vanishes if either of these two quantities is zero.

\begin{lemma}
\label{lemma:delta_gamma_uniqueness}
The sufficient uniqueness
condition \eqref{eq:ass_q_g} in Theorem~\ref{thm:gne_uniqueness} 
becomes
\begin{equation}
\delta>
\gamma\max_{i\in\pazocal N}\xi_i, \qquad\text{where \(\xi_i\) is defined in~\eqref{eq:ass_q_g}}
\label{eq:delta_gamma_uniqueness}
\end{equation}
\end{lemma}

\begin{proof}
From~\eqref{eq:psi}, \(Q= (q_i^t)''(h_i^t)=\delta\). 
Also, for \(j\neq i\),
\(
\frac{\partial^2 g_i^t(\mathbf h^t)}{\partial h_i^t\partial h_j^t}=\gamma,
\frac{\partial^2 g_i^t(\mathbf h^t)}{\partial (h_i^t)^2}=~0.
\)
Hence \(G=\gamma\).
Substituting \(Q=\delta\) and \(G=\gamma\) into~\eqref{eq:ass_q_g} gives~\eqref{eq:delta_gamma_uniqueness}. 
See Supplementary Material,~\ref{appendix:delta_gamma_uniqueness} for the full proof.
\end{proof}
\subsection{Follower-Equilibrium Characterization and Comparative Statics}
We now state the implicit characterization of the follower VE.

\begin{proposition}
\label{prop:follower_explicit}
For every \(i\in\pazocal N\),
the VE resource commitment \(h_i^{t*}\) is
\begin{equation}
h_i^{t*}
=
\max\left\{
0,\,
\frac{
-\kappa_i^t(\delta)
+
\sqrt{(\kappa_i^t(-\delta))^2+4(1+\beta_i)\delta A_i^t}
}{
2(1+\beta_i)\delta
}
\right\},
\label{eq:explicit_hi_star}
\end{equation}
where
\(
\kappa_i^t(x)
:=
(1+ \beta_i) (\theta
+
\gamma\sum_{j\in\pazocal N\setminus\{i\}} h_j^{t*}
+ x) + 
\lambda^{t*},
\quad x\in\mathbb R
\), $A_i^t:=\mathbb E[a_i^t]-\beta_i\,\mathrm{CVaR}_{\alpha_i}(-a_i^t)$ is the mean-CVaR benefit factor of follower \(i\) at time slot \(t\)
and \(\lambda^{t*}\) is the shadow price of the capacity constraint. 
\end{proposition}
\begin{proof}[Proof sketch]
Fix \(t\in\pazocal T\) and write \(b_i:=1+\beta_i\). 
Substituting~\eqref{eq:revenue}, \eqref{eq:price}, and \eqref{eq:psi} into \eqref{eq:u_follower}, and using translation invariance and positive homogeneity of CVaR \citep[Prop.~2(i),(ii)]{pflug2000some}, gives 
\(
A_i^t\ln(1+h_i^t)
-b_i\theta h_i^t
-b_i\frac{\delta}{2}(h_i^t)^2
-b_i\gamma h_i^t\sum_{j\ne i}h_j^t.
\)
Since the follower objective function at GNE problem \eqref{eq:utility_gne} has a
concave objective function~\eqref{eq:u_follower} and convex constraints \eqref{eq:constraints}, the
Karush--Kuhn--Tucker conditions~\citep[\S 5.5.2]{boyd2004convex} provide necessary and sufficient optimality conditions.
For \(h_i^{t*}>0\), from KKT conditions, we obtain
\(
\frac{A_i^t}{1+h_i^{t*}}
=
b_i\theta+b_i\delta h_i^{t*}
+b_i\gamma\sum_{j\ne i}h_j^{t*}
+\lambda^{t*}.
\)
Rearranging gives the quadratic equation
\(
b_i\delta(h_i^{t*})^2
+\kappa_i^t(\delta)h_i^{t*}
+\kappa_i^t(0)-A_i^t=0 .
\)
Solving this quadratic yields the expression inside the maximum in~\eqref{eq:explicit_hi_star}. 
See Supplementary Material,~\ref{appendix:follower_explicit} for the full proof.
\end{proof}
\paragraph{\textbf{Observation}} Proposition~\ref{prop:follower_explicit} yields comparative statics for the follower VE; details are in Supplementary Material~\ref{appendix:follower_comparative_statics}. A higher \(A_i^t\) increases follower \(i\)'s willingness to commit resources, while higher effective marginal costs of resources, through \(\theta\), \(\delta\), \(\gamma\), \(\lambda^{t*}\), or other followers' commitments, make additional resources less attractive and lower \(h_i^{t*}\).

\subsection{Procedure to Approximate the Stackelberg Equilibrium}
\label{sec:proc}

We present the main steps of the procedure for computing an approximate SE.
For a fixed price \(\theta\), the procedure works as follows.

\begin{enumerate}
\item \label{step:unconstrained} For each time slot \(t\in\pazocal T\), compute the unconstrained aggregate
commitment \(\widehat{s}^t(\theta)\) implied by the follower VE characterized in
Proposition~\ref{prop:follower_explicit}, after dropping the shared-capacity
constraint~\eqref{eq:constraints}.

\item \label{step:capacity_candidates}
For the fixed price \(\theta\), the total equilibrium commitment served at time
slot \(t\) is
\(
\min\{\widehat{s}^t(\theta),C\}.
\)
Hence the leader's fixed-price value is
\begin{equation}
\phi(\theta)
=
\max_{0\le C\le \overline C}
\left\{
\theta\sum_{t\in\pazocal T}\min\{\widehat{s}^t(\theta),C\}
-\mathrm{Cost}(C)
\right\}
\label{eq:phi_theta}
\end{equation}
To solve this problem, sort the values
\(\widehat{s}^t(\theta)\), \(t\in\pazocal T\), together with \(0\) and
\(\overline C\). These breakpoints divide \([0,\overline C]\) into intervals
on which each term \(\min\{\widehat{s}^t(\theta),C\}\) is either equal to \(C\)
or to $\widehat{s}^t(\theta)$ throughout the interval. We solve \eqref{eq:phi_theta} on each interval and collect
the resulting maximizers, together with the breakpoints, in a finite candidate
set \(\mathcal C(\theta)\). The optimal capacity \(C^*(\theta)\) is the element
of \(\mathcal C(\theta)\) that gives the largest~\eqref{eq:phi_theta}.

\item \label{step:price_grid} 
Optimize $\phi(\theta)$ over a uniform price grid. Let $\theta_{\Delta}$ be the
best grid price, set $C^*=C^*(\theta_{\Delta})$, and recover the VE at
$(C^*,\theta_{\Delta})$.
\end{enumerate}

\begin{proposition}
\label{prop:computational_tractability}
For any \(\eta>0\), the procedure computes an approximate SE with optimality gap
at most \(\eta\) in time polynomial in \(|\pazocal N|\), \(|\pazocal T|\),
\(1/\eta\), \(1/(\delta-\gamma)\), and \(\Gamma_C\), with logarithmic dependence
on \(1/\varepsilon\), where \(\varepsilon\) is the VE accuracy and \(\Gamma_C\)
is the cost of solving one one-dimensional capacity subproblem in
Step~\ref{step:capacity_candidates}.
\end{proposition}

\begin{proof}[Proof sketch]
For each fixed grid price \(\theta\), Steps~\ref{step:unconstrained}--
\ref{step:capacity_candidates} compute the exact maximizer \(C^*(\theta)\) of
the fixed-price problem \(\phi(\theta)\) \eqref{eq:phi_theta}; hence the only approximation comes from
discretizing \([0,\overline\theta]\). Let \(\Theta_\Delta\) be the price grid,
containing \(0\) and \(\overline\theta\), with mesh \(\Delta_\theta\), i.e., the
maximum distance between consecutive grid points. Let
\(\theta^*\in\arg\max_{\theta\in[0,\overline\theta]}\phi(\theta)\). By the grid
construction, there exists \(\tilde\theta\in\Theta_\Delta\) such that
\(|\theta^*-\tilde\theta|\le\Delta_\theta/2\). For any fixed \(C\), changing \(\theta\)
affects \(\theta\sum_{t\in\pazocal T}\min\{\widehat s^t(\theta),C\}\) through the
multiplier \(\theta\) and through the aggregate responses
\(\widehat s^t(\theta)\). For all \(\theta_1,\theta_2\in[0,\overline\theta]\), the multiplier effect is bounded by
\(
|\theta_1-\theta_2|
\sum_{t\in\pazocal T}\min\{\widehat s^t(\theta_1),C\}
\le
|\pazocal T|\overline C|\theta_1-\theta_2|,
\)
since \(\min\{\widehat s^t(\theta),C\}\le\overline C\). The response effect is bounded
using the fact that \(x\mapsto \min\{x,C\}\) is 1-Lipschitz and
\(
|\widehat s^t(\theta_1)-\widehat s^t(\theta_2)|
\le |\pazocal N|(\delta-\gamma)^{-1}|\theta_1-\theta_2|.
\)
Thus
\(
\overline\theta
\sum_{t\in\pazocal T}
\left|
\min\{\widehat s^t(\theta_1),C\}
-
\min\{\widehat s^t(\theta_2),C\}
\right|
\le
\overline\theta|\pazocal T||\pazocal N|(\delta-\gamma)^{-1}
|\theta_1-\theta_2|
\). Thus objective in \eqref{eq:phi_theta} is uniformly
\(L_\phi\)-Lipschitz in \(\theta\), with
\(L_\phi=|\pazocal T|(\overline C+\overline\theta|\pazocal N|/(\delta-\gamma))\);
thus \(\phi\) is \(L_\phi\)-Lipschitz. Choose \(\Delta_\theta\le 2\eta/L_\phi\). Therefore
\(\phi(\theta^*)-\phi(\tilde\theta)\le
L_\phi|\theta^*-\tilde\theta|\le L_\phi\Delta_\theta/2\le\eta\). Since the
procedure selects \(\theta_\Delta\in\arg\max_{\theta\in\Theta_\Delta}\phi(\theta)\),
we have \(\phi(\theta_\Delta)\ge\phi(\tilde\theta)\), and hence
\(\phi(\theta^*)-\phi(\theta_\Delta)\le\eta\).

It remains to account for the runtime. For one fixed price \(\theta\),
Step~\ref{step:unconstrained} computes all \(\widehat s^t(\theta)\) at cost
\(O(|\pazocal T||\pazocal N|\log(1/\varepsilon))\). In
Step~\ref{step:capacity_candidates}, sorting the \(|\pazocal T|\) breakpoints
costs \(O(|\pazocal T|\log|\pazocal T|)\), and checking the
\(O(|\pazocal T|)\) capacity intervals costs \(O(|\pazocal T|\Gamma_C)\).
Thus the cost per grid price is
\(O(|\pazocal T|(|\pazocal N|\log(1/\varepsilon)+\log|\pazocal T|+\Gamma_C))\).
By Step~\ref{step:price_grid}, the grid has \(O(L_\phi\overline\theta/\eta)\)
prices. Thus, the total runtime is
\[
O\!\left(
\frac{\overline\theta}{\eta}
|\pazocal T|^2
\left(\overline C+\frac{\overline\theta|\pazocal N|}{\delta-\gamma}\right)
\left(|\pazocal N|\log(1/\varepsilon)+\log|\pazocal T|+\Gamma_C\right)
\right),
\]
which is polynomial in the stated quantities. The pseudocodes implementing Steps \ref{step:unconstrained}--\ref{step:price_grid}, together with the full
proof of Proposition~\ref{prop:computational_tractability}, are in the
Supplementary Material,~\ref{appendix:computational_tractability}.
\end{proof}
\section{Application to Mobile Edge Computing (MEC)}\label{sec:numerical}
\added{The numerical study models a realistic Mobile Edge Computing deployment in which the infrastructure provider is a network operator investing in costly distributed edge capacity and leasing resources to service providers through take-or-pay contracts under uncertain revenues. It uses calibrated cloud and edge cost data, simulated revenues with seasonal demand and persistent shocks, real mobile-traffic traces, and comparisons with representative benchmarks.} For reproducibility, the code is publicly available at \href{https://github.com/amalsakr-tsp/EJOR}{Code}. All appendices cited below are provided in the Supplementary Material.
\subsection{Settings}
The investment cost includes a fixed technical overhead, capacity-dependent deployment costs due to equipment or physical infrastructure needs, and maintenance costs over the investment horizon~\citep[Eq. (16)--(17)]{liu2025joint}. Since further expansion becomes progressively more expensive due to energy, space, coordination, and operational constraints, we include a convex quadratic term, consistent with~\citep[Eq.~(1)]{kim2025impact}. 
Accordingly,    

\begin{equation}
\mathrm{Cost}(C)
=
F_0+\left(d_{\mathrm{dep}}+\sum_{t\in\pazocal T} d_{\mathrm{maint}}(t)\Delta\right)C
+\frac{d_{\mathrm{exp}}}{2}C^2
\label{eq:cost_parameters}
\end{equation}
Here, \(F_0\) denotes the fixed technical overhead, \(d_{\mathrm{dep}}\) the upfront investment cost per unit of capacity, \(d_{\mathrm{maint}}(t)\) denotes the maintenance cost per unit capacity per unit time during time slot \(t\), and \(\Delta\) is the time slot length. The term \(\frac{d_{\mathrm{exp}}}{2}C^2\) captures increasing marginal capacity-expansion costs. For simplicity, in the experiments, \(d_{\mathrm{maint}}(t)=d_{\mathrm{maint}}\) is constant over time.


For each SP $i \in \pazocal{N}$ and time slot $t \in \pazocal{T}$, the revenue is modeled as in~\eqref{eq:revenue}. The uncertain revenue coefficient associated with SP \(i\), time slot \(t\), and realization \(\omega\) is modeled as \(a_{i,\omega}^t=\mu_i w_{i,d(t)}\left(1+\rho_i\sin\left(2\pi t/T_{\mathrm{season}}\right)\right)\epsilon_{i,\omega}\). Here, \(\mu_i>0\) is the baseline revenue coefficient, \(w_{i,d(t)}\) captures day-of-week effects, \(\rho_i\in[0,1)\) controls the seasonal amplitude, and \(T_{\mathrm{season}}\) is the seasonal-cycle length. Coefficient \(a_{i,\omega}^t\) represents the effective revenue opportunity in slot \(t\), measured in monetary units over one slot.
Shock \(\epsilon_{i,\omega}\sim\mathrm{Lognormal}(-\sigma_i^2/2,\sigma_i^2)\) is positive and satisfies \(\mathbb{E}[\epsilon_{i,\omega}]=1\). Thus, \(\sigma_i\) controls dispersion without changing the mean revenue level. 
\(\epsilon_{i,\omega}\) is specified at the realization level to represent persistent revenue uncertainty. This avoids adding short-term noise at each time slot that would mainly average out over \(I\). The day-of-week term captures systematic weekly demand variation, such as
weekday--weekend differences, while the seasonal term captures longer-term
demand variation across the investment period.



Access payments and SP-side costs are given by~\eqref{eq:price} and~\eqref{eq:psi}. In each experiment, \(\delta=1.01\,\gamma\max_{i\in\pazocal N}\xi_i\), computed using the largest \(\max_{i\in\pazocal N}\xi_i\) over the risk classes considered so that the sufficient uniqueness condition~\eqref{eq:delta_gamma_uniqueness} is satisfied. Sensitivity with respect to \(\delta\) and \(\gamma\) is examined in \S \ref{sec:sens}. 

Table~\ref{tab:model_parameters} summarizes the parameters, and \ref{appendix:equilibrium-computation} reports the runtime. 
The experiments use Monte Carlo simulations. Since the InP problem~\eqref{eq:leader_problem} may have multiple maximizers, we report the capacity--price pair \((C^*,\theta^*)\) selected by the procedure (\S \ref{sec:proc}), where \(\theta^*=\theta_\Delta\) denotes the best price on the grid. In our runs, the procedure returned a single exact computed maximizer, with only negligible-gap near-optimal alternatives.
\begin{table}[h!]
\centering
\caption{Model parameters used in the numerical experiments.}
\label{tab:model_parameters}
\resizebox{0.99\columnwidth}{!}{
\begin{tabular}{lll}
\hline
Parameter & Symbol & Value \\
\hline
Number of SPs & \(N\)~(\S~\ref{sec:overview}) & \(3\) and \(10\) \\
Investment horizon & \(I\)~(\S~\ref{sec:overview}) & \(5\) years \\
Time slot length & \(\Delta\) & \(1\) hour\\
Number of time slots & \(|\pazocal T|\) & \(5\times365\times24=43{,}800\) \\
Low risk-aversion threshold & \(\beta_L\)~(\S~\ref{sec:class}) & \(1\) \\
High risk-aversion threshold & \(\beta_H\)~(\S~\ref{sec:class}) & \(2\) \\
Extreme risk-aversion threshold & \(\beta_E\)~(\S~\ref{sec:class}) & \(10\) \\
Congestion cost parameter & \(\gamma\)~\eqref{eq:psi} & \(0.01\) \$/vCore\(^2\)-hour \\
Operational cost parameter & \(\delta\)~\eqref{eq:psi} & $0.0365$\$/vCore\(^2\)-hour for $N=3$; $0.175$ for $N=10$
\\
Fixed technical cost & \(F_0\)~\eqref{eq:cost_parameters} & \(2000\) \$ \\
Deployment cost per capacity vCore & \(d_{\mathrm{dep}}\)~\eqref{eq:cost_parameters} & \(10\) \$/vCore~\citep{azureStackEdgePricing}\\
Maintenance cost per capacity vCore per hour & $d_{\mathrm{maint}}$~\eqref{eq:cost_parameters} & \(0.0225\) \$/vCore-hour~\citep{azureStackEdgePricing} \\
Quadratic capacity cost coefficient & $d_{\mathrm{exp}}$~\eqref{eq:cost_parameters} & \(0.02\) \$/vCore\(^2\) \\
\hline
\end{tabular}
}
\end{table}
\vspace{-4mm}
\subsection{Classification of SP Risk Attitudes}\label{sec:class}
SP risk preferences are described by \((\beta_i,\alpha_i)\), and we consider the following representative risk-attitude classes:
\begin{itemize}
    \item \textbf{Risk-neutral (RN):} \(\beta_i=0\).
    \item \textbf{Moderately risk-averse (MRA):} \(0<\beta_i\leq\beta_L,\ 0.9<\alpha_i\leq0.95\).
    \item \textbf{Highly risk-averse (HRA):} \(\beta_L<\beta_i\leq\beta_H,\ 0.9<\alpha_i\leq0.95\).
    \item \textbf{Extremely risk-averse (ERA):} \(\beta_i>\beta_{E},\ \alpha_i>0.95\).
\end{itemize}
We focus on \(\alpha_i>0.9\), since smaller values make CVaR less focused on downside risk; \(\alpha_i>0.95\) indicates more severe tail-risk exposure.
\subsection*{We now evaluate the proposed framework.} 
\subsection{Impact of Homogeneous SP Risk Preferences}
We first consider three SPs with homogeneous risk preferences: all SPs belong to the same risk class in \S~\ref{sec:class}, although their \((\beta_i,\alpha_i)\) values may differ. The three SPs differ in their revenue levels, with SP~1, SP~2, and SP~3 corresponding to low-, medium-, and high-revenue SPs, respectively.
This allows us to isolate the effect of a common risk attitude while preserving revenue heterogeneity. Revenue uncertainty is measured by the coefficient of variation (CV) of \(a_{i,\omega}^t\), i.e., the ratio of its standard deviation to its mean. 

\noindent\textbf{Impact on SP resource commitments.}
Fig.~\ref{fig:average_demand_followers} reports the average total resource commitment across SPs for different risk classes and CVs. Resource commitments decrease as SPs become more risk-averse. This is because more risk-averse SPs assign greater weight to unfavorable revenue realizations through CVaR and therefore reduce resource commitments to limit downside exposure. The decrease is stronger for larger \(CV\), since higher uncertainty makes low-revenue outcomes more severe, while access payments and SP-side costs still have to be paid.

\begin{figure*}[!ht]
\centering
\begin{minipage}[t]{0.44\textwidth}
\centering
\includegraphics[width=\linewidth]{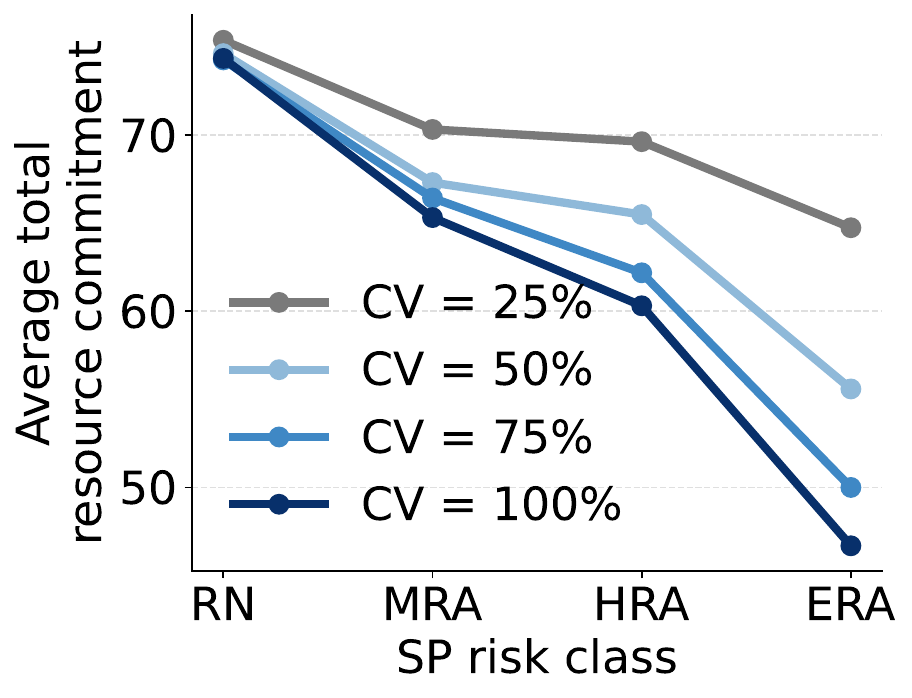}
\caption{Total resource commitment under different risk classes and~CVs.}
\label{fig:average_demand_followers}
\end{minipage}
\hfill
\begin{minipage}[t]{0.44\textwidth}
\centering
\includegraphics[width=\linewidth]{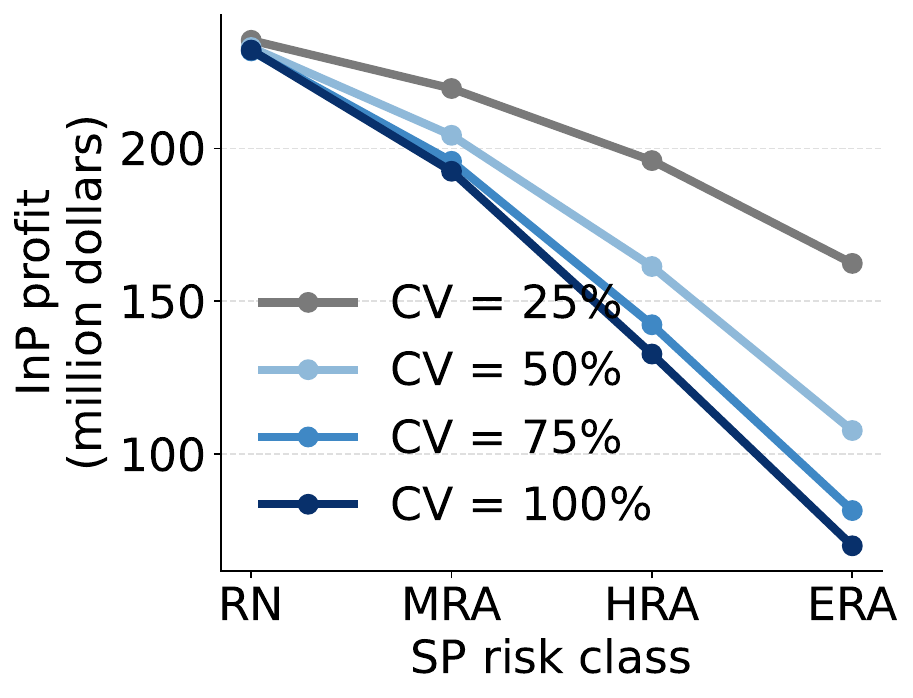}
\caption{InP profit under different risk classes and~CVs.}
\label{fig:leader_profit_all_CV}
\end{minipage}

\vspace{4mm}
\begin{minipage}[t]{0.44\textwidth}
\centering
\includegraphics[width=\linewidth]{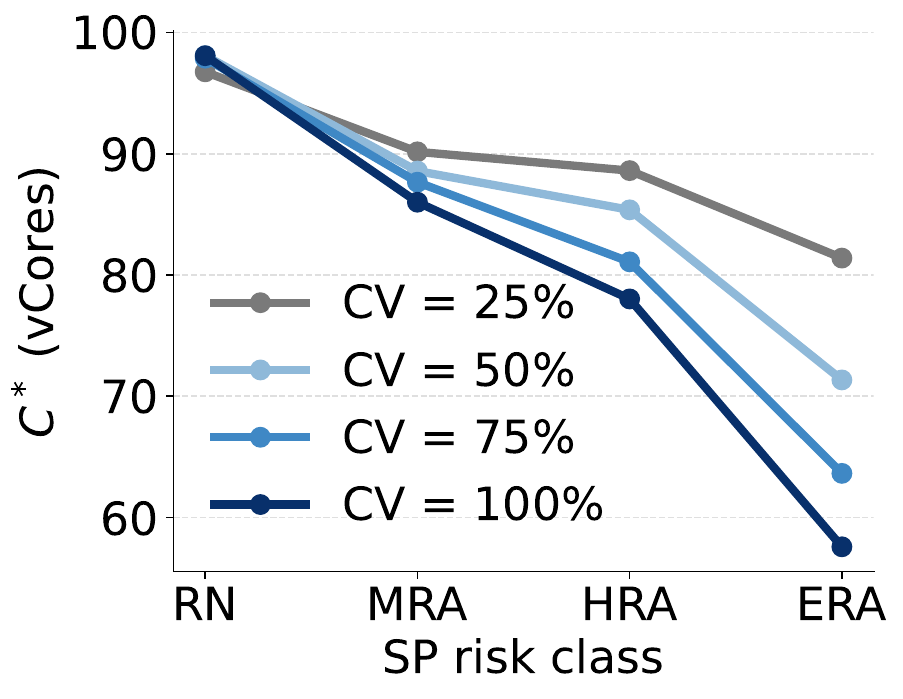}
\caption{Optimal capacity \(C^*\) under different risk classes and~CVs.}
\label{fig:capacity_all_CV}
\end{minipage}
\hfill
\begin{minipage}[t]{0.44\textwidth}
\centering
\includegraphics[width=\linewidth]{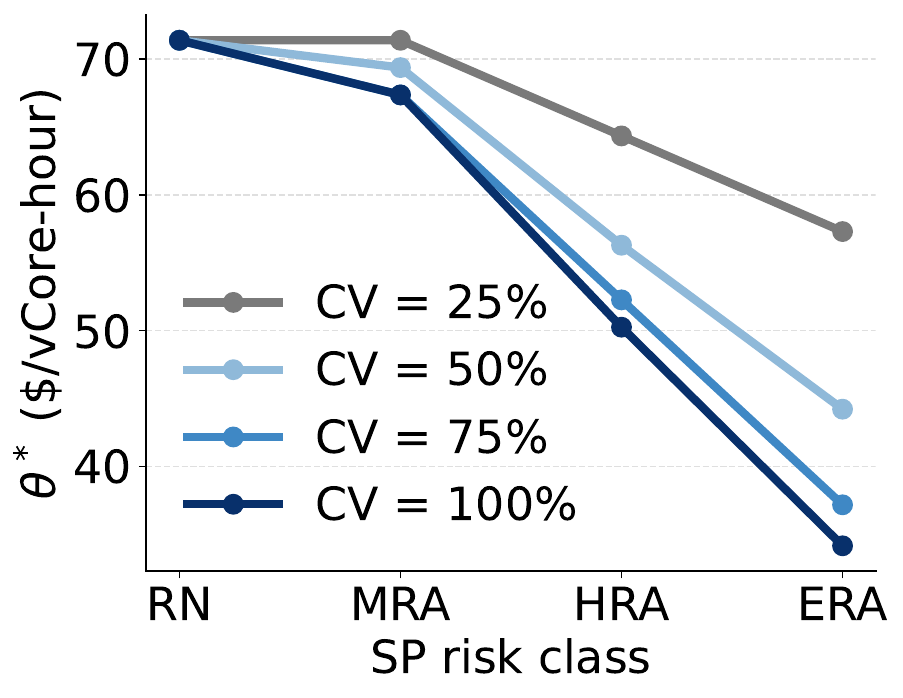}
\caption{Optimal price \(\theta^*\) under different risk classes and~CVs.}
\label{fig:price_all_CV}
\end{minipage}
\end{figure*}
\noindent
\textbf{Impact on the InP.}
Fig.~\ref{fig:leader_profit_all_CV} shows that InP profit decreases with SP risk aversion. Since the InP's revenue is proportional to allocated resources, lower SP resource commitments in Fig.~\ref{fig:average_demand_followers} directly reduce profit. This effect is stronger under higher \(CV\), as greater uncertainty leads SPs to commit to resources more conservatively.
Fig.~\ref{fig:capacity_all_CV} shows that the optimal capacity \(C^*\) decreases as SPs become more risk-averse. This follows the same trend as the total SP resource commitment in Fig.~\ref{fig:average_demand_followers}: when the InP anticipates lower resource commitment, it reduces capacity investment to avoid costly over-provisioning. The decline is more pronounced under higher \(CV\), indicating that SP revenue uncertainty propagates upstream from SP resource decisions to the InP's long-term capacity choice.
Comparing Fig.~\ref{fig:capacity_all_CV} with Fig.~\ref{fig:average_demand_followers} shows that capacity utilization remains high and slightly increases with risk aversion, from about \(76\%\) to above \(81\%\). Fig.~\ref{fig:price_all_CV} shows the optimal price \(\theta^*\). Pricing balances access-revenue extraction against SP resource commitment. As SPs become more risk-averse, their resource commitments become more price-sensitive, so raising the price would reduce resource commitments too strongly. The InP therefore adjusts price jointly with capacity to sustain participation while extracting revenue from the resources.

\begin{figure*}[t]
    \centering

    \begin{minipage}{0.46\textwidth}
        \centering
        \includegraphics[width=\linewidth]{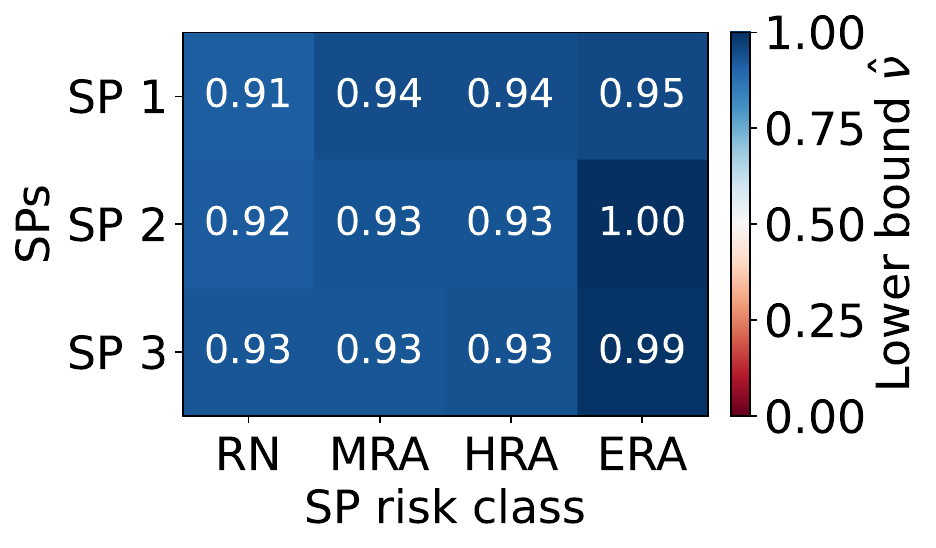}
        \subcaption{$CV=25\%$}
    \end{minipage}
    \hfill
    \begin{minipage}{0.46\textwidth}
        \centering
        \includegraphics[width=\linewidth]{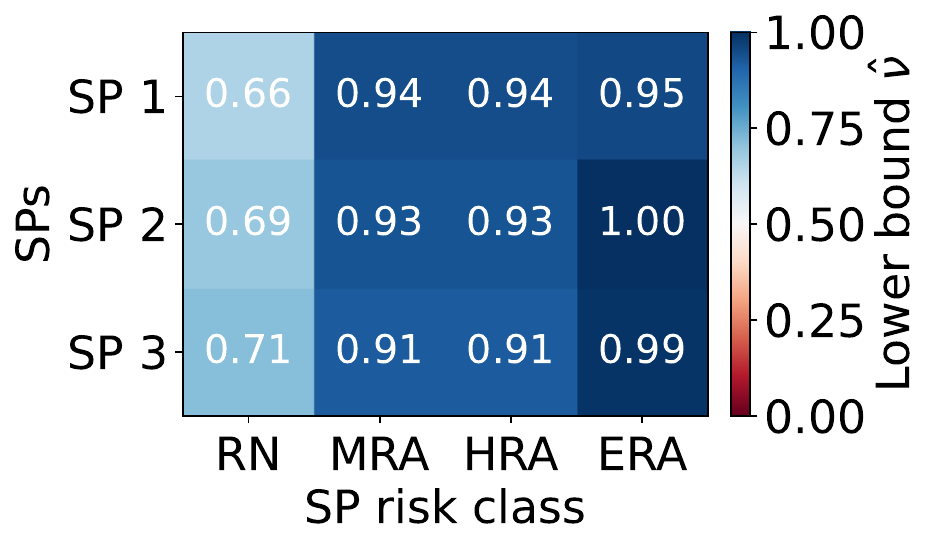}
        \subcaption{$CV=50\%$}
    \end{minipage}

    \vspace{0.4em}

    \begin{minipage}{0.46\textwidth}
        \centering
        \includegraphics[width=\linewidth]{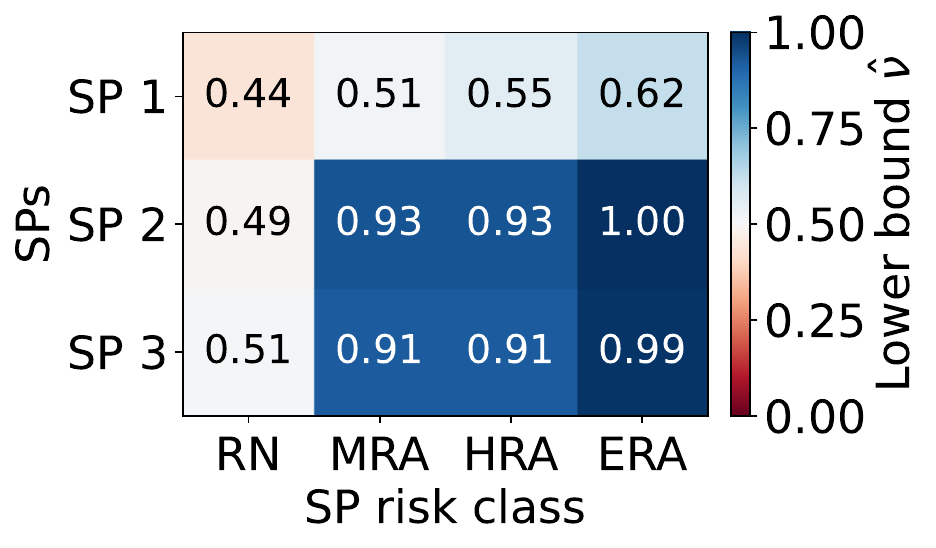}
        \subcaption{$CV=75\%$}
    \end{minipage}
    \hfill
    \begin{minipage}{0.46\textwidth}
        \centering
        \includegraphics[width=\linewidth]{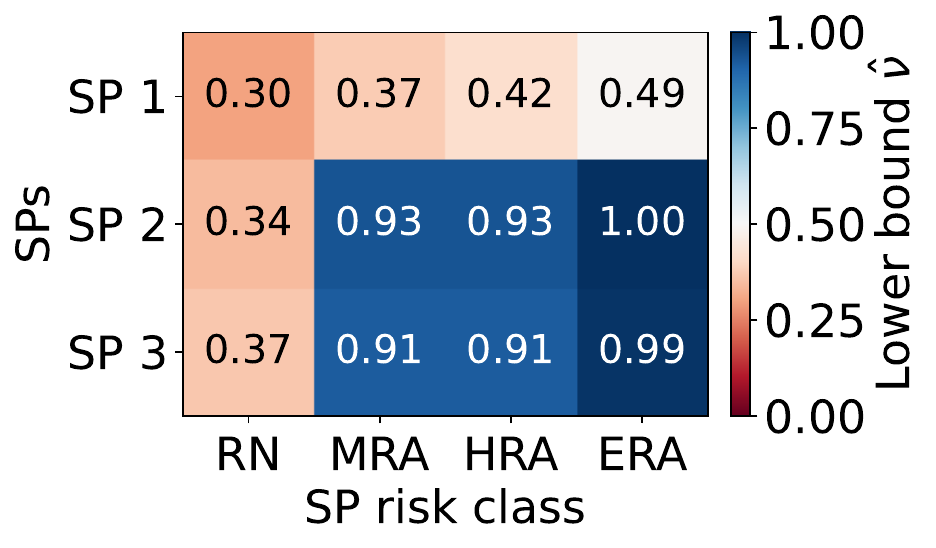}
        \subcaption{$CV=100\%$}
    \end{minipage}

    \caption{Lower bound on the PoP $\hat{\nu}_i$ under different SP risk classes and CVs.}
    \label{fig:nu_hat_cv}
\end{figure*}
\noindent\textbf{Impact on SP profit.}
Fig.~\ref{fig:nu_hat_cv} reports the theoretical lower bound~$\hat{\nu}_i$ on the PoP. Moving from risk-neutral to more risk-averse behavior increases the lower bound, because risk-averse SPs reduce their exposure to access payments and SP-side costs in adverse revenue scenarios. The effect of $CV$ is heterogeneous across SPs. At low and moderate $CV$, the small SP, SP~1, achieves a high lower bound, sometimes comparable to or slightly higher than the medium and large SPs, SP~2 and SP~3, because its smaller resource commitment limits cost exposure. However, as $CV$ becomes large, SP~1 is the most affected, since its limited revenue buffer makes severe adverse realizations harder to absorb. For SP~2 and SP~3, the lower bound is sensitive to $CV$ mainly under risk neutrality, whereas once they become risk-averse it remains almost unchanged and close to one across all $CV$ levels. Overall, we show that CVaR-based risk aversion strengthens the PoP lower bound, while high revenue uncertainty primarily hurts SP 1.
\ref{app:pop_gap} shows that the PoP lower bound $\hat{\nu}_i$ is reasonably tight. The gap between the empirical PoP and the lower bound increases under higher revenue uncertainty, especially for risk-neutral SPs, and decreases as SPs become more risk-averse.

\subsection{Impact of Individual SP Risk Preferences}
We vary the risk preference of one SP at a time while keeping the other two fixed as MRA. This isolates the effect of the small, medium, and large SPs on the SE and the profit. We fix $\delta = 0.0695$ and $CV=100\%$.

\noindent\textbf{Impact on the InP.}
Fig.~\ref{fig:leader_profit_changing_sp} shows the InP profit. All curves meet at MRA, which corresponds to the baseline where all SPs are MRA. When the changing SP becomes more risk-averse, InP profit decreases because that SP commits to fewer resources, reducing the InP's access revenue. The size of the decrease depends on the SP: changing SP~1 has the smallest effect, while changing SP~3 has the strongest effect, since SP~3 is the largest resource commitment contributor. Thus, the InP is affected not only by the level of risk aversion, but also by which SP becomes more risk-averse.

\noindent\textbf{Impact on SP profit.}
Figs.~\ref{fig:prob_follower_1_changes}--\ref{fig:prob_follower_3_changes}
show the lower bound on the PoP. The change mainly
affects the lower bound of the SP whose risk class varies, while the other SPs
are only weakly affected. This indicates that risk aversion acts primarily as a
self-protection mechanism; the SP that becomes more risk-averse reduces its own
resource exposure and improves its own profitability guarantee.  Cross-effects appear only for SP~1,
whose smaller revenue buffer makes it more sensitive to other SPs' changes. 
\begin{figure*}[!ht]
\centering
\begin{minipage}[t]{0.44\textwidth}
\centering
\includegraphics[width=\linewidth]{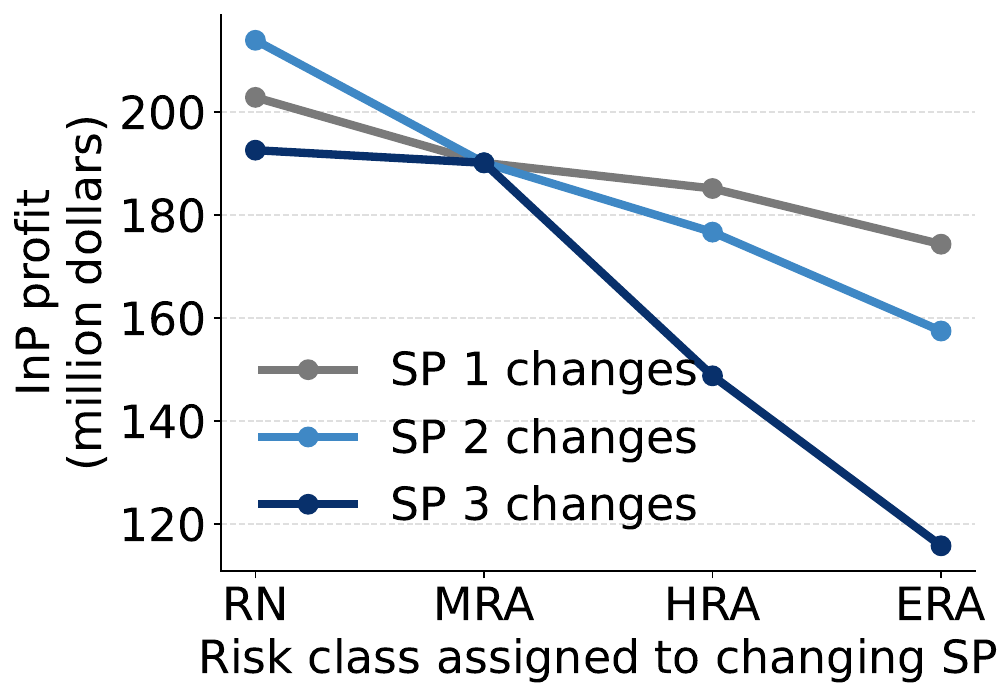}
\caption{InP profit when one SP changes risk class at \(CV=100\%\).}
\label{fig:leader_profit_changing_sp}
\end{minipage}
\hfill
\begin{minipage}[t]{0.46\textwidth}
\centering
\includegraphics[width=\linewidth]{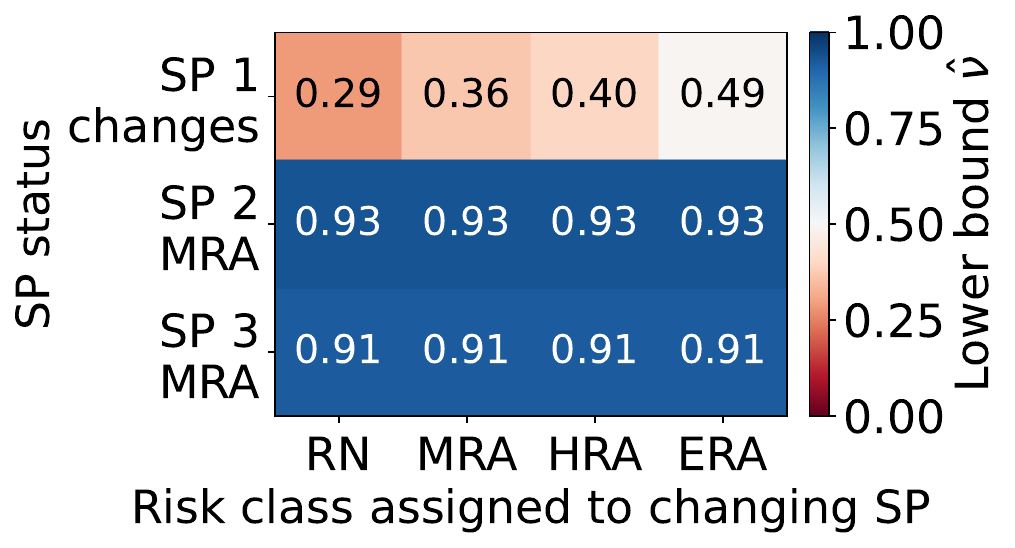}
\caption{Lower bound when SP 1 changes risk class at $CV=100\%$.}
\label{fig:prob_follower_1_changes}
\end{minipage}
\vspace{4mm}
\begin{minipage}[t]{0.46\textwidth}
\centering
\includegraphics[width=\linewidth]{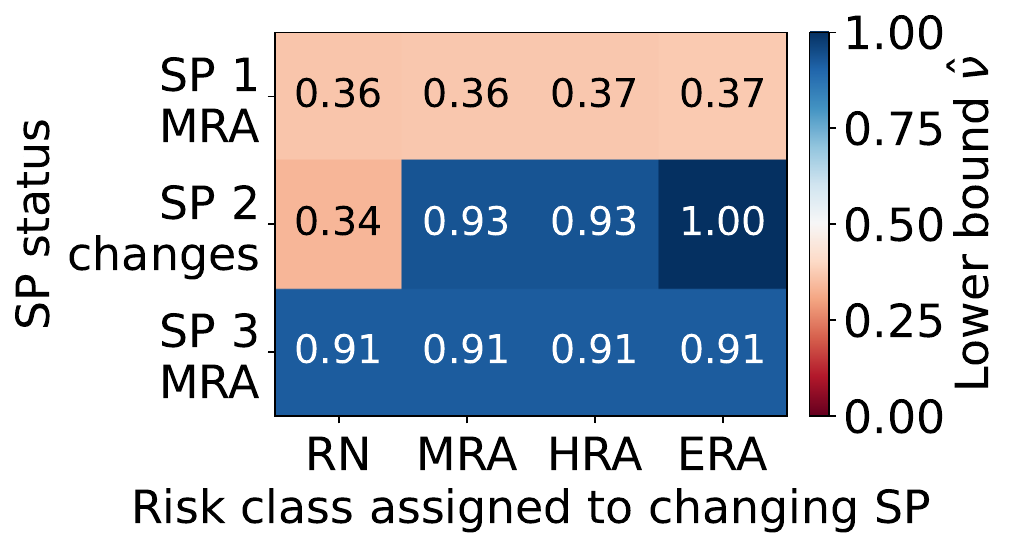}
\caption{Lower bound when SP 2 changes risk class at $CV=100\%$.}
\label{fig:prob_follower_2_changes}
\end{minipage}
\hfill
\begin{minipage}[t]{0.46\textwidth}
\centering
\includegraphics[width=\linewidth]{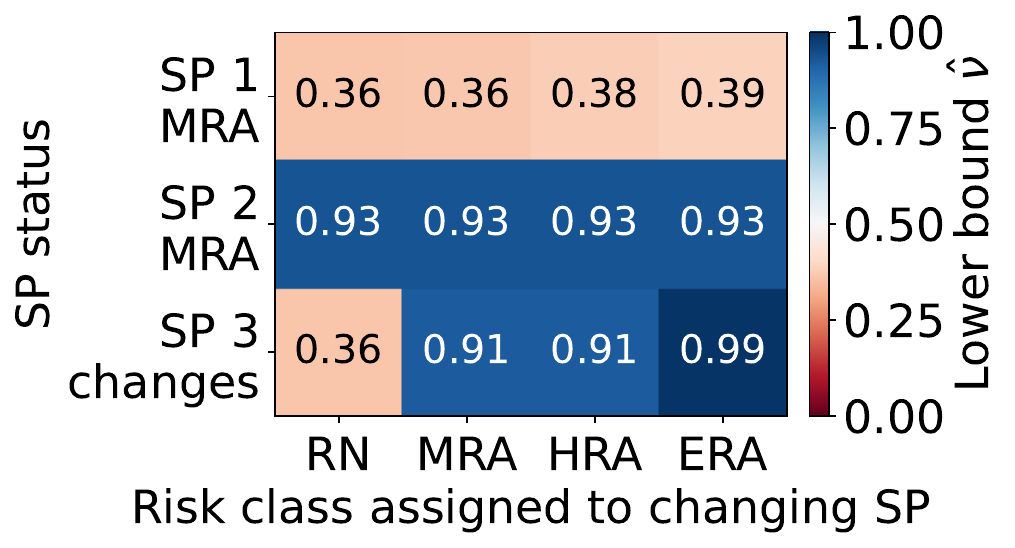}
\caption{Lower bound when SP 3 changes risk class at $CV=100\%$.}
\label{fig:prob_follower_3_changes}
\end{minipage}
\end{figure*}
\subsection{Scalability to Ten SPs}
We next consider ten SPs with different revenues and risk classes.

\noindent\textbf{Impact on the InP.}
Fig.~\ref{fig:scalability_inp} compares the InP outcomes with 3 and 10 SPs. With 10 SPs, InP profit, capacity, and price are higher because more SPs share the infrastructure and generate larger aggregate resource commitments. This makes additional capacity more valuable for the InP and allows a higher access price while still inducing SPs to commit to resources. Importantly, scaling from 3 to 10 SPs changes the magnitude of the equilibrium outcomes, but the qualitative behavior remains the same.
\begin{figure*}[!ht]
\centering
\begin{minipage}[t]{0.32\textwidth}
\centering
\includegraphics[width=\linewidth]{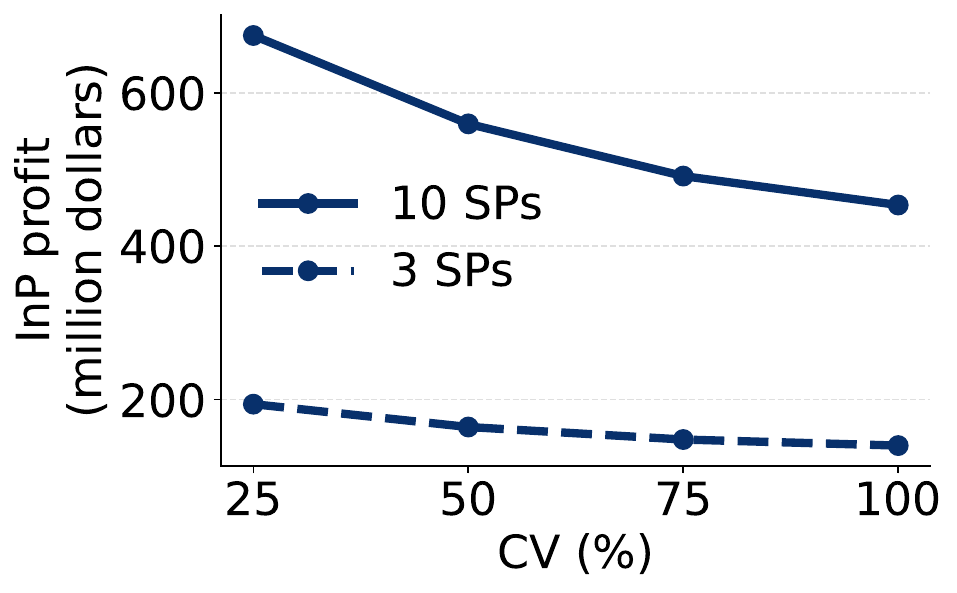}
{\small (a) InP profit}
\end{minipage}
\hfill
\begin{minipage}[t]{0.32\textwidth}
\centering
\includegraphics[width=\linewidth]{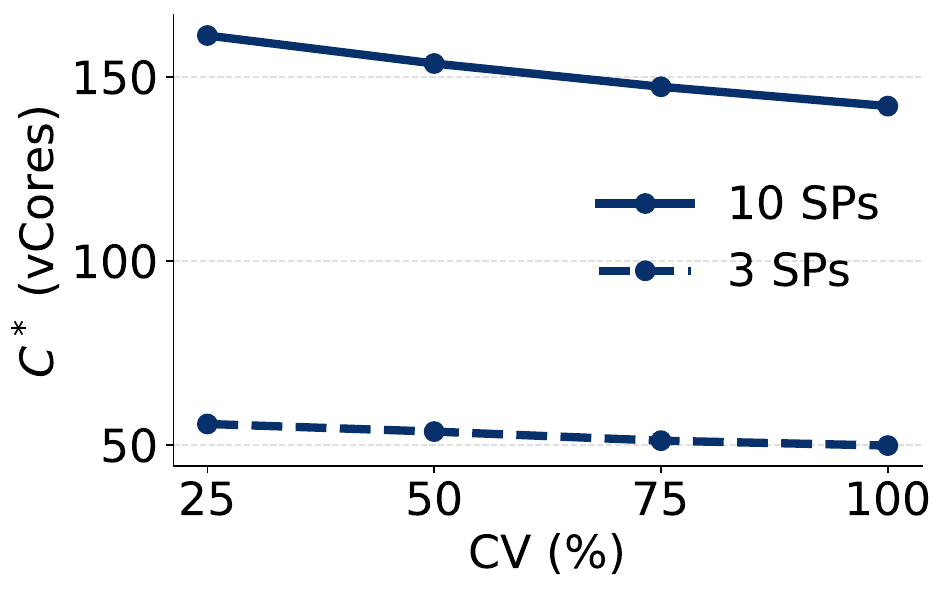}
{\small (b) Capacity \(C^*\)}
\end{minipage}
\hfill
\begin{minipage}[t]{0.32\textwidth}
\centering
\includegraphics[width=\linewidth]{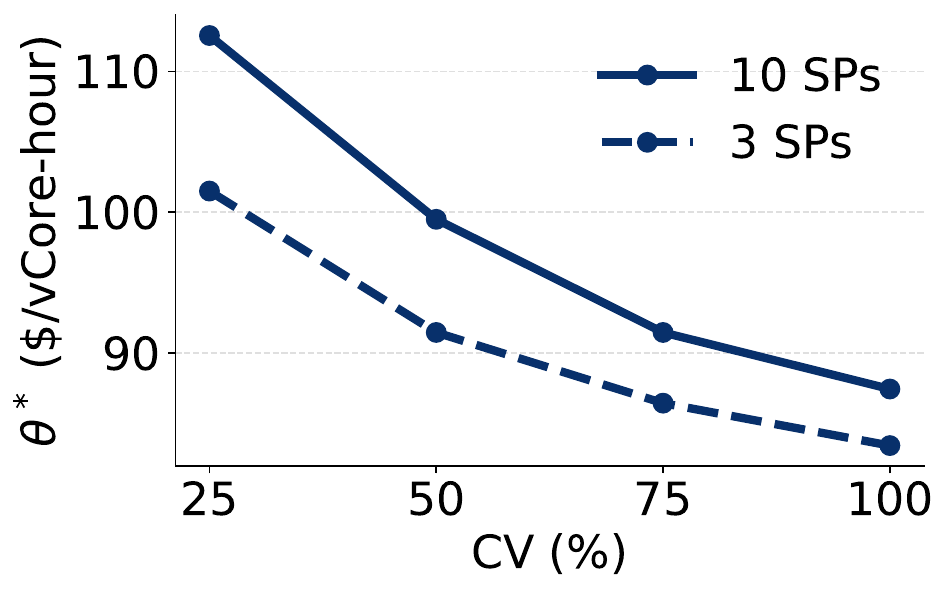}
{\small (c) Price \(\theta^*\)}
\end{minipage}
\caption{InP profit and capacity--pricing decisions for 3 and 10 SPs under varying \(CV\).}
\label{fig:scalability_inp}
\end{figure*}

\noindent\textbf{Impact on SP profit.}
Fig.~\ref{fig:nu_hat_scalability} compares the lower bound for the 3-SP and 10-SP cases. In the 3-SP case, SP~2 and SP~3 keep high and almost unchanged bounds across $CV$, while SP~1 is affected at high uncertainty because its revenue buffer is smaller. The 10-SP case shows the same behavior at larger scale; well-protected SPs, such as SP~2, SP~3, and especially the ERA SP~9, keep high bounds, whereas vulnerable SPs deteriorate more as $CV$ increases. The RN SP~5 is the most affected because it does not reduce risk exposure through CVaR. Thus, the lower bound is mainly determined by each SP's own risk class and revenue size and not by the number of SPs.
\begin{figure}[!t]
\centering

\begin{minipage}[t]{0.37\columnwidth}
\centering
\includegraphics[width=\linewidth]{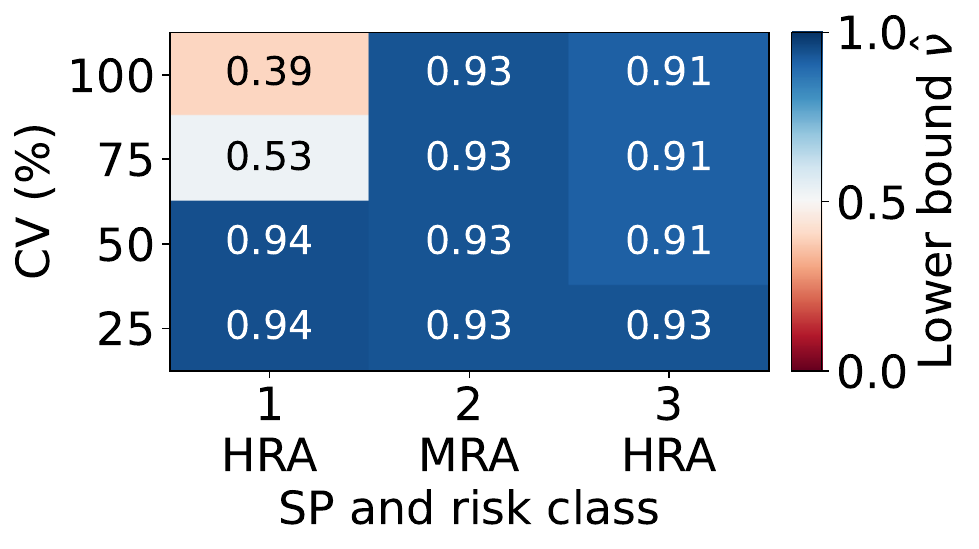}
{\small (a) 3-SP case}
\end{minipage}
\hfill
\begin{minipage}[t]{0.62\columnwidth}
\centering
\includegraphics[width=\linewidth]{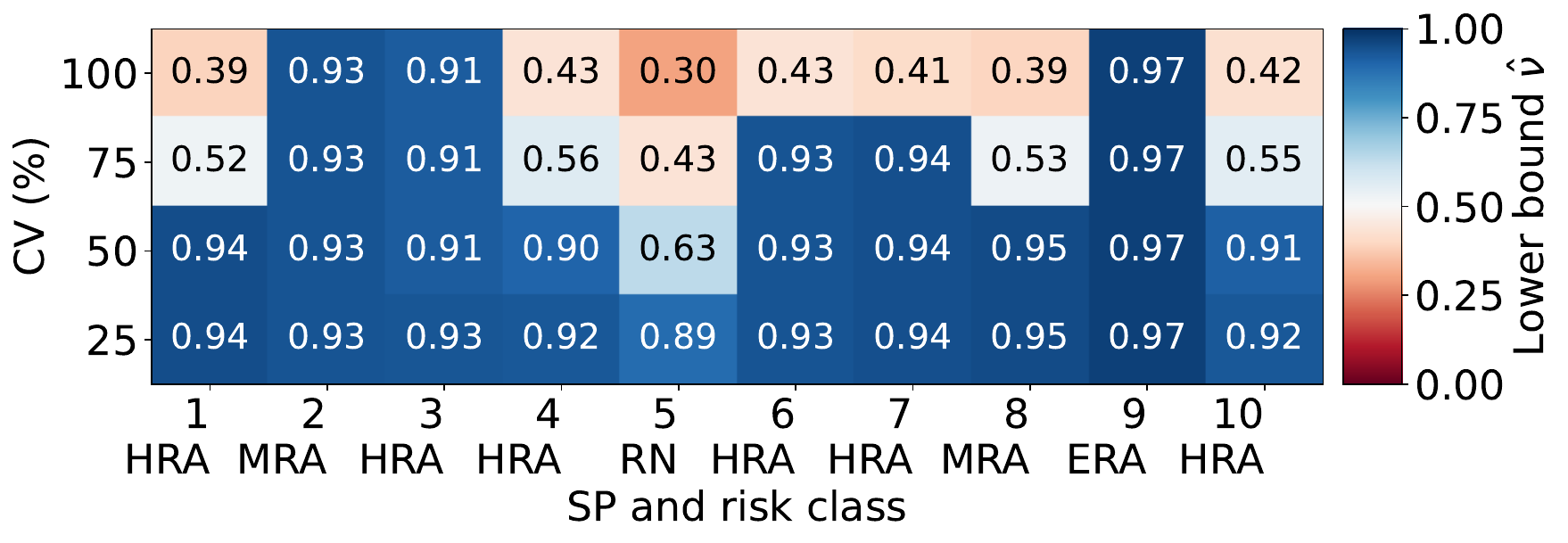}
{\small (b) 10-SP case}
\end{minipage}

\caption{Lower bound on the PoP for 3 and 10 SPs under heterogeneous risk classes and varying $CV$.}
\label{fig:nu_hat_scalability}
\end{figure}
%
%
%
%
\vspace{-8mm}
\subsection{Sensitivity to SP-side and Investment Costs}\label{sec:sens}
\ref{appendix:sensitivity} shows that higher SP-side costs lower capacity and InP profit, raise the access price, and affect SP profit heterogeneously, whereas investment-cost changes have only limited impact.

\subsection{Validation with Real-World Dataset}
\ref{app:dataset_aiwt} reports a validation experiment using real
mobile-traffic traces mapped to five SPs with heterogeneous revenues
and risk preferences. The SE is profitable for the InP
($C^*=109$ vCores, $\theta^*=27$~\$/vCore-hour, profit $=\$1.37$M),
and all SPs achieve high lower bounds on the PoP ($0.94$--$1.00$), supporting the model under realistic revenue variability.

\subsection{Benchmark Comparison}
~\ref{app:benchmark_comparison} compares the proposed model with representative state-of-the-art benchmarks. Since existing approaches capture only subsets of our model, we treat them as partial formulations of the broader setting addressed in this paper. We thus let each benchmark determine the decisions it models, while the remaining decisions are fixed according to our method. The resulting outcomes are qualitatively aligned with the full model, indicating that our formulation subsumes the main economic mechanisms of existing approaches while endogenizing infrastructure investment, pricing, SP commitments, risk exposure, and congestion.
The proposed model achieves high InP profit, high utilization, and a strong SP-PoP lower bound.

\section{Conclusion}\label{sec:conclusion}
This paper has introduced a Stackelberg game for shared-infrastructure investment and pricing under uncertainty. The model endogenizes future utilization through the resource commitments of heterogeneous risk-aware firms facing shared capacity, congestion, and downside risk. We have proved the existence of a Stackelberg equilibrium and derived lower bounds on the PoP. The MEC results show that higher firm risk aversion reduces InP profit while improving the SP-PoP lower bound. Future work may consider multiple InPs and dynamic capacity and pricing expansion.

\paragraph{\textbf{Acknowledgment}}
This work was supported by the ANR under the France 2030 program, grant NF-NAI: ANR-22-PEFT-0003.





\clearpage
\appendix

\numberwithin{theorem}{section}

\renewcommand{\thetheorem}{\Alph{section}\arabic{theorem}}
\renewcommand{\theproposition}{\Alph{section}\arabic{theorem}}
\renewcommand{\thelemma}{\Alph{section}\arabic{theorem}}
\renewcommand{\thedefinition}{\Alph{section}\arabic{theorem}}
\renewcommand{\theassumption}{\Alph{section}\arabic{theorem}}
\renewcommand{\theremark}{\Alph{section}\arabic{theorem}}
\renewcommand{\thecorollary}{\Alph{section}\arabic{theorem}}


\begin{center} {\Large\bfseries Supplementary Material}\\[0.5em] {\large\bfseries Shared Infrastructure Investment and Pricing:}\\ {\large\bfseries Stackelberg Equilibria in Risk-Aware Take-or-Pay Contracts} \end{center}

\section{Notation}
\label{appendix:notation}
Table~\ref{tab:main_notation} summarizes the main notation.
\begin{table}[!t]
\centering
\caption{Notation}
\label{tab:main_notation}
\resizebox{0.95\columnwidth}{!}{
\begin{tabular}{p{0.24\columnwidth} p{0.74\columnwidth}}
\toprule
\textbf{Notation} & \textbf{Description} \\
\midrule
$\pazocal N=\{1,\dots,N\}$ & Set of firms (\S\ref{sec:overview}). \\
$\pazocal T=\{1,\dots,T\}$ & Set of time slots (\S\ref{sec:overview}). \\
$C$ & Infrastructure capacity chosen by the InP (\S\ref{sec:overview}). \\
$\theta$ & Access price chosen by the InP (\S\ref{sec:overview}). \\
$h_i^t$ & Resource commitment of firm $i$ at time slot $t$ (\S\ref{sec:overview}). \\
$\mathbf h_i$ & Resource commitment vector of firm $i$ over all time slots \eqref{eq:u_follower}. \\
$\mathbf h_{-i}$ & Resource commitment vectors of all firms except \(i\) (\S\ref{sec:stackelberg_game}). \\
$\mathbf h^t$ & Resource commitment vector of all firms at time slot $t$ (\S\ref{sec:stackelberg_game}). \\
$\Pi_i^t(\mathbf h^t;\theta)$ & Profit of firm $i$ at time slot $t$ under access price \(\theta\) \eqref{eq:profit_sp_realization}. \\
$L_i^t(\mathbf h^t;\theta)$ & Loss of firm $i$ at time slot $t$ under access price \(\theta\) \eqref{eq:loss_i_t}. \\
$\beta_i$ & Risk-aversion coefficient of firm $i$ \eqref{eq:u_follower}. \\
$\alpha_i$ & CVaR confidence level of firm $i$ \eqref{eq:u_follower}. \\
$\mathrm{CVaR}_{\alpha_i}$ & Conditional Value-at-Risk at level $\alpha_i$ \eqref{eq:u_follower}. \\
$U_i(\mathbf h_i,\mathbf h_{-i};\theta)$ & Mean-CVaR utility of follower $i$ \eqref{eq:u_follower}. \\
$(C^*,\theta^*,\mathbf{H}^*)$ & Stackelberg equilibrium (Definition~\ref{def:stack}). \\
$\mathbb R_+$ & Set of nonnegative real numbers (\S\ref{sec:stackelberg_game}). \\
$\omega$ & Revenue-uncertainty realization, with $\omega\in\Omega$ (\S\ref{sec:follower_problem}). \\
$r_{i,\omega}^t(h_i^t)$ & Realized revenue of firm $i$ in slot $t$ under realization $\omega$ \eqref{eq:profit_sp_realization}; specified in \eqref{eq:revenue}. \\
$p_\theta(h_i^t)$ & Access payment paid by firm $i$ \eqref{eq:profit_sp_realization}; specified in \eqref{eq:price}. \\
$\psi_i^t(\mathbf h^t)$ & firm-side cost of firm $i$ in slot $t$ \eqref{eq:profit_sp_realization}; specified in \eqref{eq:psi}. \\
$\pazocal Y$ & Feasible set of the InP's capacity--price decisions \eqref{eq:leader_set}. \\
$\mathrm{Cost}(C)$ & Investment cost of capacity $C$ \eqref{eq:leader_problem}; specified in \eqref{eq:cost_parameters}. \\
$\mathbf H^*(C,\theta)$ & Follower equilibrium response induced by leader decision $(C,\theta)$ \eqref{eq:leader_problem}. \\
$\Pi_{i,\omega}^*, R_{i,\omega}^*, K_i^*$ & Equilibrium realized profit, realized total revenue, and deterministic total cost of firm $i$ \eqref{eq:pi_star}. \\
$\hat{\nu}_i$ & Lower bound on the Probability of Profit \eqref{eq:profit_bound_final}. \\
$\overline C,\overline\theta$ & Maximum admissible capacity and access price \eqref{eq:leader_set}. \\
$a_{i,\omega}^t$ & Revenue coefficient of firm $i$ in time slot $t$ under realization $\omega$ \eqref{eq:revenue}. \\
$\delta$ & Operational cost parameter \eqref{eq:psi}. \\
$\gamma$ & Congestion cost parameter \eqref{eq:psi}. \\
$A_i^t$ & Mean-CVaR benefit factor of firm $i$ at $t$ (Prop.~\ref{prop:follower_explicit}). \\
$\lambda^{t*}$ & Shadow price of the shared capacity constraint at $t$ (Prop.~\ref{prop:follower_explicit}). \\
\bottomrule
\end{tabular}}
\end{table}

\section{Step-by-Step Proof of Proposition~\ref{prop:ve_gne_existence}}\label{appendix:ve_gne_existence}
For fixed $(C,\theta)$, define the joint feasible set
\begin{equation}
X(C)
:=
\left\{
\mathbf{H}\in\mathbb R_+^{NT}
\;\middle|\;
\sum_{j\in\pazocal N} h_j^t \le C,
\quad \forall t\in\pazocal T
\right\}
\label{eq:x_joint_feasible}
\end{equation}
where \(\mathbf{H}:=(h_i^t)_{i\in\pazocal N,\ t\in\pazocal T}\) denotes the resource vectors of all followers. 
For the VE, it is convenient to work with the disutility. For each \(i\in\pazocal N\), define
\begin{align}
P_i(\mathbf h_i, \mathbf h_{-i};\theta):&\overset{\eqref{eq:u_follower}}{=}\sum_{t\in\pazocal T}
\left(
-\mathbb E_\omega[\Pi_i^t(\mathbf h^t;\theta)]
+
\beta_i\,\mathrm{CVaR}_{\alpha_i}(L_i^t(\mathbf h^t;\theta))
\right)
\label{eq:cost_follower_j}
\end{align}
where \(\mathbf h_i=(h_i^t)_{t\in\pazocal T}\in\mathbb R_+^T\) is the resource commitment vector of follower \(i\).
 
For compactness, write
\(
P_i(\mathbf H;\theta)
:=
P_i(\mathbf h_i,\mathbf h_{-i};\theta)
\),
where \(\mathbf H=(\mathbf h_1,\dots,\mathbf h_N)\).
We first define the notion of the variational equilibrium of the follower game.
\begin{definition}
For fixed $(C,\theta)$, a vector $\mathbf{H}^*\in X(C)$~\eqref{eq:x_joint_feasible} is called a variational equilibrium (VE) of the follower game if it solves the variational inequality
$\mathrm{VI}(X(C),\mathcal G(\cdot;\theta))$\citep{facchinei2007generalized}, namely,
\begin{equation}
\left\langle
\mathcal G(\mathbf{H}^*;\theta),\,
\mathbf{H}-\mathbf{H}^*
\right\rangle
\ge 0,
\qquad
\forall \mathbf{H}\in X(C)
\label{eq:VI}
\end{equation}
where 
\begin{align}
\mathcal G(\mathbf{H};\theta)
&:=
\begin{pmatrix}
\nabla_{\mathbf h_1}P_1(\mathbf{H};\theta)\\
\vdots\\
\nabla_{\mathbf h_N}P_N(\mathbf{H};\theta)
\end{pmatrix}=
\begin{pmatrix}
\left(\dfrac{\partial P_1}{\partial h_1^t}(\mathbf{H};\theta)\right)_{t\in\pazocal T}\\
\vdots\\
\left(\dfrac{\partial P_N}{\partial h_N^t}(\mathbf{H};\theta)\right)_{t\in\pazocal T}
\end{pmatrix},
\label{eq:pseudo_grad}
\end{align}
is the pseudo-gradient mapping of the follower game.
\end{definition}
The variational inequality~\eqref{eq:VI} states that, at equilibrium \(\mathbf{H}^*\), no follower can improve its mean-CVaR, i.e., \(U_i(\mathbf h_i,\mathbf h_{-i};\theta)\)~\eqref{eq:u_follower}, by changing its resource commitment, given the resource commitments of the other followers and the shared capacity constraint.

We decompose the proof into the following three lemmas.

\begin{lemma}
\label{lem:F_cont}
For every fixed \(C\in[0,\overline C]\), pseudo-gradient mapping
\(
(\mathbf H,\theta)\mapsto \mathcal G(\mathbf H;\theta)
\)
defined in~\eqref{eq:pseudo_grad} is continuous on
\(X(C)\times[0,\overline \theta]\). 
\end{lemma}

\begin{proof}
Fix \(C\in[0,\overline C]\). We prove the continuity of
\(
(\mathbf H,\theta)\mapsto \mathcal G(\mathbf H;\theta)
\)
on \(X(C)\times[0,\overline \theta]\). The capacity \(C\) only determines the
feasible set \(X(C)\); it does not enter the expression of the follower
disutilities directly~\eqref{eq:cost_follower_j}. Hence, \(C\) does not appear in the
partial derivatives that define the pseudo-gradient mapping
\(\mathcal G\) in~\eqref{eq:pseudo_grad}.

For each follower \(i\in \pazocal N\), the function
\(P_i(\mathbf h_i,\mathbf h_{-i};\theta)\)~\eqref{eq:cost_follower_j}
is defined as a finite sum over \(t\in\pazocal T\) of terms involving the
expected profit and the CVaR penalty. For every \(\mathbf H\in X(C)\) and
every \(t\in\pazocal T\), the associated time-slot vector \(\mathbf h^t\)
belongs to \(K\subset\pazocal O\). Hence Lemma~\ref{lemma:cvar} applies on
the domain considered here.

For each \(i\in\pazocal N\) and \(t\in\pazocal T\), the loss realization is \eqref{eq:loss_i_t}
\[
L^t_{i,\omega}(\mathbf h^t;\theta)
:=
-r^t_{i,\omega}(h_i^t)
+
p_\theta(h_i^t)
+
\psi_i^t(\mathbf h^t).
\]
For fixed \(\mathbf h^t\) and \(\theta\), the terms
\(p_\theta(h_i^t)\) and \(\psi_i^t(\mathbf h^t)\) are deterministic. Thus, by translation invariance of CVaR \citep[Prop.~2(i)]{pflug2000some}, 
\begin{equation}
    \mathrm{CVaR}_{\alpha_i}
\!\left(L_i^t(\mathbf h^t;\theta)\right)
=
p_\theta(h_i^t)
+
\psi_i^t(\mathbf h^t)
+
\mathrm{CVaR}_{\alpha_i}
\!\left(-r_i^t(h_i^t)\right)
\label{eq:cvar_decomposition_r}
\end{equation}

Using this identity in~\eqref{eq:cost_follower_j}, the disutility of
follower \(i\) can be written as
\[
P_i(\mathbf h_i,\mathbf h_{-i};\theta)
=
\sum_{t\in\pazocal T}
\left[
-\mathbb E_\omega r_i^t(h_i^t)
+
(1+\beta_i)
\left(
p_\theta(h_i^t)+\psi_i^t(\mathbf h^t)
\right)
+
\beta_i
\mathrm{CVaR}_{\alpha_i}
\!\left(-r_i^t(h_i^t)\right)
\right]
\]
Therefore, for every \(i\in\pazocal N\) and \(t\in\pazocal T\),
\begin{equation}
\begin{aligned}
\frac{\partial P_i(\mathbf h_i,\mathbf h_{-i};\theta)}
{\partial h_i^t}
=&
-
\frac{\partial}{\partial h_i^t}
\mathbb E_\omega r_i^t(h_i^t)
\\
&+
(1+\beta_i)
\left[
\frac{\partial p_\theta(h_i^t)}{\partial h_i^t}
+
\frac{\partial \psi_i^t(\mathbf h^t)}{\partial h_i^t}
\right]
\\
&+
\beta_i
\frac{\partial}{\partial h_i^t}
\mathrm{CVaR}_{\alpha_i}
\!\left(-r_i^t(h_i^t)\right).
\end{aligned}
\label{eq:partial_derivative_p}
\end{equation}

We now check the continuity of each term. Under
Assumption~\ref{ass:fct_c1},
\(
(h,\theta)\mapsto \frac{\partial p_\theta(h)}{\partial h}
\)
is continuous on \(\mathbb R_+\times[0,\overline \theta]\). Hence the
payment-gradient term $\frac{\partial p_\theta(h_i^t)}{\partial h_i^t}$ is continuous in \((h_i^t,\theta)\). By Assumption~\ref{ass:fct_c1}, the mapping
\(\mathbf h\mapsto \psi_i^t(\mathbf h)\) is continuously differentiable. Hence, the firm-side cost-gradient term
\(
\frac{\partial \psi_i^t(\mathbf h^t)}{\partial h_i^t}
\)
is continuous in \(\mathbf h^t\).

Moreover by Assumption~\ref{ass:fct_c1}, the mapping
\(h\mapsto r_{i,\omega}^t(h)\) is continuously differentiable. The continuity of the expected-revenue gradient follows from the dominated
convergence theorem and the standard differentiation-under-the-integral
result; see~\cite[Theorem~2.27]{folland1999real}. By
Assumption~\ref{ass:cvar} (2), the derivative of the random loss with
respect to the resource vector is dominated by an integrable random
variable. Since the access payment and firm-side cost derivatives are deterministic
and continuous on the compact feasible set, they are bounded there. Hence
\(\partial r_{i,\omega}^t(h_i^t)/\partial h_i^t\) is also dominated by an
integrable random variable. Therefore, applying \cite[Theorem~2.27]{folland1999real} componentwise,
\begin{equation}
    \frac{\partial}{\partial h_i^t}
\mathbb E_\omega [r_i^t(h_i^t)]
=
\mathbb E_\omega
\left[
\frac{\partial r_{i,\omega}^t(h_i^t)}{\partial h_i^t}
\right]
\label{eq:this_map}
\end{equation}
Moreover, since
\(\partial r_{i,\omega}^t(h_i^t)/\partial h_i^t\) is continuous in
\(h_i^t\) for almost every \(\omega\) and is dominated by the same
integrable random variable, dominated convergence implies that this
mapping \eqref{eq:this_map} is continuous in \(h_i^t\).

Finally, consider the CVaR term. Although Lemma~\ref{lemma:cvar} is stated
for the loss \(L_i^t\), it also implies the required regularity of the random
part. Indeed, fix any \(\theta_0\in[0,\overline \theta]\) and set the other
followers' resource commitments to zero. Then
\[
L^t_{i,\omega}(h_i^t,\mathbf 0;\theta_0)
=
-r^t_{i,\omega}(h_i^t)
+
p_{\theta_0}(h_i^t)
+
\psi_i^t(h_i^t,\mathbf 0).
\]
By \eqref{eq:cvar_decomposition_r},
\[
\mathrm{CVaR}_{\alpha_i}\!\left(-r_i^t(h_i^t)\right)
=
\mathrm{CVaR}_{\alpha_i}
\!\left(L_i^t(h_i^t,\mathbf 0;\theta_0)\right)
-
p_{\theta_0}(h_i^t)
-
\psi_i^t(h_i^t,\mathbf 0).
\]
The first term is continuously differentiable by Lemma~\ref{lemma:cvar},
and the last two terms are continuously differentiable by
Assumption~\ref{ass:fct_c1}. Hence
\(\mathrm{CVaR}_{\alpha_i}(-r_i^t(h_i^t))\) is continuously
differentiable in $h_i^t$, and its derivative is continuous.

Thus each component
\(
\frac{\partial P_i(\mathbf h_i,\mathbf h_{-i};\theta)}
{\partial h_i^t}
\)
is continuous in \((\mathbf H,\theta)\) on
\(X(C)\times[0,\overline \theta]\). Hence each block component
\(
\nabla_{\mathbf h_i}P_i(\mathbf h_i,\mathbf h_{-i};\theta)
\)
is continuous in \((\mathbf H,\theta)\). Since the pseudo-gradient mapping
\(\mathcal G(\mathbf H;\theta)\)~\eqref{eq:pseudo_grad} is obtained by
stacking finitely many continuous components, it is continuous in
\((\mathbf H,\theta)\) on \(X(C)\times[0,\overline \theta]\). This proves the
Lemma.
\end{proof}



We now ensure that the follower game admits at least one resource commitment equilibrium for any leader decision \((C,\theta)\). Without such a guarantee, the InP's anticipated follower response may not be well defined. 
\begin{lemma}\label{lem:ve_exists} 
Fix \((C,\theta)\in\pazocal Y\). Then the follower game admits at least one variational equilibrium. 
\end{lemma}

\begin{proof}
Fix \((C,\theta)\in\pazocal Y\). The joint feasible set~\eqref{eq:x_joint_feasible}
is nonempty, compact, and convex. Indeed, \(X(C)\) is nonempty since
\(\mathbf 0\in X(C)\); it is convex and closed because it is defined by linear
constraints; and it is bounded since \(0\le h_i^t\le C\) for every
\(i\in\pazocal N\) and \(t\in\pazocal T\).
Moreover, by Lemma~\ref{lem:F_cont}, the pseudo-gradient mapping
\(\mathcal G(\mathbf{H}; \theta)\) is continuous on \(X(C)\). Hence, by
\citep[Proposition~2.2]{facchinei2007generalized}, the variational
inequality problem
\(
\mathrm{VI}(X(C),\mathcal G(\cdot;\theta))
\)
admits at least one solution. Therefore, the follower game
admits at least one variational equilibrium.
\end{proof}
We now relate the VE to the GNE. 

\begin{lemma}\label{lem:ve_gne} Fix \((C,\theta)\in\pazocal Y\). Then any solution of the variational inequality \( \mathrm{VI}(X(C),\mathcal G(\cdot;\theta)) \)~\eqref{eq:VI} is a generalized Nash equilibrium of the follower game associated with \((C,\theta)\). Consequently, \( \mathrm{VE}(C,\theta)\subseteq \mathrm{GNE}(C,\theta)\). Therefore, the follower game
admits at least one GNE.
\end{lemma}

\begin{proof}
Fix \((C,\theta)\in\pazocal Y\). For each follower~\(i\in\pazocal N\), let
\(
\mathbf h_i := (h_i^t)_{t\in\pazocal T}
\)
denote the resource commitment vector of follower \(i\) over all time slots, and let
\(
\mathbf h_{-i}:=(\mathbf h_j)_{j\in\pazocal N,\;j\neq i}
\)
denote the collection of resource commitment vectors of all followers except \(i\).

For each follower~\(i\in\pazocal N\), the
feasible set is
\[
X_i(\mathbf h_{-i};C)
=
\left\{
\mathbf h_i\in\mathbb R_+^T
\;\middle|\;
h_i^t+\sum_{j\neq i} h_j^t \le C,\ \forall t\in\pazocal T
\right\}.
\]
Equivalently,
\[
X_i(\mathbf h_{-i};C)
=
\left\{
\mathbf h_i\in\mathbb R^T
\;\middle|\;
(\mathbf h_i,\mathbf h_{-i})\in X(C)
\right\},
\]
where \(X(C)\) is defined in~\eqref{eq:x_joint_feasible}.
Hence, the follower game is a generalized Nash equilibrium problem with a
jointly convex constraint set in the sense of
\citep{facchinei2007generalizedproblems}. Moreover, \(X(C)\) is nonempty,
closed, and convex. Since \(0\le h_i^t\le C\) for all
\(i\in\pazocal N\) and \(t\in\pazocal T\), the set \(X(C)\) is also compact.
By 
Lemma~\ref{lem:F_cont}, for every \(i\in\pazocal N\),
\(P_i(\mathbf h_i,\mathbf h_{-i};\theta)\) is continuously differentiable
with respect to \(\mathbf h_i\).

Fix $\theta$. For each
\(t\in\pazocal T\) and \(\omega\in\Omega\), the loss realization is
\[
L_{i,\omega}^t(\mathbf h^t;\theta)
= L_{i,\omega}^t(h_i^t,\mathbf h_{-i}^t;\theta)
=-r_{i,\omega}^t(h_i^t)
+
p_\theta(h_i^t)
+
\psi_i^t(h_i^t, \mathbf h_{-i}^t)
\]
Fix \(i\in\pazocal N\) and \(t\in\pazocal T\). Let
\(h\in\mathbb R_+\) denote a scalar value of follower \(i\)'s resource at time \(t\), and fix
\(
\mathbf h_{-i}^t := (h_j^t)_{j\ne i}\in\mathbb R_+^{N-1}.
\)
By Assumption~\ref{ass:convex_follower}, for each
\(t\in\pazocal T\), \(\omega\in\Omega\), and fixed
\(\mathbf h_{-i}^t\in\mathbb R_+^{N-1}\), the mapping
\[
h\mapsto
L_{i,\omega}^t(h,\mathbf h_{-i}^t;\theta)
=
-r_{i,\omega}^t(h)
+
p_\theta(h)
+
\psi_i^t(h,\mathbf h_{-i}^t)
\]
is convex on \(\mathbb R_+\), since
\(-r_{i,\omega}^t(h)\) is convex in \(h\),
\(p_\theta(h)\) is convex in \(h\), and
\(\psi_i^t(h, \mathbf h_{-i}^t)\) is convex in \(h\).

Hence, for every fixed \(\mathbf h_{-i}^t\), any
\(h,\tilde h\in\mathbb R_+\), and any \(\lambda\in[0,1]\), we have
\begin{equation}
L_{i,\omega}^t\!\left(
\lambda h+(1-\lambda)\tilde h,\mathbf h_{-i}^t;\theta
\right)
\le
\lambda L_{i,\omega}^t(h,\mathbf h_{-i}^t;\theta)
+
(1-\lambda)L_{i,\omega}^t(\tilde h,\mathbf h_{-i}^t;\theta)
\label{eq:cvar_composition_loss_convex}
\end{equation}

By monotonicity of \(\mathrm{CVaR}_{\alpha_i}\)
\citep[Prop.~2(v)]{pflug2000some},
\eqref{eq:cvar_composition_loss_convex} implies
\begin{equation}
\begin{aligned}
&\mathrm{CVaR}_{\alpha_i}\!\left(
L_i^t\!\left(
\lambda h+(1-\lambda)\tilde h,
\mathbf h_{-i}^t;\theta
\right)
\right) \\&\le
\mathrm{CVaR}_{\alpha_i}\!\left(
\lambda L_i^t(h,\mathbf h_{-i}^t;\theta)
+
(1-\lambda)L_i^t(\tilde h,\mathbf h_{-i}^t;\theta)
\right)
\label{eq:cvar_composition_monotone}
\end{aligned}
\end{equation}

Then, by convexity of \(\mathrm{CVaR}_{\alpha_i}\) with respect to its loss
argument~\citep[Prop.~2(iv)]{pflug2000some},
\begin{align}
&
\mathrm{CVaR}_{\alpha_i}\!\left(
\lambda L_i^t(h,\mathbf h_{-i}^t;\theta)
+
(1-\lambda)L_i^t(\tilde h,\mathbf h_{-i}^t;\theta)
\right)
\nonumber\\
&\le
\lambda
\mathrm{CVaR}_{\alpha_i}\!\left(
L_i^t(h,\mathbf h_{-i}^t;\theta)
\right)
+
(1-\lambda)
\mathrm{CVaR}_{\alpha_i}\!\left(
L_i^t(\tilde h,\mathbf h_{-i}^t;\theta)
\right)
\label{eq:cvar_composition_convex}
\end{align}

Combining \eqref{eq:cvar_composition_monotone} and
\eqref{eq:cvar_composition_convex}, we obtain
\begin{align}
&
\mathrm{CVaR}_{\alpha_i}\!\left(
L_i^t\!\left(
\lambda h+(1-\lambda)\tilde h,
\mathbf h_{-i}^t;\theta
\right)
\right)
\nonumber\\
&\le
\lambda
\mathrm{CVaR}_{\alpha_i}\!\left(
L_i^t(h,\mathbf h_{-i}^t;\theta)
\right)
+
(1-\lambda)
\mathrm{CVaR}_{\alpha_i}\!\left(
L_i^t(\tilde h,\mathbf h_{-i}^t;\theta)
\right)
\label{eq:cvar_loss_convex}
\end{align}

Hence, for every fixed \(\mathbf h_{-i}^t\),
\(
h
\mapsto
\mathrm{CVaR}_{\alpha_i}(L_i^t(h, \mathbf h_{-i}^t;\theta))
\)
is convex.

Define
\[
z_i^t(h_i^t,\mathbf h_{-i}^t;\theta)
:=
\mathbb E_\omega\!\left[
L_{i,\omega}^t(h_i^t,\mathbf h_{-i}^t)
\right]
+
\beta_i
\mathrm{CVaR}_{\alpha_i}
\!\left(
L_i^t(h_i^t,\mathbf h_{-i}^t;\theta)
\right)
\]
Since \(L_{i,\omega}^t(\cdot,\mathbf h_{-i}^t;\theta)\) is convex for every realization
\(\omega\), and since convexity is preserved under nonnegative weighted
sums~\citep[\S 3.2.1]{boyd2004convex}, the expectation term
\[
h \mapsto
\mathbb E_\omega\!\left[
L_{i,\omega}^t(h,\mathbf h_{-i}^t;\theta)
\right]
\]
is convex. Indeed, expectation is a nonnegative weighted average over the
realizations \(\omega\).

Also, as shown above,
\(
h\mapsto
\mathrm{CVaR}_{\alpha_i}
\!\left(L_i^t(h,\mathbf h_{-i}^t;\theta)\right)
\)
is convex~\eqref{eq:cvar_loss_convex}. Since \(\beta_i\ge0\), it follows that
\(h \mapsto z_i^t(h,\mathbf h_{-i}^t);\theta\) is convex.

Furthermore, since
\(
P_i(\mathbf h_i, \mathbf h_{-i};\theta)
=
\sum_{t\in\pazocal T}
z_i^t(h_i^t,\mathbf h_{-i}^t;\theta),
\)
where
\(
\mathbf h_i=(h_i^t)_{t\in\pazocal T},
\)
the function \(P_i(\mathbf h_i,\mathbf h_{-i};\theta)\) is a finite separable
sum of convex functions of the components of \(\mathbf h_i\). 


Therefore, for
every fixed \(\mathbf h_{-i}=( \mathbf h_j)_{j\ne i}\),
\(
\mathbf h_i
\mapsto
P_i(\mathbf h_i,\mathbf h_{-i};\theta)
\)
is convex since  convexity is
preserved under nonnegative weighted sums~\citep[\S 3.2.1]{boyd2004convex}.

Therefore, the Convexity Assumption in \citep{facchinei2007generalizedproblems}
is satisfied: for every follower \(i\), \(P_i(\cdot,\mathbf h_{-i};\theta)\)
is convex and \(X_i(\mathbf h_{-i};C)\) is
closed and convex.

Moreover, since
\(
X_i(\mathbf h_{-i};C)
=
\{\mathbf h_i\in\mathbb R^T:(\mathbf h_i,\mathbf h_{-i})\in X(C)\},
\)
with \(X(C)\) closed and convex, the follower game is a jointly convex GNE problem in
the sense of Definition~2 of \citep{facchinei2007generalizedproblems}.

Since
Lemma~\ref{lem:F_cont} shows that the pseudo-gradient mapping
\(\mathcal G(\mathbf{H};\theta)\), whose \(i\)-th block is
\(\nabla_{\mathbf h_i}P_i(\mathbf h_i, \mathbf h_{-i};\theta)\), is continuous, it follows that
\(P_i\) is continuously differentiable with respect to \(\mathbf h_i\).

Thus all assumptions of
\citep[Theorem~5]{facchinei2007generalizedproblems}
are satisfied, namely, first, \(P_i\) is \(C^1\) in \(\mathbf h_i\).
second, \(P_i\) is convex in \(\mathbf h_i\).
third, \(X_i(\mathbf h_{-i};C)\) is a section of the common closed convex set \(X(C)\).
Therefore, the follower game is a jointly convex GNE problem, and Theorem~5 applies. Therefore, every solution of
\(
\mathrm{VI}(X(C),\mathcal G(\cdot;\theta))
\)
is a GNE of the follower game associated with
\((C,\theta)\). Hence,
\begin{equation}
\mathrm{VE}(C,\theta)\subseteq \mathrm{GNE}(C,\theta)
\label{eq:ve_is_gne}
\end{equation}
By Lemma~\ref{lem:ve_exists}, the follower game admits at least one
VE. The existence of at least one GNE then follows from \eqref{eq:ve_is_gne}.
\end{proof}

\section{Step-by-Step Proof of Theorem~\ref{thm:gne_uniqueness}}
\label{appendix:gne_uniqueness}
We decompose the proof into the following four lemmas:

\begin{lemma}
For each \(i\in\pazocal N\), the disutility \(P_i(\mathbf h_i,\mathbf h_{-i};\theta)\) \eqref{eq:cost_follower_j} can be written as
\begin{equation}
P_i(\mathbf h_i,\mathbf h_{-i};\theta)
=
\sum_{t\in\pazocal T}
\Bigl(
\phi_i^t(h_i^t)
+
\Psi_i^t(\mathbf h^t)
\Bigr),
\label{eq:Ji_final_decomp}
\end{equation}
where
\begin{equation}
\phi_i^t(h_i^t)
:=
-\mathbb E_\omega[r_{i,\omega}^t(h_i^t)]
+
p_\theta(h_i^t)
+
\beta_i\,\mathrm{CVaR}_{\alpha_i}(f_i^t(h_i^t))
\label{eq:phi_def}
\end{equation}
and
\begin{equation}
\Psi_i^t(\mathbf h^t):=(1+\beta_i)\psi_i^t(\mathbf h^t).
\label{eq:coupled_term}
\end{equation}
\end{lemma}

\begin{proof}
For each \(i\in\pazocal N\), \(t\in\pazocal T\), and \(\omega \in \Omega\), rewrite~\eqref{eq:loss_i_t} as
\begin{equation}
L_{i,\omega}^t(\mathbf h^t;\theta)
=
\underbrace{-r_{i,\omega}^t(h_i^t)+p_\theta(h_i^t)}_{f_{i,\omega}^t(h_i^t)}
+
\psi_i^t(\mathbf h^t).
\label{eq:loss_decomp}
\end{equation}
Since \(\psi_i^t(\mathbf h^t)\) is deterministic, CVaR translation invariance~\citep[Prop.~2(i)]{pflug2000some} gives
\begin{equation}
\mathrm{CVaR}_{\alpha_i}\!\left(L_i^t(\mathbf h^t;\theta)\right)
\overset{\eqref{eq:loss_decomp}}{=}
\mathrm{CVaR}_{\alpha_i}\!\left(f_i^t(h_i^t)\right)
+
\psi_i^t(\mathbf h^t).
\label{eq:cvar_translation_model}
\end{equation}
Substituting \eqref{eq:loss_decomp}, \eqref{eq:cvar_translation_model}, and \eqref{eq:profit_sp_realization} into~\eqref{eq:cost_follower_j} gives
\begin{align}
&P_i(\mathbf h_i,\mathbf h_{-i};\theta) \notag\\
&\overset{\eqref{eq:cost_follower_j}}{=}
\sum_{t\in\pazocal T}
\Bigl(
-\mathbb E_\omega[\Pi_i^t(\mathbf h^t;\theta)]
+
\beta_i\,\mathrm{CVaR}_{\alpha_i}(L_i^t(\mathbf h^t;\theta))
\Bigr) \notag\\
&\overset{\eqref{eq:profit_sp_realization}}{=}
\sum_{t\in\pazocal T}
\Biggl(
-\mathbb E_\omega
\Bigl[
r_{i,\omega}^t(h_i^t)
-
p_\theta(h_i^t)
-
\psi_i^t(\mathbf h^t)
\Bigr] \notag\\
&\hspace{3.5cm}
+
\beta_i\,\mathrm{CVaR}_{\alpha_i}(L_i^t(\mathbf h^t;\theta))
\Biggr) \notag\\
&\overset{\eqref{eq:cvar_translation_model}}{=}
\sum_{t\in\pazocal T}
\Biggl(
-\mathbb E_\omega[r_{i,\omega}^t(h_i^t)]
+
p_\theta(h_i^t)
+
\beta_i\,\mathrm{CVaR}_{\alpha_i}(f_i^t(h_i^t)) \notag\\
&\hspace{3.5cm}
+
(1+\beta_i)\psi_i^t(\mathbf h^t)
\Biggr) \notag\\
&\overset{\eqref{eq:phi_def},\,\eqref{eq:coupled_term}}{=}
\sum_{t\in\pazocal T}
\Bigl(
\phi_i^t(h_i^t)
+
\Psi_i^t(\mathbf h^t)
\Bigr).
\end{align}
which completes the proof.
\end{proof}
The uniqueness analysis relies on the decomposition of $P_i(\mathbf h_i,\mathbf h_{-i};\theta)$. 

We first state a key property of the individual term~\eqref{eq:phi_def}.

\begin{lemma}\label{lemma:phi_convex}
For every \(i\in\pazocal N\) and \(t\in\pazocal T\), 
\(h_i^t\mapsto \phi_i^t(h_i^t)\) is convex.
\end{lemma}

\begin{proof}
By Assumption~\ref{ass:convex_follower}, for every \(\omega\in\Omega\), the
mapping \(h \mapsto r_{i,\omega}^t(h)\) is concave on \(\mathbb R_+\).
Hence, the mapping \(h \mapsto -r_{i,\omega}^t(h)\) is convex on
\(\mathbb R_+\). Since \(h \mapsto p_\theta(h)\) is also convex by
Assumption~\ref{ass:convex_follower}, it follows that
\(h \mapsto f_{i,\omega}^t(h)\) is convex on \(\mathbb R_+\), where
\(
f_{i,\omega}^t(h_i^t)
:=
-r_{i,\omega}^t(h_i^t)+p_\theta(h_i^t).
\)
Let
\(
f_i^t(h_i^t)
\)
denote the corresponding random variable.
We now show that the mapping
\(
h \mapsto
\mathrm{CVaR}_{\alpha_i}\!\bigl(f_i^t(h)\bigr)
\)
is convex on \(\mathbb R_+\).

Since \(f_{i,\omega}^t(\cdot)\) is convex for every \(\omega\in\Omega\),
the same monotonicity--convexity argument used in
\eqref{eq:cvar_composition_loss_convex}--\eqref{eq:cvar_loss_convex}
implies that
\(
h \mapsto
\mathrm{CVaR}_{\alpha_i}\!\bigl(f_i^t(h)\bigr)
\)
is convex on \(\mathbb R_+\).
Since convexity is preserved under nonnegative weighted sums~\citep[\S 3.2.1]{boyd2004convex}, and expectation is a nonnegative weighted average, mapping
\(h \mapsto -\mathbb E_\omega[r_{i,\omega}^t(h)]\)
is convex.
Since \(\beta_i\ge 0\),
mapping
\(h \mapsto
\beta_i\,\mathrm{CVaR}_{\alpha_i}\!\bigl(f_i^t(h)\bigr)\)
is also convex.

Therefore, \(\phi_i^t\), defined by
\(
\phi_i^t(h_i^t)
=
-\mathbb E_\omega[r_{i,\omega}^t(h_i^t)]
+
p_\theta(h_i^t)
+
\beta_i\,\mathrm{CVaR}_{\alpha_i}\!\bigl(f_i^t(h_i^t)\bigr),
\)
is a sum of convex mappings, and is thus convex in $h_i^t$ on \(\mathbb R_+\).
\end{proof}
We now show the strong monotonicity of~\eqref{eq:coupled_term}.
\begin{lemma}\label{lemma:congestion_coupled_term}
The firm-side cost-gradient mapping, whose \((i,t)\)-component is
\((1+\beta_i)\,\frac{\partial \psi_i^t(\mathbf h^t)}{\partial h_i^t}\), is strongly monotone on \(X(C)\)~\eqref{eq:x_joint_feasible}.
\end{lemma}

\begin{proof}
The firm-side cost is modeled as
\label{appendix:congestion_coupled_term}
\begin{equation}
\psi_i^t(\mathbf h^t)
=
q_i^t(h_i^t)
+
g_i^t(\mathbf h^t).
\label{eq:congestion_decomp}
\end{equation}
Define 
\begin{equation}
    Q:=\inf_{\substack{i\in\pazocal N,\; t\in\pazocal T,\\ h_i^t\in\mathbb R_+}}
(q_i^t)''(h_i^t)
,\qquad
G:=\sup_{\substack{i,j\in\pazocal N,\; t\in\pazocal T,\\ \mathbf h^t\in\mathbb R_+^{|\pazocal N|}}}
\left|
\frac{\partial^2 g_i^t(\mathbf h^t)}{\partial h_i^t\,\partial h_j^t}
\right|
\label{eq:cond_q_g}
\end{equation}
For each \(i\in\pazocal N\), differentiating \eqref{eq:congestion_decomp} with respect to the own decision \(h_i^t\), we obtain
\begin{equation}
\frac{\partial \psi_i^t(\mathbf h^t)}{\partial h_i^t}
=
(q_i^t)'(h_i^t)
+
\frac{\partial g_i^t(\mathbf h^t)}{\partial h_i^t}
\label{eq:first_partial_general_with_t}
\end{equation}
Multiplying by \(1+\beta_i\), it follows that
\begin{equation}
M_i^t(\mathbf h^t)
:=
(1+\beta_i)
\left(
(q_i^t)'(h_i^t)
+
\frac{\partial g_i^t(\mathbf h^t)}{\partial h_i^t}
\right),
\qquad \forall i\in\pazocal N
\label{eq:F_i_general_with_t}
\end{equation}

Equivalently, for each \(t\in\pazocal T\), define the vector mapping
\begin{align}
&M^t(\mathbf h^t)
:=
\bigl(M_i^t(\mathbf h^t)\bigr)_{i\in\pazocal N}
\nonumber\\&=
\left(
(1+\beta_i)
\left[
(q_i^t)'(h_i^t)
+
\frac{\partial g_i^t(\mathbf h^t)}{\partial h_i^t}
\right]
\right)_{i\in\pazocal N}
\label{eq:m_t}
\end{align}
We now compute the Jacobian matrix of the vector mapping \(M^t\)~\eqref{eq:m_t}.
Let
\begin{align}
J^t(\mathbf h^t)
&=
\nabla_{\mathbf h^t} M^t(\mathbf h^t)
=
\left[
\frac{\partial M_i^t(\mathbf h^t)}{\partial h_j^t}
\right]_{i,j\in\pazocal N}
\nonumber\\
&=
\begin{pmatrix}
\partial_{h_1^t}M_1^t(\mathbf h^t) & \cdots & \partial_{h_N^t}M_1^t(\mathbf h^t) \\
\vdots & \ddots & \vdots \\
\partial_{h_1^t}M_N^t(\mathbf h^t) & \cdots & \partial_{h_N^t}M_N^t(\mathbf h^t)
\end{pmatrix}
\label{eq:jacobian_matrix_entries_M}
\end{align}
Using \eqref{eq:F_i_general_with_t}, we obtain, for each \(i,j\in\pazocal N\),
\begin{align}
\frac{\partial M_i^t(\mathbf h^t)}{\partial h_j^t}
&=
(1+\beta_i)
\frac{\partial}{\partial h_j^t}
\left(
(q_i^t)'(h_i^t)
+
\frac{\partial g_i^t(\mathbf h^t)}{\partial h_i^t}
\right)
\nonumber\\
&=
(1+\beta_i)
\left(
(q_i^t)''(h_i^t)\delta_{ij}
+
\frac{\partial^2 g_i^t(\mathbf h^t)}{\partial h_i^t\partial h_j^t}
\right),
\label{eq:jacobian_entry_general_with_t}
\end{align}
where \(\delta_{ij}\) denotes the Kronecker delta.\footnote{\(\delta_{ij}=1\) if \(i=j\), and \(0\) otherwise.} Therefore,
\begin{equation}
J^t(\mathbf h^t)
=
\left[
(1+\beta_i)
\left(
(q_i^t)''(h_i^t)\delta_{ij}
+
\frac{\partial^2 g_i^t(\mathbf h^t)}{\partial h_i^t\partial h_j^t}
\right)
\right]_{i,j\in\pazocal N}
\label{eq:jacobian_matrix_general_with_t}
\end{equation}

We study the symmetric matrix given by the symmetric part of the Jacobian defined as:
\begin{align}
S^t(\mathbf h^t)
&:=
\frac{J^t(\mathbf h^t)+\big(J^t(\mathbf h^t)\big)^\top}{2}
\nonumber\\
&=
\left[
\frac{1}{2}
\left(
\frac{\partial M_i^t(\mathbf h^t)}{\partial h_j^t}
+
\frac{\partial M_j^t(\mathbf h^t)}{\partial h_i^t}
\right)
\right]_{i,j\in\pazocal N}
\label{eq:jacobian}
\end{align}
Its \((i,j)\)-entry is
\begin{align}
S^t_{ij}(\mathbf h^t)
&=
\frac{1}{2}
\left(
\frac{\partial M_i^t(\mathbf h^t)}{\partial h_j^t}
+
\frac{\partial M_j^t(\mathbf h^t)}{\partial h_i^t}
\right)
\nonumber\\
&=
\frac{1}{2}
\Bigg[
(1+\beta_i)
\left(
(q_i^t)''(h_i^t)\delta_{ij}
+
\frac{\partial^2 g_i^t(\mathbf h^t)}{\partial h_i^t\partial h_j^t}
\right)
\nonumber\\
&\hspace{0.5cm}
+
(1+\beta_j)
\left(
(q_j^t)''(h_j^t)\delta_{ji}
+
\frac{\partial^2 g_j^t(\mathbf h^t)}{\partial h_j^t\partial h_i^t}
\right)
\Bigg]
\label{eq:sym_entry_general_with_t}
\end{align}

We distinguish between the diagonal and off-diagonal entries.

\medskip
\paragraph{\textbf{Diagonal entries}}

Let \(i\in\pazocal N\). Setting \(j=i\) in \eqref{eq:sym_entry_general_with_t}, and using \(\delta_{ii}=1\), we obtain
\begin{align}
S^t_{ii}(\mathbf h^t)
&=
(1+\beta_i)
\left(
(q_i^t)''(h_i^t)
+
\frac{\partial^2 g_i^t(\mathbf h^t)}{\partial (h_i^t)^2}
\right)
\label{eq:diag_entry_general_with_t}
\end{align}
By \eqref{eq:cond_q_g}, we have
\(
(q_i^t)''(h_i^t)\ge Q
\) and
\(
\left|
\frac{\partial^2 g_i^t(\mathbf h^t)}{\partial (h_i^t)^2}
\right|
\le G,
\) respectively.
The latter implies
\(
\frac{\partial^2 g_i^t(\mathbf h^t)}{\partial (h_i^t)^2}\ge -G
\).
Substituting these two bounds into \eqref{eq:diag_entry_general_with_t}, we obtain
\begin{align}
S^t_{ii}(\mathbf h^t)
&\ge
(1+\beta_i)(Q-G)
\label{eq:diag_lower_bound_general_with_t}
\end{align}

\medskip

\paragraph{\textbf{Off-diagonal entries}}

Let \(i,j\in\pazocal N\) with \(i\neq j\). Then \(\delta_{ij}=\delta_{ji}=0\), and \eqref{eq:sym_entry_general_with_t} becomes
\begin{align}
S^t_{ij}(\mathbf h^t)
&=
\frac{1}{2}
\left[
(1+\beta_i)
\frac{\partial^2 g_i^t(\mathbf h^t)}{\partial h_i^t\partial h_j^t}
+
(1+\beta_j)
\frac{\partial^2 g_j^t(\mathbf h^t)}{\partial h_j^t\partial h_i^t}
\right]
\label{eq:offdiag_entry_general_with_t}
\end{align}
Taking absolute values and using the triangle inequality, we obtain
\begin{align}
|S^t_{ij}(\mathbf h^t)|
&\le
\frac{1}{2}
\left[
(1+\beta_i)
\left|
\frac{\partial^2 g_i^t(\mathbf h^t)}{\partial h_i^t\partial h_j^t}
\right|
+
(1+\beta_j)
\left|
\frac{\partial^2 g_j^t(\mathbf h^t)}{\partial h_j^t\partial h_i^t}
\right|
\right]
\label{eq:offdiag_abs_step_with_t}
\end{align}
Applying \eqref{eq:cond_q_g} to each term yields
\begin{align}
|S^t_{ij}(\mathbf h^t)|
&\le
\frac{1}{2}
\Big(
(1+\beta_i)G+(1+\beta_j)G
\Big)
\label{eq:offdiag_upper_bound_general_with_t}
\end{align}

\medskip

\paragraph{\textbf{Positive definiteness of \(S^t(\mathbf h^t)\)}}

Summing \eqref{eq:offdiag_upper_bound_general_with_t} over all \(j\neq i\), we obtain
\begin{align}
\sum_{j\neq i}|S^t_{ij}(\mathbf h^t)|
&\le
\frac{1}{2}
\sum_{j\neq i}
\Big(
(1+\beta_i)G+(1+\beta_j)G
\Big)
\label{eq:row_offdiag_bound_general_with_t}
\end{align}
Subtracting \eqref{eq:row_offdiag_bound_general_with_t} from \eqref{eq:diag_lower_bound_general_with_t}, we get
\begin{align}
&S^t_{ii}(\mathbf h^t)
-
\sum_{j\neq i}|S^t_{ij}(\mathbf h^t)|
\ge
(1+\beta_i)(Q
-G)
-
\frac{1}{2}
\sum_{j\neq i}
\Big(
(1+\beta_i)G+(1+\beta_j)G
\Big)
\label{eq:diag_minus_offdiag_general_with_t}
\end{align}
For any $i \in \pazocal N$, we simplify the right-hand side of \eqref{eq:diag_minus_offdiag_general_with_t} as follows:
\begin{align}
&(1+\beta_i)(Q-G)
-
\frac{1}{2}
\sum_{j\neq i}
\Big(
(1+\beta_i)G+(1+\beta_j)G
\Big)
\nonumber\\
&=
(1+\beta_i)Q
-(1+\beta_i)G
\nonumber\\&-\frac{N-1}{2}(1+\beta_i)G
-\frac12\sum_{j\neq i}(1+\beta_j)G
\nonumber\\
&=
(1+\beta_i)Q
-\left(
1+\frac{N-1}{2}
\right)(1+\beta_i)G
\nonumber\\&-\frac12\sum_{j\neq i}(1+\beta_j)G
\nonumber\\
&=
(1+\beta_i)Q
-\frac{N+1}{2}(1+\beta_i)G
-\frac12\sum_{j\neq i}(1+\beta_j)G
\nonumber\\
&=
(1+\beta_i)Q
-\frac{N+1}{2}(1+\beta_i)G
-\frac12\sum_{j\neq i}(1+\beta_j)G
\nonumber\\&+\frac12(1+\beta_i)G
-\frac12(1+\beta_i)G
\nonumber\\
&=
(1+\beta_i)Q
-\frac{N}{2}(1+\beta_i)G
-\frac12\sum_{j\in\pazocal N}(1+\beta_j)G
\nonumber\\
&=
(1+\beta_i)Q
-\frac{N}{2}(1+\beta_i)G
-\frac12\left(N+\sum_{j\in\pazocal N}\beta_j\right)G
\nonumber\\
&=
(1+\beta_i)Q
-
\frac{
2N+N\beta_i+\sum_{j\in\pazocal N}\beta_j
}{2}\,G
\nonumber\\
&=
(1+\beta_i)\bigl(Q-\xi_i G\bigr),
\label{eq:xi_algebra_general_with_t}
\end{align}
where
\(
\xi_i
=
\frac{
2N+N\beta_i+\sum_{j\in\pazocal N}\beta_j
}{
2(1+\beta_i)
}
\).
By \eqref{eq:xi_algebra_general_with_t}, the right-hand side of
\eqref{eq:diag_minus_offdiag_general_with_t} is equal to
\((1+\beta_i)(Q-\xi_iG)\). Hence, \eqref{eq:diag_minus_offdiag_general_with_t} becomes
\begin{equation}
S^t_{ii}(\mathbf h^t)
-
\sum_{j\neq i}|S^t_{ij}(\mathbf h^t)|
\ge
(1+\beta_i)(Q-\xi_iG)
\label{eq:sii_sij_q_g}
\end{equation}
Since \(Q-\xi_iG>0\) by \eqref{eq:ass_q_g} and \(1+\beta_i>0\), we obtain
\begin{equation}
S^t_{ii}(\mathbf h^t)
-
\sum_{j\neq i}|S^t_{ij}(\mathbf h^t)|
>0
\end{equation}
Therefore,
\begin{equation}
S^t_{ii}(\mathbf h^t)
>
\sum_{j\neq i}|S^t_{ij}(\mathbf h^t)|
\label{eq:pos}
\end{equation}
Since the off-diagonal absolute-value sum is nonnegative, \eqref{eq:pos} also implies
\(S^t_{ii}(\mathbf h^t)>0\). Hence
\(
|S^t_{ii}(\mathbf h^t)|=S^t_{ii}(\mathbf h^t),
\)
and, by \eqref{eq:pos},
\(
|S^t_{ii}(\mathbf h^t)|
>
\sum_{j\neq i}|S^t_{ij}(\mathbf h^t)|
\)
Therefore, for each \(t\in\pazocal T\) and each \(\mathbf h\in X(C)\), the symmetric matrix
\(S^t(\mathbf h^t)\) is strictly row diagonally dominant in the sense of
\citep[Theorem~8.6(1)]{gallier2022applications}.

Since \(S^t(\mathbf h^t)\) is strictly row diagonally dominant and has
\(S^t_{ii}(\mathbf h^t)>0\) for all \(i\), \citep[Theorem~8.6(2)]{gallier2022applications}
implies that all eigenvalues of the symmetric matrix \(S^t(\mathbf h^t)\) are strictly positive.
Consequently, \(S^t(\mathbf h^t)\) is positive definite for every \(t\in\pazocal T\) and every \(\mathbf h\in X(C)\)~\citep[Theorem~2.4]{agosti2005theoretical}, i.e., $S^t(\mathbf h^t)\succ0$.

\paragraph{\textbf{Uniform lower bound of eigenvalues}}

Define
\begin{align}
m
:=
\min_{i\in\pazocal N}
\left\{
(1+\beta_i)\bigl(Q-\xi_i G\bigr)
\right\}
\label{eq:m_general_vector_with_t}
\end{align}
By \eqref{eq:ass_q_g}, we have \(m>0\).

We now show that every eigenvalue of \(S^t(\mathbf h^t)\) is bounded below by \(m\). Let \(\lambda\) be any eigenvalue of \(S^t(\mathbf h^t)\). Indeed, \citep[Lemma~2]{luo2022directional} implies that every eigenvalue \(\lambda\) of the symmetric matrix \(S^t(\mathbf h^t)\) belongs to at least one interval of the form
\begin{align*}
\Bigg[
S^t_{ii}(\mathbf h^t)-\sum_{j\neq i}|S^t_{ij}(\mathbf h^t)|,\,
S^t_{ii}(\mathbf h^t)+\sum_{j\neq i}|S^t_{ij}(\mathbf h^t)|
\Bigg]
\end{align*}
for some \(i\in\pazocal N\).

By \eqref{eq:sii_sij_q_g} and the definition of \(m\) in \eqref{eq:m_general_vector_with_t}, the left endpoint of each such interval is at least \(m\). Therefore, every eigenvalue of \(S^t(\mathbf h^t)\) satisfies
\begin{equation}
    \lambda \ge m
    \label{eq:lambda_m}
\end{equation}
Let
\(v\in\mathbb R^{|\pazocal N|}\setminus\{0\}\)
be an associated eigenvector of \(\lambda\). Then
\(
S^t(\mathbf h^t)v=\lambda v
\)
Consequently,
\(
\bigl(S^t(\mathbf h^t)-mI\bigr)v
=
(\lambda-m)v
\)
Therefore, the eigenvalues of
\(
S^t(\mathbf h^t)-mI
\)
are \(\lambda-m\).
By~\eqref{eq:lambda_m}, every eigenvalue of
\(
S^t(\mathbf h^t)-mI
\)
is of the form \(\lambda-m\ge 0\). Hence,
\(
S^t(\mathbf h^t)-mI \succeq 0,
\)
that is,
\begin{equation}
S^t(\mathbf h^t)\succeq mI,
\qquad
\forall t\in\pazocal T,\ \forall \mathbf h\in X(C).
\label{eq:uniform_psd_bound_general_with_t}
\end{equation}

\paragraph{\textbf{Strong monotonicity of the full coupling term}}
For each \(t\in\pazocal T\), 
since each \(M^t(\cdot)\)~\eqref{eq:m_t} is \(C^1\) on \(\mathbb{R}_+^{|\pazocal N|}\) and satisfies
\(S^t(\mathbf h^t)\succeq mI\) for all \(t\in\pazocal T\) and
\(\mathbf h\in X(C)\) by \eqref{eq:uniform_psd_bound_general_with_t},
the Jacobian criterion for strong monotonicity
\citep[Prop.~3(a)]{parise2019variational} and
\citep[Prop.~2.3.2(c)]{facchinei2003finitebook} yields 
\begin{equation}
\Big(
M^t(\mathbf h^t)-M^t(\mathbf h^{t\prime})
\Big)^\top
(\mathbf h^t-\mathbf h^{t\prime})
\ge
m\|\mathbf h^t-\mathbf h^{t\prime}\|^2,
\qquad \forall t\in\pazocal T
\label{eq:strong_monotonicity_each_t}
\end{equation}
Condition~\eqref{eq:strong_monotonicity_each_t} is the standard definition of
strong monotonicity, following \citep[Def.~2.3.1 (e)]{facchinei2003finitebook}.

Summing \eqref{eq:strong_monotonicity_each_t} over all
\(t\in\pazocal T\), we obtain
\begin{align}
&\sum_{t\in\pazocal T}
\Big(
M^t(\mathbf h^t)-M^t(\mathbf h^{t\prime})
\Big)^\top
(\mathbf h^t-\mathbf h^{t\prime})
\nonumber\\
&\qquad \ge
m\sum_{t\in\pazocal T}\|\mathbf h^t-\mathbf h^{t\prime}\|^2
=
m\|\mathbf{H}-\mathbf{H}'\|^2,
\qquad \forall \mathbf{H},\mathbf{H}'\in X(C)
\label{eq:strong_monotonicity_vector_form_with_t}
\end{align}

Expanding this inequality componentwise, we obtain
\begin{align}
&\sum_{i\in\pazocal N}\sum_{t\in\pazocal T}
(1+\beta_i)
\left(
\frac{\partial \psi_i^t(\mathbf h^t)}{\partial h_i^t}
-
\frac{\partial \psi_i^t(\mathbf h^{t\prime})}{\partial h_i^t}
\right)
(h_i^t-h_i^{t\prime})
\nonumber\\
&\qquad \ge
m\|\mathbf{H}-\mathbf{H}'\|^2,
\qquad \forall \mathbf{H},\mathbf{H}'\in X(C).
\label{eq:strong_monotonicity_general_vector_final_with_t}
\end{align}

Therefore, the coupling term~\eqref{eq:coupled_term} is strongly monotone on \(X(C)\).
\end{proof}
We now show uniqueness.

\begin{lemma}
\label{lem:ve_unique}
    VE is unique.
\end{lemma}

\begin{proof}
By Assumption~\ref{ass:fct_c1} and Lemma~\ref{lemma:cvar}, for every
\(i\in\pazocal N\) and \(t\in\pazocal T\), the mappings
\(h_i^t\mapsto \phi_i^t(h_i^t)\) \eqref{eq:phi_def} and
\(\mathbf h^t\mapsto \psi_i^t(\mathbf h^t)\) are continuously differentiable.
From \eqref{eq:Ji_final_decomp}, for each player \(i\), we have
\[
P_i(\mathbf h_i, \mathbf h_{-i};\theta)
=
\sum_{t\in\pazocal T}
\left(
\phi_i^t(h_i^t)+(1+\beta_i)\psi_i^t(\mathbf h^t)
\right)
\]
Differentiating with respect to
\(\mathbf h_i=(h_i^t)_{t\in\pazocal T}\) gives
\begin{align}
\nabla_{\mathbf h_i}P_i(\mathbf h_i, \mathbf h_{-i};\theta)
&=
\left(
\frac{\partial P_i(\mathbf h_i, \mathbf h_{-i};\theta)}{\partial h_i^t}
\right)_{t\in\pazocal T}
\nonumber\\
&=
\left(
(\phi_i^t)'(h_i^t)
+
(1+\beta_i)\frac{\partial \psi_i^t(\mathbf h^t)}{\partial h_i^t}
\right)_{t\in\pazocal T}
\label{eq:block_gradient}
\end{align}

Let $\mathbf{H},\mathbf{H}'\in X(C)$ \eqref{eq:x_joint_feasible}. Using the definition of the pseudo-gradient $\mathcal{G}(\cdot; \theta)$ \eqref{eq:pseudo_grad}
we obtain
\begin{align}
&\left\langle
\mathcal G(\mathbf{H};\theta)-\mathcal G(\mathbf{H}';\theta),
\mathbf{H}-\mathbf{H}'
\right\rangle
\nonumber\\
&=
\sum_{i\in\pazocal N}
\left\langle
\nabla_{\mathbf h_i}P_i(\mathbf h_i, \mathbf h_{-i};\theta)-\nabla_{\mathbf h_i}P_i(\mathbf{H}';\theta),
\mathbf h_i-\mathbf h_i'
\right\rangle
\nonumber\\
&\overset{\eqref{eq:block_gradient}}{=} \sum_{i\in\pazocal N}\sum_{t\in\pazocal T}
\left[
(\phi_i^t)'(h_i^t)
+
(1+\beta_i)\frac{\partial \psi_i^t(\mathbf h^t)}{\partial h_i^t}
\right. \nonumber \\
&\quad \left.
-
(\phi_i^t)'(h_i^{t\prime})
-
(1+\beta_i)\frac{\partial \psi_i^t(\mathbf h^{t\prime})}{\partial h_i^t}
\right]
(h_i^t-h_i^{t\prime})
\nonumber\\
&=
\sum_{i\in\pazocal N}\sum_{t\in\pazocal T}
\Big(
(\phi_i^t)'(h_i^t)-(\phi_i^t)'(h_i^{t\prime})
\Big)
(h_i^t-h_i^{t\prime})
\nonumber\\
&\qquad
+
\sum_{i\in\pazocal N}\sum_{t\in\pazocal T}
(1+\beta_i)
\left(
\frac{\partial \psi_i^t(\mathbf h^t)}{\partial h_i^t}
-
\frac{\partial \psi_i^t(\mathbf h^{t\prime})}{\partial h_i^t}
\right)
(h_i^t-h_i^{t\prime})
\label{eq:split_main}
\end{align}

We now analyze the two sums separately.

\medskip
\noindent
\textit{First term.}
Since \(\phi_i^t\) defined in~\eqref{eq:phi_def} is convex by
Lemma~\ref{lemma:phi_convex} and differentiable by Assumption~\ref{ass:fct_c1} and
Lemma~\ref{lemma:cvar}, its derivative is nondecreasing on the feasible interval \([0,C]\)~\citep[Prop. 3.17]{penot2013calculus}.
Therefore, for every \(i\in\pazocal N\), \(t\in\pazocal T\), and
\(h_i^t,h_i^{t\prime}\) in \([0,C]\),
\begin{equation}
\Big(
(\phi_i^t)'(h_i^t)-(\phi_i^t)'(h_i^{t\prime})
\Big)
(h_i^t-h_i^{t\prime})
\ge 0
\label{eq:phi_monotone}
\end{equation}
Summing \eqref{eq:phi_monotone} over all $i$ and $t$ yields
\begin{equation}
\sum_{i\in\pazocal N}\sum_{t\in\pazocal T}
\Big(
(\phi_i^t)'(h_i^t)-(\phi_i^t)'(h_i^{t\prime})
\Big)
(h_i^t-h_i^{t\prime})
\ge 0
\label{eq:first_term_nonneg}
\end{equation}

\medskip
\noindent
\textit{Second term.}
By Lemma~\ref{lemma:congestion_coupled_term}, the coupling term is strongly monotone, i.e., 
\begin{align}
&\sum_{i\in\pazocal N}\sum_{t\in\pazocal T}
(1+\beta_i)
\left(
\frac{\partial \psi_i^t(\mathbf h^t)}{\partial h_i^t}
-
\frac{\partial \psi_i^t(\mathbf h'^t)}{\partial h_i^t}
\right)
(h_i^t-h_i^{t\prime})
\nonumber \\&\ge
m\|\mathbf{H}-\mathbf{H}'\|^2
\label{eq:second_term_bound}
\end{align}
Finally, substituting \eqref{eq:first_term_nonneg} and
\eqref{eq:second_term_bound} into \eqref{eq:split_main}, we obtain
\begin{align}
\left\langle
\mathcal G(\mathbf{H};\theta)-\mathcal G(\mathbf{H}';\theta),
\mathbf{H}-\mathbf{H}'
\right\rangle
&\ge
0
+
m\|\mathbf{H}-\mathbf{H}'\|^2
\nonumber\\
&=
m\|\mathbf{H}-\mathbf{H}'\|^2
\label{eq:G_str_monotone}
\end{align}
Hence, \(\mathcal G(\mathbf{H};\theta)\) is strongly monotone on \(X(C)\) (also called \(2\) monotone). \eqref{eq:G_str_monotone} is the standard definition of
strong monotonicity, following \citep[Def.~2.3.1 (e)]{facchinei2003finitebook}.
Therefore, by \citep[Theorem~2.3.3 (b)]{facchinei2003finitebook}, the variational
inequality problem \(\mathrm{VI}(X(C),\mathcal G(\cdot;\theta))\)~\eqref{eq:VI} admits a unique
solution. Thus, the follower game admits a unique VE.
\end{proof}

\section{Discussion on the uniqueness condition}\label{appendix:condition_unique}
Condition \eqref{eq:ass_q_g} is used to ensure that, for each leader decision \((C,\theta)\), the selected follower response is unique. This condition is sufficient and is not required for the existence of a follower GNE. Indeed, even when Condition \eqref{eq:ass_q_g} is not satisfied, the follower game may still admit a GNE by Proposition~\ref{prop:ve_gne_existence}. However, without uniqueness, the follower response to a given \((C,\theta)\) may no longer be single-valued. In other words, the same capacity–price pair may lead to several possible follower equilibria. In that case, the InP cannot predict a unique vector of firm resource commitments from \((C,\theta)\) alone. The leader problem then becomes an equilibrium-selection problem, because the InP profit may differ across the possible follower equilibria. This does not invalidate the follower problem~\S~\ref{sec:follower_problem}, but it changes the interpretation of the upper-level problem. One possible formulation is an optimistic leader problem, in which the InP assumes that, among the possible follower equilibria, the one giving the highest InP profit is selected. Another possible formulation is a conservative leader problem, in which the InP evaluates its decision under the follower equilibrium giving the lowest InP profit (see \citep{bacsar1998dynamic}). 

\section{Step-by-Step Proof of Proposition~\ref{prop:continuity}}
\label{appendix:continuity}
This proof explores special structure of the feasible set $X(C)$~\eqref{eq:x_joint_feasible}. We decompose the proof into two lemmas.
\begin{lemma}\label{lem:feasible_set_structure}
Let $\{(C_n,\theta_n)\}_{n\ge1}\subset\pazocal Y$~\eqref{eq:leader_set} be any sequence such that
\(
(C_n,\theta_n)\to (C,\theta)\in\pazocal Y
\). Define
\(
\mathbf{H}_n := \mathbf{H}^*(C_n,\theta_n), n\ge1.
\)
Then the sequence $\{\mathbf H_n\}$ is bounded. Moreover, for any convergent subsequence $\mathbf H_{n_k}\to\overline{\mathbf H}$, its limit satisfies $\overline{\mathbf H}\in X(C)$, and for every $\mathbf Y\in X(C)$ there exists a sequence $\mathbf Y_k\in X(C_{n_k})$ such that $\mathbf Y_k\to \mathbf Y$.
\end{lemma}

\begin{proof}

Since $(C_n,\theta_n)\to(C,\theta)$, the sequence $\{C_n\}$ is bounded~\citep[Theorem 3.2]{rudin1976principles}.
Hence there exists a constant $\overline C>0$ such that
\[
0\le C_n\le \overline C,
\qquad \forall n
\]
Because $\mathbf{H}_n\in X(C_n)$, we have for every $i\in\pazocal N$ and
$t\in\pazocal T$,
\[
0\le h_{i,n}^t \le \sum_{j\in\pazocal N} h_{j,n}^t \le C_n \le \overline C
\]
Therefore,
\(
\mathbf{H}_n\in [0,\overline C]^{NT}, \forall n,
\)
and thus the sequence $\{\mathbf{H}_n\}$ is bounded in~$\mathbb R^{NT}$.

Since bounded sequences in finite-dimensional spaces admit convergent
subsequences~\citep[Lemma 1.4]{sasane2017friendly}, there exists a subsequence $\{\mathbf{H}_{n_k}\}_{k\ge1}$ and a
vector $\overline{\mathbf{H}}\in\mathbb R^{NT}$ such that
\begin{equation}
\mathbf{H}_{n_k}\to \overline{\mathbf{H}}
\qquad \text{as } k\to\infty
\label{eq:subseq_conv}
\end{equation}
Since $\mathbf{H}_{n_k}\in X(C_{n_k})$, we have for every $t\in\pazocal T$,
\[
\sum_{j\in\pazocal N} h_{j,n_k}^t \le C_{n_k},
\qquad
h_{i,n_k}^t\ge 0,\ \forall i\in\pazocal N
\]
Passing to the limit as $k\to\infty$, using
\eqref{eq:subseq_conv} and the fact that $C_{n_k}\to C$, we obtain
\[
\sum_{j\in\pazocal N} \overline h_j^t \le C,
\qquad
\overline h_i^t\ge 0,\ \forall i\in\pazocal N.
\]
Thus
\begin{equation}
\overline{\mathbf{H}}\in X(C)
\label{eq:hbar_feasible}
\end{equation}
Let $\mathbf Y\in X(C)$ be arbitrary. We show that there exists a sequence
$\mathbf Y_k\in X(C_{n_k})$ such that
\(
\mathbf Y_k\to \mathbf Y
\).

If $C=0$, then necessarily $X(C)=X(0)=\{\mathbf 0\}$, hence
$\mathbf Y=\mathbf 0$, and we may simply take $\mathbf Y_k=\mathbf 0$ for all
$k$. Then clearly $\mathbf Y_k\in X(C_{n_k})$ and
$\mathbf Y_k\to \mathbf Y$.

Assume now that $C>0$. Define
\begin{equation}
\mathbf Y_k := \frac{C_{n_k}}{C}\,\mathbf Y
\label{eq:y_k_def}
\end{equation}
Since $\mathbf Y\in X(C)$, we have $\mathbf Y\ge 0$ and
\(
\sum_{j\in\pazocal N} y_j^t \le C,
\quad \forall t\in\pazocal T
\).
Therefore, for every $t\in\pazocal T$,
\[
\sum_{j\in\pazocal N} y_{j,k}^t
=
\frac{C_{n_k}}{C}\sum_{j\in\pazocal N} y_j^t
\le
\frac{C_{n_k}}{C}\,C
=
C_{n_k},
\]
and clearly $y_{i,k}^t\ge 0$ for all $i,t$. Hence
\begin{equation}
\mathbf Y_k\in X(C_{n_k}),
\qquad \forall k
\label{eq:y_k_feasible}
\end{equation}
Moreover, since $C_{n_k}\to C$, from \eqref{eq:y_k_def} we obtain
\begin{equation}
\mathbf Y_k\to \mathbf Y
\label{eq:y_k_conv}
\end{equation}
This completes the proof of Lemma~\ref{lem:feasible_set_structure}.
\end{proof}
\begin{lemma}\label{lem:limit_vi}
Under the setting of Lemma~\ref{lem:feasible_set_structure}, let
$\{\mathbf H_{n_k}\}_{k\ge1}$ be any convergent subsequence such that
\(
\mathbf H_{n_k}\to \overline{\mathbf H}.
\)
Then
\(
\overline{\mathbf H}=\mathbf H^*(C,\theta).
\)
Consequently,
\(
\mathbf H_n\to \mathbf H^*(C,\theta),
\)
and the map $(C,\theta)\mapsto \mathbf H^*(C,\theta)$ is continuous on
$\pazocal Y$.
\end{lemma}

\begin{proof}
By Lemma~\ref{lem:feasible_set_structure}, we have
\[
\overline{\mathbf H}\in X(C).
\]

Fix an arbitrary $\mathbf Y\in X(C)$. By Lemma~\ref{lem:feasible_set_structure}, there exists a sequence
$\mathbf Y_k\in X(C_{n_k})$ such that
\(
\mathbf Y_k\to \mathbf Y
\).

Since $\mathbf Y_k\in X(C_{n_k})$, the variational inequality
\eqref{eq:VI} applied at index $n_k$ gives
\begin{equation}
\left\langle
\mathcal G(\mathbf{H}_{n_k};\theta_{n_k}),\,
\mathbf Y_k-\mathbf{H}_{n_k}
\right\rangle
\ge 0,
\qquad \forall k.
\label{eq:vi_subseq}
\end{equation}

By Lemma~\ref{lem:F_cont} applied on the common set
\(X(\overline C)\times[0,\overline \theta]\), we obtain
\[
\mathcal G(\mathbf H_{n_k};\theta_{n_k})
=
\mathcal G(\mathbf H_{n_k};\theta_{n_k})
\to
\mathcal G(\overline{\mathbf H};\theta)
=
\mathcal G(\overline{\mathbf H};\theta).
\]

Therefore, taking the limit in \eqref{eq:vi_subseq}, we obtain by \citep[Lemma 3.2.2 (Continuity of inner product)]{kreyszig1991introductory}
\[
\left\langle
\mathcal G(\overline{\mathbf{H}};\theta),\,
\mathbf Y-\overline{\mathbf{H}}
\right\rangle
\ge 0,
\qquad \forall \mathbf Y\in X(C).
\]
Hence $\overline{\mathbf{H}}$ solves the variational inequality
\(
\mathrm{VI}(X(C),\mathcal G(\cdot;\theta)).
\)

By Lemma~\ref{lem:ve_unique}, the variational inequality
$\mathrm{VI}(X(C),\mathcal G(\cdot;\theta))$ admits a unique solution,
namely $\mathbf{H}^*(C,\theta)$. Since $\overline{\mathbf{H}}$ is also a
solution, it follows that
\begin{equation}
\overline{\mathbf{H}}=\mathbf{H}^*(C,\theta).
\label{eq:hbar_equals_hstar}
\end{equation}

Thus every convergent subsequence of $\{\mathbf{H}_n\}$ converges to the same
limit $\mathbf{H}^*(C,\theta)$. Since the original sequence $\{\mathbf{H}_n\}$
is bounded, this implies that the whole sequence converges to
$\mathbf{H}^*(C,\theta)$, i.e.,
\[
\mathbf{H}_n \to \mathbf{H}^*(C,\theta).
\]

Because the sequence $\{(C_n,\theta_n)\}$ was arbitrary, the map
$(C,\theta)\mapsto \mathbf{H}^*(C,\theta)$ is continuous on $\pazocal Y$. This completes the proof of Lemma~\ref{lem:limit_vi}.
\end{proof}

\section{Step-by-Step Proof of Theorem~\ref{thm:global_bound_pi}}
\label{appendix:global_bound_pi}
We first derive the lower bound \(\nu_i\). 

For player \(i\), the realized profit~\eqref{eq:pi_star} is positive if and only if
\begin{equation}
\Pi_{i,\omega}^* > 0
\quad\Longleftrightarrow\quad
R_{i,\omega}^* > K_i^*
\label{eq:equivalence_pi_r_direct}
\end{equation}

We are interested in the Probability of Profit
\(
\mathbb P\bigl(\Pi_{i,\omega}^* >0\bigr) = \mathbb P\bigl(R_{i,\omega}^* >K_i^*\bigr).
\)
Since \(R_{i}^*\) is a nonnegative random variable with finite variance, we can apply the Paley--Zygmund inequality~\citep[Page~8]{kahane1985some}(see also \citep[Section 3]{ghosh2002probability}). For any \(\vartheta_i\in(0,1)\), it yields
\begin{align}
\mathbb P\bigl(
R_{i,\omega}^*>\vartheta_i\,\mathbb E_{\omega}[R_i^*]
\bigr)
\ge
\frac{
(1-\vartheta_i)^2\bigl(\mathbb E_{\omega}[R_i^*]\bigr)^2
}{
(1-\vartheta_i)^2\bigl(\mathbb E_{\omega}[R_i^*]\bigr)^2
+
\mathrm{Var}_{\omega}(R_i^*)
}
\label{eq:PZ_direct_pi}
\end{align}

Define
\(
\vartheta_i
:=
\frac{K_i^*}{\mathbb E_{\omega}[R_i^*]}.
\)
By assumptions, we have \(K_i^*<\mathbb E_{\omega}[R_i^*]\), thus \(\vartheta_i\in(0,1)\).
Substituting this value of \(\vartheta_i\) into \eqref{eq:PZ_direct_pi}, we get
\begin{align}
\mathbb P\bigl(
R_{i,\omega}^*>K_i^*
\bigr) \ge
\frac{
\Bigl(1-\frac{K_i^*}{\mathbb E_{\omega}[R_i^*]}\Bigr)^2
\bigl(\mathbb E_{\omega}[R_i^*]\bigr)^2
}{
\Bigl(1-\frac{K_i^*}{\mathbb E_{\omega}[R_i^*]}\Bigr)^2
\bigl(\mathbb E_{\omega}[R_i^*]\bigr)^2
+
\mathrm{Var}_{\omega}(R_i^*)
}
\label{eq:first_ineq_direct_pi}
\end{align}

Noting that
\(
\Bigl(1-\frac{K_i^*}{\mathbb E_{\omega}[R_i^*]}\Bigr)^2
\bigl(\mathbb E_{\omega}[R_i^*]\bigr)^2
=
\bigl(\mathbb E_{\omega}[R_i^*]-K_i^*\bigr)^2,
\)
we can rewrite \eqref{eq:first_ineq_direct_pi} as
\[
\mathbb P\bigl(
R_{i,\omega}^*>K_i^*
\bigr)
\ge
\frac{
\bigl(\mathbb E_{\omega}[R_i^*]-K_i^*\bigr)^2
}{
\bigl(\mathbb E_{\omega}[R_i^*]-K_i^*\bigr)^2
+
\mathrm{Var}_{\omega}(R_i^*)
}
\]

Dividing the numerator and denominator by
\(
\bigl(\mathbb E_{\omega}[R_i^*]-K_i^*\bigr)^2,
\)
we obtain
\[
\mathbb P\bigl(
R_{i,\omega}^*>K_i^*
\bigr)
\ge
\frac{1}{
1+
\dfrac{
\mathrm{Var}_{\omega}(R_i^*)
}{
\bigl(\mathbb E_{\omega}[R_i^*]-K_i^*\bigr)^2
}
} :=\nu_i
\]

Using again \eqref{eq:equivalence_pi_r_direct}, we conclude that
\begin{equation}
    \mathbb P\bigl(\Pi_{i,\omega}^*>0\bigr)\ge \nu_i
\label{eq:case_1}
\end{equation}

We now assume that \(\alpha_i<\bar\alpha_i\). By the definition of
\(\bar\alpha_i\), we have\footnote{Note that CvaR$_\alpha$ is nondecreasing in $\alpha$.}
\[
\Gamma_i^*:=\sum_{t\in\pazocal T}
\mathrm{CVaR}_{\alpha_i}
\bigl(L_i^t(\mathbf h^{t*};\theta^*)\bigr)<0.
\]
Define the total realized loss
\(
L_{i,\omega}^*
:=
-\Pi_{i,\omega}^*
=
K_i^*-R_{i,\omega}^*.
\)
The total loss is the sum of the losses over the investment period:
\(
L_{i,\omega}^*
=
\sum_{t\in\pazocal T}L_{i,\omega}^t(\mathbf h^{t*},\theta^*),
\)
and the corresponding random loss variable is
\(
L_i^*
\).

By subadditivity of CVaR (see \citep[P.15]{pflug2000some}),
\begin{equation}
\mathrm{CVaR}_{\alpha_i}\!\bigl(L_i^*\bigr)
=
\mathrm{CVaR}_{\alpha_i}\!\Bigl(\sum_{t\in\pazocal T}L_i^t(\mathbf h^{t*};\theta^*)\Bigr)
\le
\sum_{t\in\pazocal T}\mathrm{CVaR}_{\alpha_i}\!\bigl(L_i^t(\mathbf h^{t*};\theta^*)\bigr)
<0
\end{equation}
Recall that the Value-at-Risk at level \(\alpha_i\) is defined as \citep[Eq.~(28),(29)]{li2022risk},
\(
\mathrm{VaR}_{\alpha_i}(L_i^*)
=
\min \left\{x\in\mathbb R:\mathbb P(L_i^*\le x)\ge \alpha_i\right\}
\).
By this definition, the quantity $\mathrm{VaR}_{\alpha_i}(L_i^*)$ is a threshold such that the loss does not exceed it with probability at least $\alpha_i$ (see~\citep{rockafellar2000optimization}), i.e.,
\begin{equation}
    \mathbb P\bigl(L_{i,\omega}^*\le \mathrm{VaR}_{\alpha_i}(L_i^*)\bigr)\ge \alpha_i.
\label{eq:p_alpha}
\end{equation}
Moreover, it is well known that
\(
\mathrm{VaR}_{\alpha_i}(L_i^*)\le \mathrm{CVaR}_{\alpha_i}(L_i^*)
\) (see \citep[Prop.~4]{pflug2000some}).
Hence,
\[
\bigl\{\omega \in \Omega:L_{i,\omega}^*\le \mathrm{VaR}_{\alpha_i}(L_i^*)\bigr\}
\subseteq
\bigl\{\omega \in \Omega:L_{i,\omega}^*\le \mathrm{CVaR}_{\alpha_i}(L_i^*)\bigr\},
\]
which implies
\[
\mathbb P\bigl(L_{i,\omega}^*\le \mathrm{CVaR}_{\alpha_i}(L_i^*)\bigr)\ge\mathbb P\bigl(L_{i,\omega}^*\le \mathrm{VaR}_{\alpha_i}(L_i^*)\bigr) \overset{\eqref{eq:p_alpha}}{\ge} \alpha_i
\]
Since \(\mathrm{CVaR}_{\alpha_i}(L_i^*)<0\), it follows that
\[
\bigl\{\omega \in \Omega:L_{i,\omega}^*\le \mathrm{CVaR}_{\alpha_i}(L_i^*)\bigr\}
\subseteq
\bigl\{\omega \in \Omega:L_{i,\omega}^*< 0\bigr\}
\]
Therefore,
\(
\mathbb P\bigl(L_{i,\omega}^*< 0\bigr)\ge \alpha_i.
\)

Finally, since
\(
L_{i,\omega}^*< 0
\quad\Longleftrightarrow\quad
\Pi_{i,\omega}^*> 0,
\) we conclude that
\begin{equation}
    \mathbb P\bigl(\Pi_{i,\omega}^*> 0\bigr)\ge \alpha_i
\label{eq:case_two}
\end{equation}
Combining \eqref{eq:case_two} with~\eqref{eq:case_1}, we obtain
\begin{equation}
\mathbb P\bigl(\Pi_{i,\omega}^*>0\bigr)
\ge \max\{\nu_i,\alpha_i\},
\qquad \text{if } \alpha_i<\bar\alpha_i
\label{eq:profit_bound_final_proof}
\end{equation}
This completes the proof.

\section{Step-by-Step Proof of Lemma \ref{lemma:delta_gamma_uniqueness}}\label{appendix:delta_gamma_uniqueness}
Under the firm-cost function~\eqref{eq:psi}, we can write
\(
\psi_i^t(\mathbf h^t)
=
q_i^t(h_i^t)
+
g_i^t(\mathbf h^t),
\)
where
\(
q_i^t(h_i^t)=\frac{\delta}{2}(h_i^t)^2,
\quad
g_i^t(\mathbf h^t)
=
\gamma h_i^t\sum_{j\in\pazocal N\setminus\{i\}}h_j^t .
\)
Therefore,
\[
(q_i^t)''(h_i^t)=\delta,
\quad
\forall i\in\pazocal N,\; t\in\pazocal T .
\]
By the definition of \(Q\) in~\eqref{eq:cond_q_g}, this gives
\(
Q=\delta .
\)

We now identify \(G\) \eqref{eq:cond_q_g}. For the interaction term \(g_i^t\), we have
\[
\frac{\partial^2 g_i^t(\mathbf h^t)}
{\partial h_i^t\partial h_j^t}
=
\gamma,
\quad j\neq i,
\qquad\text{and}\qquad
\frac{\partial^2 g_i^t(\mathbf h^t)}
{\partial (h_i^t)^2}
=
0 .
\]
Hence
\[
G=\gamma =\sup_{\substack{i,j\in\pazocal N,\; t\in\pazocal T,\\ \mathbf h^t\in\mathbb R_+^{|\pazocal N|}}}
\left|
\frac{\partial^2 g_i^t(\mathbf h^t)}{\partial h_i^t\,\partial h_j^t}
\right|
\]

Substituting \(Q=\delta\) and \(G=\gamma\) into the sufficient uniqueness
condition~\eqref{eq:ass_q_g} gives
\(
\delta>
\gamma\max_{i\in\pazocal N}\xi_i,
\)
which is exactly~\eqref{eq:delta_gamma_uniqueness}. Therefore, whenever
\eqref{eq:delta_gamma_uniqueness} holds, the sufficient condition in
Theorem~\ref{thm:gne_uniqueness} is satisfied. The follower variational
equilibrium is then unique by Theorem~\ref{thm:gne_uniqueness}.

\section{Step-by-Step Proof of Proposition~\ref{prop:follower_explicit}}
\label{appendix:follower_explicit}
We first characterize when a follower commits to a positive resource.

\begin{proposition}
\label{prop:follower_threshold}
Given \((C,\theta)\), follower \(i\) has a positive resource commitment iff
\[
    h_i^{t*}>0
\Longleftrightarrow
A_i^t:=\mathbb E[a_i^t]-\beta_i\,\mathrm{CVaR}_{\alpha_i}(-a_i^t)
>
(1+\beta_i) (\theta
+
\gamma\sum_{j\in\pazocal N\setminus\{i\}} h_j^{t*})
+
\lambda^{t*}
\]
where $A_i^t$ is the mean-CVaR benefit factor of follower \(i\) at time slot \(t\)
and \(\lambda^{t*}\) is the shadow price of the capacity constraint.
\end{proposition}
\begin{proof}
Substitute \eqref{eq:revenue}, \eqref{eq:price}, and \eqref{eq:psi} into \eqref{eq:profit_sp_realization} for a fixed time slot $t\in\pazocal T$, the realized profit of follower $i$ becomes:
\begin{align}
\Pi_{i,\omega}^t(\mathbf h^t;\theta)
&=
a_{i,\omega}^t\ln(1+h_i^t)
-\theta h_i^t
-\frac{\delta}{2}(h_i^t)^2
-\gamma\, h_i^t\sum_{j\in\pazocal N\setminus\{i\}}h_j^t,
\label{eq:profit_after_subs}
\end{align}
and the corresponding loss~\eqref{eq:loss_i_t} becomes
\begin{align}
L_{i,\omega}^t(\mathbf h^t;\theta)
&=
-a_{i,\omega}^t\ln(1+h_i^t)
+\theta h_i^t
+\frac{\delta}{2}(h_i^t)^2
+\gamma\, h_i^t\sum_{j\in\pazocal N\setminus\{i\}}h_j^t
\label{eq:loss_terms}
\end{align}
For fixed \(\mathbf h^t\), all terms in \eqref{eq:loss_terms} except \(-a_{i, \omega}^t\ln(1+h_i^t)\) are
deterministic. Hence, by translation invariance of CVaR
\citep[Prop.~2(i)]{pflug2000some}, and since \(h_i^t\ge0\), we have
\(\ln(1+h_i^t)\ge0\); therefore, by positive homogeneity of CVaR
\citep[Prop.~2(ii)]{pflug2000some}, we obtain
\begin{align}
\mathrm{CVaR}_{\alpha_i}\!\left(L_i^t(\mathbf h^t;\theta)\right)
&=
\mathrm{CVaR}_{\alpha_i}\!\left(
-a_i^t\ln(1+h_i^t)
+\theta h_i^t
+\frac{\delta}{2}(h_i^t)^2
+\gamma h_i^t\sum_{j\in\pazocal N\setminus\{i\}}h_j^t
\right)
\nonumber\\
&=
\theta h_i^t
+\frac{\delta}{2}(h_i^t)^2
+\gamma h_i^t\sum_{j\in\pazocal N\setminus\{i\}}h_j^t
+\ln(1+h_i^t)\,
\mathrm{CVaR}_{\alpha_i}(-a_i^t)
\label{eq:cvar_loss}
\end{align}
Moreover, \eqref{eq:cvar_loss} is continuously differentiable
with respect to \(h_i^t\) on \(\mathbb R_+\). Indeed, for fixed
\(\mathbf h_{-i}^t\), each term
\[
\theta h_i^t,\qquad
\frac{\delta}{2}(h_i^t)^2,\qquad
\gamma h_i^t\sum_{j\in\pazocal N\setminus\{i\}}h_j^t,\qquad
\ln(1+h_i^t)\mathrm{CVaR}_{\alpha_i}(-a_i^t)
\]
is continuously differentiable for \(h_i^t\ge 0\).

The utility~\eqref{eq:u_follower} of follower $i$ can be written
as
\begin{align}
U(\mathbf h_i, \mathbf h_{-i}; \theta) &:= \mathbb E_\omega[\Pi_i^t(\mathbf h^t;\theta)]
-
\beta_i\,\mathrm{CVaR}_{\alpha_i}\!\left(L_i^t(\mathbf h^t;\theta)\right)
\nonumber\\
&\overset{\eqref{eq:profit_after_subs},\,\eqref{eq:cvar_loss}}{=}
\left(
\mathbb E[a_i^t]
-\beta_i\,\mathrm{CVaR}_{\alpha_i}(-a_i^t)
\right)\ln(1+h_i^t)
\nonumber\\
&\quad
-b_i\theta h_i^t
-b_i\frac{\delta}{2}(h_i^t)^2
-b_i\gamma\, h_i^t\sum_{j\in\pazocal N\setminus\{i\}}h_j^t,
\label{eq:follower_reduced_form_setup}
\end{align}
where
\(
b_i:=1+\beta_i.
\)

Since the follower objective function at generalized Nash equilibrium problem \eqref{eq:utility_gne} has a
concave objective function~\eqref{eq:u_follower} and convex constraints \eqref{eq:constraints}, the
Karush--Kuhn--Tucker conditions~\citep[\S 5.5.2]{boyd2004convex} provide necessary and sufficient optimality
conditions.

For a fixed time slot $t$, the Lagrangian of follower $i$ is
\begin{align}
\pazocal L_i^t&=
\left(
A_i^t
\right)\ln(1+h_i^t)
-b_i\theta h_i^t
-b_i\frac{\delta}{2}(h_i^t)^2
-b_i\gamma\, h_i^t\sum_{j\in\pazocal N\setminus\{i\}}h_j^t
\nonumber\\
&\quad
+\lambda^t\left(C-\sum_{j\in\pazocal N}h_j^t\right)
+\mu_i^t h_i^t
\label{eq:lagrangian}
\end{align}
where \(\lambda^t\ge 0\) is the Lagrange multiplier associated with the
shared capacity constraint
\(\sum_{j\in\pazocal N}h_j^t\le C\), and \(\mu_i^t\ge 0\) is the
Lagrange multiplier associated with the nonnegativity constraint
\(h_i^t\ge 0\).

Differentiating~\eqref{eq:lagrangian} with respect to $h_i^t$ gives
\begin{align}
\frac{\partial \pazocal L_i^t}{\partial h_i^t}
&=
\frac{A_i^t}{1+h_i^t}
-b_i\theta
-b_i\delta h_i^t
-b_i\gamma\sum_{j\in\pazocal N\setminus\{i\}}h_j^t
-\lambda^t
+\mu_i^t
\label{eq:derivative_lagrangian}
\end{align}
where $A_i^t:=\mathbb E[a_i^t]-\beta_i\,\mathrm{CVaR}_{\alpha_i}(-a_i^t)$.

Hence the KKT conditions for the VE at time slot $t$ are:
\paragraph{Stationarity}
For every $i\in\pazocal N$,
\begin{align}
\eqref{eq:derivative_lagrangian}
=0
\label{eq:stationarity_special_case}
\end{align}

\paragraph{Primal feasibility}
\begin{align}
h_i^t
&\ge 0
\qquad \forall i\in\pazocal N,\nonumber\\
\sum_{j\in\pazocal N} h_j^t
&\le C
\label{eq:primal}
\end{align}

\paragraph{Dual feasibility}
\begin{align}
\lambda^t
&\ge 0,\nonumber\\
\mu_i^t
&\ge 0,
\qquad \forall i\in\pazocal N
\label{eq:mu_non_neg}
\end{align}

\paragraph{Complementary slackness}
\begin{align}
\lambda^t\left(C-\sum_{j\in\pazocal N} h_j^t\right)
&=0,\qquad
\mu_i^t h_i^t=0,
\qquad \forall i\in\pazocal N
\label{eq:complementary}
\end{align}
Let $(\mathbf h^{t*},\mu^{t*},\lambda^{t*})$ be an equilibrium.

First suppose that $h_i^{t*}>0$. By~\eqref{eq:complementary},
$\mu_i^{t*}=0$. Substituting $\mu_i^{t*}=0$ into
\eqref{eq:stationarity_special_case} gives
\begin{equation}
\frac{A_i^t}{1+h_i^{t*}}
=
b_i\theta
+
b_i\delta h_i^{t*}
+
b_i\gamma\sum_{j\in\pazocal N\setminus\{i\}} h_j^{t*}
+
\lambda^{t*}
\label{eq:a_over_h}
\end{equation}
Since $h_i^{t*}>0$ and $\delta>0$ (by definition, see \eqref{eq:psi}), the right-hand side of \eqref{eq:a_over_h} is strictly larger than
\(
b_i\theta
+
b_i\gamma\sum_{j\in\pazocal N\setminus\{i\}} h_j^{t*}
+
\lambda^{t*}.
\)
Multiplying~\eqref{eq:a_over_h} by $1+h_i^{t*}>1$ yields
\begin{align}
A_i^t
&=
(1+h_i^{t*})
\left(
b_i\theta
+
b_i\delta h_i^{t*}
+
b_i\gamma\sum_{j\in\pazocal N\setminus\{i\}} h_j^{t*}
+
\lambda^{t*}
\right)
\nonumber\\&>
b_i\theta
+
b_i\gamma\sum_{j\in\pazocal N\setminus\{i\}} h_j^{t*}
+
\lambda^{t*}
\end{align}
Therefore,
\(
h_i^{t*}>0
\quad\Longrightarrow\quad
A_i^t
>
b_i\theta
+
b_i\gamma\sum_{j\in\pazocal N\setminus\{i\}} h_j^{t*}
+
\lambda^{t*}
\).
Conversely, assume that
\(
A_i^t
>
b_i\theta
+
b_i\gamma\sum_{j\in\pazocal N\setminus\{i\}} h_j^{t*}
+
\lambda^{t*}
\)
If, by contradiction, $h_i^{t*}=0$, then
\eqref{eq:stationarity_special_case} becomes
\(
A_i^t
-
b_i\theta
-
b_i\gamma\sum_{j\in\pazocal N\setminus\{i\}} h_j^{t*}
-
\lambda^{t*}
+
\mu_i^{t*}
=0,
\)
that is,
\(
\mu_i^{t*}
=
b_i\theta
+
b_i\gamma\sum_{j\in\pazocal N\setminus\{i\}} h_j^{t*}
+
\lambda^{t*}
-
A_i^t
<0,
\)
which contradicts dual feasibility $\mu_i^{t*}\ge 0$~\eqref{eq:mu_non_neg}. Hence $h_i^{t*}\neq 0$,
and since $h_i^{t*}\ge 0$, it follows that $h_i^{t*}>0$.
Thus,
\[
h_i^{t*}>0
\quad\Longleftrightarrow\quad
A_i^t
>
b_i\theta
+
b_i\gamma\sum_{j\in\pazocal N\setminus\{i\}} h_j^{t*}
+
\lambda^{t*}
\]
Taking the negation gives
\(
h_i^{t*}=0
\quad\Longleftrightarrow\quad
A_i^t
\le
b_i\theta
+
b_i\gamma\sum_{j\in\pazocal N\setminus\{i\}} h_j^{t*}
+
\lambda^{t*}
\)
This proves Proposition~\ref{prop:follower_threshold}.
\end{proof}
Fix \(i\in\pazocal N\). By the characterization established in Proposition~\ref{prop:follower_threshold}, we have
\begin{equation}
h_i^{t*}>0
\quad\Longleftrightarrow\quad
A_i^t
>
(1+ \beta_i) (\theta
+ \gamma\sum_{j\in\pazocal N\setminus\{i\}} h_j^{t*})
+
\lambda^{t*}
\label{eq:positivity_threshold_hi}
\end{equation}
First define 
\begin{equation}
\kappa_i^t(x)
:=
(1+ \beta_i) (\theta
+
\gamma\sum_{j\in\pazocal N\setminus\{i\}} h_j^{t*}
+ x) + 
\lambda^{t*},
\qquad x\in\mathbb R
\label{eq:def_kappa}
\end{equation}
We distinguish between two cases.

\paragraph{ \textbf{Case 1: \(\boldsymbol{h_i^{t*}=0}\)}}
In this case, \eqref{eq:explicit_hi_star} holds if and only if the expression
inside the maximum is nonpositive. Recall that $\beta_i \ge 0, \quad \forall i \in \pazocal N$ and let
\begin{equation}
    b_i=1+\beta_i >0
\label{eq:b_i_is_pos}
\end{equation} and
\begin{equation}
B_i^t
:=
b_i\theta
+
b_i\gamma\sum_{j\in\pazocal N\setminus\{i\}} h_j^{t*}
+
\lambda^{t*}
\label{eq:def_Bit}
\end{equation}
Note that \(B_i^t\ge 0\). Indeed, \(b_i>0\) by~\eqref{eq:b_i_is_pos},
\(\theta\ge 0\) by~\eqref{eq:leader_set}, \(\gamma\ge 0\) by~\eqref{eq:psi},
\(h_j^{t*}\ge 0\) for all \(j\in\pazocal N\setminus\{i\}\) by primal
feasibility~\eqref{eq:primal}, and \(\lambda^{t*}\ge 0\) by dual
feasibility~\eqref{eq:mu_non_neg}. Hence every term in~\eqref{eq:def_Bit}
is nonnegative.

\eqref{eq:def_kappa} becomes
\begin{equation}
\kappa_i^t(\delta)=B_i^t+b_i\delta,
\qquad
\kappa_i^t(-\delta)=B_i^t-b_i\delta
\label{eq:kappa_pm_delta}
\end{equation}
The numerator in \eqref{eq:explicit_hi_star} is nonpositive if and only if
\begin{equation}
\sqrt{(\kappa_i^t(-\delta))^2+4b_i\delta A_i^t}
\le \kappa_i^t(\delta)
\label{eq:sqrt_condition}
\end{equation}
Since $B_i^t \ge 0$ \eqref{eq:def_Bit}  \(b_i>0\) (see \eqref{eq:b_i_is_pos} and \(\delta>0\) (by definition, see \eqref{eq:psi}), it follows that \(\kappa_i^t(\delta)=B_i^t+b_i\delta>0\), squaring both sides of
\eqref{eq:sqrt_condition} yields
\begin{equation}
(\kappa_i^t(-\delta))^2+4b_i\delta A_i^t
\le
(\kappa_i^t(\delta))^2
\label{eq:squared_condition}
\end{equation}
Using
\begin{equation}
(\kappa_i^t(\delta))^2-(\kappa_i^t(-\delta))^2
=
(B_i^t+b_i\delta)^2-(B_i^t-b_i\delta)^2
=
4b_i\delta B_i^t,
\label{eq:difference_of_squares}
\end{equation}
we obtain from \eqref{eq:squared_condition} and
\eqref{eq:difference_of_squares} that
\begin{equation}
4b_i\delta A_i^t\le 4b_i\delta B_i^t
\label{eq:before_division}
\end{equation}
Again, since \(b_i>0\) and \(\delta>0\), dividing both sides of
\eqref{eq:before_division} by \(4b_i\delta\) gives
\begin{equation}
A_i^t\le B_i^t,
\label{eq:A_leq_B}
\end{equation}
that is, by \eqref{eq:def_Bit},
\begin{equation}
A_i^t
\le
b_i\theta
+
b_i\gamma\sum_{j\in\pazocal N\setminus\{i\}} h_j^{t*}
+
\lambda^{t*}.
\label{eq:A_leq_threshold}
\end{equation}
Therefore, $h_i^{t*}=0$ holds if \eqref{eq:A_leq_threshold} holds.

Therefore, in this case \eqref{eq:explicit_hi_star} holds.

\paragraph{\textbf{Case 2: }\(\boldsymbol{h_i^{t*}>0}\)}
By~\eqref{eq:complementary} \(\mu_i^{t*}=0\). Then the stationarity condition
\eqref{eq:stationarity_special_case} becomes
\begin{equation}
\frac{A_i^t}{1+h_i^{t*}}
=
b_i\theta
+
b_i\delta h_i^{t*}
+
b_i\gamma\sum_{j\in\pazocal N\setminus\{i\}} h_j^{t*}
+
\lambda^{t*}
\label{eq:stationarity_positive_hi}
\end{equation}
Using~\eqref{eq:def_Bit}
we rewrite \eqref{eq:stationarity_positive_hi} as
\(
\frac{A_i^t}{1+h_i^{t*}}
=
B_i^t+b_i\delta h_i^{t*}
\)
Multiplying by \(1+h_i^{t*}\) gives
\[
A_i^t
=
(1+h_i^{t*})(B_i^t+b_i\delta h_i^{t*}),
\]
\[\iff\]
\begin{equation}
    A_i^t
=
B_i^t+(B_i^t+b_i\delta)h_i^{t*}+b_i\delta (h_i^{t*})^2
\label{eq:a_it}
\end{equation}
Rearranging~\eqref{eq:a_it}, we obtain the quadratic equation
\begin{equation}
b_i\delta (h_i^{t*})^2
+
(B_i^t+b_i\delta)h_i^{t*}
+
(B_i^t-A_i^t)
=
0
\label{eq:quadratic_hi_star}
\end{equation}
Since \(b_i>0\) and \(\delta>0\), $b_i \delta \neq 0$. Thus, the quadratic formula~\eqref{eq:quadratic_hi_star} yields
\begin{equation}
h_i^{t*}
=
\frac{
-(B_i^t+b_i\delta)
\pm
\sqrt{(B_i^t+b_i\delta)^2-4b_i\delta(B_i^t-A_i^t)}
}{
2b_i\delta
}
\label{eq:quadratic_formula_hi}
\end{equation}
The discriminant simplifies as
\begin{align}
&(B_i^t+b_i\delta)^2-4b_i\delta(B_i^t-A_i^t)=
(B_i^t-b_i\delta)^2+4b_i\delta A_i^t
\overset{\eqref{eq:kappa_pm_delta}}{=}
(\kappa_i^t(-\delta))^2+4b_i\delta A_i^t
\label{eq:discriminant_simplified_hi}
\end{align}
Moreover by~\eqref{eq:kappa_pm_delta}
\begin{equation}
    B_i^t+b_i\delta=\kappa_i^t(\delta)
\label{eq:kappa_delta}
\end{equation}
Substituting~\eqref{eq:discriminant_simplified_hi} and \eqref{eq:kappa_delta} into \eqref{eq:quadratic_formula_hi}, we obtain
\begin{equation}
    h_i^{t*}
=
\frac{
-\kappa_i^t(\delta)
\pm
\sqrt{(\kappa_i^t(-\delta))^2+4b_i\delta A_i^t}
}{
2b_i\delta
}
\label{eq:both_root}
\end{equation}
Since \(b_i\delta>0\), the denominator in~\eqref{eq:both_root} is positive. 
Moreover, the root with the minus sign has a negative numerator and is therefore infeasible for an active solution \(h_i^{t*}>0\). 
Thus, the admissible active solution is given by the larger root. 
Hence
\begin{equation}
    h_i^{t*}
=
\frac{
-\kappa_i^t(\delta)
+
\sqrt{(\kappa_i^t(-\delta))^2+4b_i\delta A_i^t}
}{
2b_i\delta
}
\label{eq:pos_root}
\end{equation}
To see that this larger root \eqref{eq:pos_root} is positive, note that the stationarity condition \eqref{eq:positivity_threshold_hi} gives
\(
A_i^t>b_i\theta+b_i\gamma\sum_{j\neq i}h_j^{t*}+\lambda^{t*}.
\)
Hence,
\[
\begin{aligned}
(\kappa_i^t(-\delta))^2+4b_i\delta A_i^t &= \Bigl(b_i\theta+b_i\gamma\sum_{j\neq i}h_j^{t*}+\lambda^{t*}-b_i\delta\Bigr)^2
+4b_i\delta A_i^t \\
&\quad>
\Bigl(b_i\theta+b_i\gamma\sum_{j\neq i}h_j^{t*}+\lambda^{t*}-b_i\delta\Bigr)^2
+4b_i\delta
\Bigl(b_i\theta+b_i\gamma\sum_{j\neq i}h_j^{t*}+\lambda^{t*}\Bigr)\\
&\quad=
\Bigl(b_i\theta+b_i\gamma\sum_{j\neq i}h_j^{t*}+\lambda^{t*}+b_i\delta\Bigr)^2 =  (\kappa_i^t(\delta))^2
\end{aligned}
\]
Thus, the numerator of the larger root \eqref{eq:pos_root} is strictly positive. Since its denominator is positive, the larger root is strictly positive.
Combining~\eqref{eq:pos_root} with the case \(h_i^{t*}=0\) yields
\eqref{eq:explicit_hi_star} and completes the proof.

\paragraph{\textbf{Shadow price explanation}}\label{appendix:eta_characterization}
\(\lambda^{t*}\) is the shadow price of the shared capacity constraint in time slot \(t\) indicating how valuable an additional unit of capacity would be at the optimum. It is equal to zero when total equilibrium resource commitment is strictly below available capacity, and becomes positive when total equilibrium resource commitment exactly reaches the capacity limit, that is, (see \eqref{eq:primal}, \eqref{eq:mu_non_neg}, \eqref{eq:complementary})
\begin{equation}
\lambda^{t*}\ge 0,\qquad
\sum_{j\in\pazocal N} h_j^{t*}\le C,
\qquad
\lambda^{t*}\left(C-\sum_{j\in\pazocal N} h_j^{t*}\right)=0
\label{eq:eta_compl}
\end{equation}

Note that
\(
\lambda^{t*}\left(C-\sum_{j\in\pazocal N} h_j^{t*}\right)=0.
\)
If \(\sum_{j\in\pazocal N} h_j^{t*}<C\), then
\[
C-\sum_{j\in\pazocal N} h_j^{t*}>0
\]
and therefore necessarily \(\lambda^{t*}=~0\). Conversely, if \(\lambda^{t*}>0\), then
\[
C-\sum_{j\in\pazocal N} h_j^{t*}=~0
\]
that is,
\(
\sum_{j\in\pazocal N} h_j^{t*}=C.
\)
This means the shadow price $\lambda^{t*}$ is equal to zero when total resources are strictly below available capacity and becomes positive when total resources exactly reach the capacity limit.

\section{Comparative Statics of Proposition \ref{prop:follower_explicit}}
\label{appendix:follower_comparative_statics}

\begin{proposition}
\label{prop:follower_comparative_statics}
Fix \(t\in\pazocal T\) and \(i\in\pazocal N\) with \(h_i^{t*}>0\), and let \(s_{-i}^{t*}:=\sum_{j\neq i}h_j^{t*}.\) Conditional on \(s_{-i}^{t*}\), the resource commitment \(h_i^{t*}\) is increasing in \(A_i^t\), decreasing in \(\theta\), \(\delta\), and \(\lambda^{t*}\), and nonincreasing in \(\gamma\), with strict decrease in \(\gamma\) whenever \(s_{-i}^{t*}>0\). Moreover, holding all primitive parameters fixed, \(h_i^{t*}\) is nonincreasing in \(s_{-i}^{t*}\), and strictly decreasing whenever \(\gamma>0\).\end{proposition}

\begin{proof}
We derive the comparative statics from the stationarity condition~\eqref{eq:cs_stationarity_active}.
Fix \(t\in\pazocal T\) and \(i\in\pazocal N\). Let
\begin{equation}
    s_{-i}^{t*}:=\sum_{j\in\pazocal N\setminus\{i\}} h_j^{t*}
\label{eq:big_h}
\end{equation}
denote the total equilibrium resource commitment of the other firms in time slot~\(t\).

Since \(h_i^{t*}>0\), the complementary-slackness condition
\eqref{eq:complementary} implies that \(\mu_i^{t*}=0\). Therefore, the
stationarity condition~\eqref{eq:stationarity_special_case} reduces to
\begin{equation}
\frac{A_i^t}{1+h_i^{t*}}
-b_i\theta
-b_i\delta h_i^{t*}
-b_i\gamma s_{-i}^{t*}
-\lambda^{t*}
=0
\label{eq:cs_stationarity_active}
\end{equation}
To study how \(h_i^{t*}\) varies with the model parameters, define
\begin{equation}
F_i(h;A_i^t,\theta,\delta,\gamma,s_{-i}^{t*},\lambda^{t*})
:=
\frac{A_i^t}{1+h}
-b_i\theta
-b_i\delta h
-b_i\gamma s_{-i}^{t*}
-\lambda^{t*}
\label{eq:def_Fi_cs}
\end{equation}
Then \eqref{eq:cs_stationarity_active} can be written as
\begin{equation}
F_i(h_i^{t*};A_i^t,\theta,\delta,\gamma,s_{-i}^{t*},\lambda^{t*})=0
\label{eq:F_zero_cs}
\end{equation}
For simplicity we use $F_i:=F_i(h;A_i^t,\theta,\delta,\gamma,s_{-i}^{t*},\lambda^{t*})$.
Differentiating \eqref{eq:def_Fi_cs} with respect to
\(h\), we obtain
\begin{equation}
\frac{\partial F_i}{\partial h}
=
-\frac{A_i^t}{(1+h)^2}-b_i\delta
\label{eq:dFdh_cs}
\end{equation}
Evaluating at \(h=h_i^{t*}\) gives
\begin{equation}
    \frac{\partial F_i}{\partial h_i^t}
=
-\frac{A_i^t}{(1+h_i^{t*})^2}-b_i\delta<0
\label{eq:fh}
\end{equation}
Indeed, since \(h_i^{t*}>0\), the stationarity condition
\eqref{eq:cs_stationarity_active} implies \(A_i^t>0\), and therefore both terms
on the right-hand side of~\eqref{eq:fh} are nonpositive, with the second one being strictly
negative because \(b_i>0\) and \(\delta>0\). Hence
\(\frac{\partial F_i}{\partial h}\neq 0\).

For convenience, define
\begin{equation}
M_i^t
:=
\frac{A_i^t}{(1+h_i^{t*})^2}+b_i\delta>0
\label{eq:def_Mi_cs}
\end{equation}
Then \(\frac{\partial F_i}{\partial h}(h_i^{t*})=-M_i^t\). Therefore 
\(
\frac{\partial h_i^{t*}}{\partial z}
=
-\frac{\partial F_i/\partial z}{\partial F_i/\partial h}
\qquad
\text{for each parameter } z
\)

We now compute the derivatives one by one.

\paragraph{\textbf{Dependence on \(A_i^t\)}}
From \eqref{eq:def_Fi_cs},
\begin{equation}
    \frac{\partial F_i}{\partial A_i^t}
=
\frac{1}{1+h_i^{t*}}
\label{eq:def_f_a}
\end{equation}
Therefore, by~\eqref{eq:fh}, \eqref{eq:def_Mi_cs} and~\eqref{eq:def_f_a}
\begin{equation}
\frac{\partial h_i^{t*}}{\partial A_i^t}
=
-\frac{\partial F_i/\partial A_i^t}{\partial F_i/\partial h}
=
\frac{1}{(1+h_i^{t*})M_i^t}
>0.
\label{eq:dhdA_cs}
\end{equation}
Thus, follower \(i\)'s equilibrium resource commitment is increasing in its own
mean-CVaR benefit~\(A_i^t\).

\paragraph{\textbf{Dependence on \(\theta\)}}
From \eqref{eq:def_Fi_cs}
\begin{equation}
    \frac{\partial F_i}{\partial \theta}=-b_i
    \label{eq:f_theta}
\end{equation}
By~\eqref{eq:f_theta} and \eqref{eq:fh}, we obtain
\begin{equation}
\frac{\partial h_i^{t*}}{\partial \theta}
=
-\frac{\partial F_i/\partial \theta}{\partial F_i/\partial h}
=
-\frac{b_i}{M_i^t}
<0
\label{eq:dhdtheta_cs}
\end{equation}
Hence \(h_i^{t*}\) is decreasing in the access price \(\theta\).

\paragraph{\textbf{Dependence on \(\delta\)}}
From \eqref{eq:def_Fi_cs},
\begin{equation}
    \frac{\partial F_i}{\partial \delta}
=
-b_i h_i^{t*}
\label{eq:f_delta}
\end{equation}
Therefore by~\eqref{eq:f_delta} and \eqref{eq:fh},
\begin{equation}
\frac{\partial h_i^{t*}}{\partial \delta}
=
-\frac{\partial F_i/\partial \delta}{\partial F_i/\partial h}
=
-\frac{b_i h_i^{t*}}{M_i^t}
<0
\label{eq:dhddelta_cs}
\end{equation}
Thus the equilibrium resource commitment is decreasing in the operational cost
parameter \(\delta\).

\paragraph{\textbf{Dependence on \(\gamma\)}}
From~\eqref{eq:def_Fi_cs}
\begin{equation}
    \frac{\partial F_i}{\partial \gamma}
=
-b_i s_{-i}^{t*}
\label{eq:f_gamma}
\end{equation}
By \eqref{eq:f_gamma} and \eqref{eq:fh}, we obtain
\begin{equation}
\frac{\partial h_i^{t*}}{\partial \gamma}
=
-\frac{\partial F_i/\partial \gamma}{\partial F_i/\partial h}
=
-\frac{b_i s_{-i}^{t*}}{M_i^t}
\le 0
\label{eq:dhdgamma_cs}
\end{equation}
Hence \(h_i^{t*}\) is nonincreasing in the congestion cost parameter
\(\gamma\), and it is strictly decreasing whenever \(s_{-i}^{t*}>0\).

\paragraph{\textbf{Dependence on \(\lambda^{t*}\)}}
From~\eqref{eq:def_Fi_cs}
\begin{equation}
    \frac{\partial F_i}{\partial \lambda^{t*}}=-1,
\label{eq:f_lambda}
\end{equation}
By~\eqref{eq:f_lambda} and \eqref{eq:fh}, we obtain
\begin{equation}
\frac{\partial h_i^{t*}}{\partial \lambda^{t*}}
=
-\frac{\partial F_i/\partial \lambda^{t*}}{\partial F_i/\partial h}
=
-\frac{1}{M_i^t}
<0
\label{eq:dhdeta_cs}
\end{equation}
Thus the equilibrium resource commitment is decreasing in the shadow price
\(\lambda^{t*}\).

\paragraph{\textbf{Dependence on \(s_{-i}^{t*}\)}}
Finally, from~\eqref{eq:def_Fi_cs}
\begin{equation}
    \frac{\partial F_i}{\partial s_{-i}^{t*}}
=
-b_i\gamma,
\label{eq:f_big_h}
\end{equation}
we get from \eqref{eq:f_big_h} and \eqref{eq:fh}
\begin{equation}
\frac{\partial h_i^{t*}}{\partial s_{-i}^{t*}}
=
-\frac{\partial F_i/\partial s_{-i}^{t*}}{\partial F_i/\partial h}
=
-\frac{b_i\gamma}{M_i^t}
\le 0
\label{eq:dhdS_cs}
\end{equation}
Therefore, holding the primitive parameters fixed, follower \(i\)'s
equilibrium resource commitment is nonincreasing in the total resource commitment of the other
firms, and it is strictly decreasing whenever \(\gamma>0\).

Collecting \eqref{eq:dhdA_cs}--\eqref{eq:dhdS_cs}, we complete the proof.
\end{proof}

\section{Step-by-Step Proof of Proposition~\ref{prop:computational_tractability}}
\label{appendix:computational_tractability}
We first establish a Lemma needed in the proof.
\begin{lemma}
\label{lemma:delta_greater_gamma}
If~\eqref{eq:delta_gamma_uniqueness} holds, then
\(\delta>\gamma\).
\end{lemma}

\begin{proof}
By~\eqref{eq:delta_gamma_uniqueness},
\(
\delta>\gamma\max_{i\in\pazocal N}\xi_i .
\)
It remains to show that \(\max_{i\in\pazocal N}\xi_i\ge 1\). From
\eqref{eq:ass_q_g}, for every \(i\in\pazocal N\),
\[
\xi_i
=
\frac{
2|\pazocal N|+|\pazocal N|\beta_i+\sum_{j\in\pazocal N}\beta_j
}{
2(1+\beta_i)
}.
\]
Hence,
\begin{align}
\xi_i-1
&=
\frac{
2(|\pazocal N|-1)
+
(|\pazocal N|-2)\beta_i
+
\sum_{j\in\pazocal N}\beta_j
}{
2(1+\beta_i)
}
\nonumber \\
&=
\frac{
2(|\pazocal N|-1)
+
(|\pazocal N|-1)\beta_i
+
\sum_{j\in\pazocal N\setminus\{i\}}\beta_j
}{
2(1+\beta_i)
}
\label{eq:xi_minus_one}
\end{align}
Since \(|\pazocal N|\ge 1\) and \(\beta_j\ge 0\) for all
\(j\in\pazocal N\), the numerator and the denominator of \eqref{eq:xi_minus_one} are nonnegative. Therefore,
\(
\xi_i\ge 1,
\quad \forall i\in\pazocal N,
\)
and consequently
\(
\max_{i\in\pazocal N}\xi_i\ge 1.
\)
Using \(\gamma\ge 0\), we obtain
\(
\gamma\max_{i\in\pazocal N}\xi_i\ge \gamma.
\)
Combining this inequality with~\eqref{eq:delta_gamma_uniqueness} gives
\(
\delta>\gamma.
\)
\end{proof}

We now decompose the proof of Proposition \ref{prop:computational_tractability} into three parts:

\subsection{Follower computation}
\label{subsec:follower_computation}

Fix \(t\in\pazocal T\) and a leader decision
\((C,\theta)\in\pazocal Y\). We first introduce the notation used in the proof.
Let
\begin{equation}
    s^t:=\sum_{j\in\pazocal N} h_j^{t*}
    \label{eq:s_t}
\end{equation}
denote the total resource commitment in time slot \(t\). Then, for each
\(i\in\pazocal N\),
\begin{equation}
    \sum_{j\in\pazocal N\setminus\{i\}} h_j^{t*}=s^t-h_i^{t*}.
\label{eq:st_h}
\end{equation}
Substituting \eqref{eq:st_h} into \eqref{eq:explicit_hi_star} with \begin{equation}
B_i^t(s,\lambda):=b_i\theta+b_i\gamma s+\lambda .
\label{eq:b_s_lambda}
\end{equation}
We obtain, for any fixed pair \((s,\lambda)\),  
\begin{align}
\varphi_i^t(s,\lambda)
&=
\max\Biggl\{
0,\,
\frac{
-\bigl(B_i^t(s,\lambda)+b_i(\delta-\gamma)\bigr)
+\Delta_i^t(s,\lambda)
}{
2b_i(\delta-\gamma)
}
\Biggr\},
\label{eq:phi_s_lambda}
\\[1mm]
\Delta_i^t(s,\lambda)
&:=
\sqrt{
\bigl(B_i^t(s,\lambda)-b_i(\delta-\gamma)\bigr)^2
+4b_i(\delta-\gamma)A_i^t
}
\label{eq:delta_s_lambda}
\end{align}
Thus, once \(s^t\) and \(\lambda^{t*}\) are known, the equilibrium resource
commitments are recovered from
\begin{equation}
h_i^{t*}=\varphi_i^t(s^t,\lambda^{t*}),
\qquad i\in\pazocal N .
\label{eq:h_phi}
\end{equation}

Define the scalar residuals
\begin{align}
R_t(s;\lambda)
&:=
s-\sum_{i\in\pazocal N}\varphi_i^t(s,\lambda),
\label{eq:R_residual}
\\
Q_t^C(\lambda)
&:=
C-\sum_{i\in\pazocal N}\varphi_i^t(C,\lambda).
\label{eq:Q_residual}
\end{align}

The computation of the follower VE is summarized in
Algorithm~\ref{alg:follower_computation}.

\begin{algorithm}[htbp]
\caption{Computation of the VE at each time slot $t$.}
\label{alg:follower_computation}
\begin{algorithmic}[1]
\REQUIRE Leader decision \((C,\theta)\in\pazocal Y\), time slot \(t\in\pazocal T\)
\ENSURE \(s^t\), \(\lambda^{t*}\), and \(\mathbf h^{t*}(C,\theta)\)

\STATE Compute \(\widehat s^t\) as the solution of
\(
    R_t(s;0)=0
\) \eqref{eq:R_residual}.

\IF{\(\widehat s^t>C\)}
    \STATE Compute \(\widehat\lambda^t\) as the solution of
    \(
        Q_t^C(\lambda)=0
    \) \eqref{eq:Q_residual}.
\ENDIF

\STATE Set
\begin{align}
    s^t=
\begin{cases}
\widehat s^t, & \text{if } \widehat s^t\le C,\\
C, & \text{if } \widehat s^t>C,
\end{cases}
\qquad
\lambda^{t*}=
\begin{cases}
0, & \text{if } \widehat s^t\le C,\\
\widehat\lambda^t, & \text{if } \widehat s^t>C.
\end{cases}
\label{eq:ca}
\end{align}

\STATE Recover
\(  h_i^{t*}=\varphi_i^t(s^t,\lambda^{t*}),\quad i\in\pazocal N
\) \eqref{eq:phi_s_lambda}.
\RETURN \(s^t\), \(\lambda^{t*}\), and \(\mathbf h^{t*}(C,\theta)\).
\end{algorithmic}
\end{algorithm}

We next justify each step of Algorithm~\ref{alg:follower_computation}. The
first Lemma establishes the monotonicity of
\(\varphi_i^t\) \eqref{eq:phi_s_lambda}.
\begin{lemma}
\label{lemma:phi_monotone}
Fix \(t\in\pazocal T\) and \(i\in\pazocal N\). Assume that
\(\delta>\gamma\). Then the function
\(\varphi_i^t(s,\lambda)\) defined in \eqref{eq:phi_s_lambda} is
nonincreasing in \(s\) and nonincreasing in \(\lambda\).
\end{lemma}
\begin{proof}[Proof of Lemma~\ref{lemma:phi_monotone}]
Let
\begin{equation}
D_i:=b_i(\delta-\gamma)>0 .
\label{eq:d_pos}
\end{equation}
By \eqref{eq:b_s_lambda},
\begin{equation}
B_i^t(s,\lambda)=b_i\theta+b_i\gamma s+\lambda .
\label{eq:b_s_lambda_recall}
\end{equation}
Hence
\begin{equation}
\frac{\partial B_i^t(s,\lambda)}{\partial s}=b_i\gamma\ge0,
\qquad
\frac{\partial B_i^t(s,\lambda)}{\partial \lambda}=1>0 .
\label{eq:b_monotone}
\end{equation}

For \(B\ge0\), define
\begin{equation}
\chi_i^t(B)
:=
\frac{
-(B+D_i)
+
\sqrt{(B-D_i)^2+4D_iA_i^t}
}{
2D_i
}.
\label{eq:chi_def}
\end{equation}
Then \eqref{eq:phi_s_lambda} can be written as
\begin{equation}
\varphi_i^t(s,\lambda)
=
\left[\chi_i^t\!\left(B_i^t(s,\lambda)\right)\right]^+,
\qquad [x]^+:=\max\{x,0\}.
\label{eq:phi_chi_composition}
\end{equation}

We first identify when the positive part in
\eqref{eq:phi_chi_composition} is active. Since \(2D_i>0\) and
\(B+D_i>0\), we have
\begin{align}
\chi_i^t(B)>0
&\Longleftrightarrow
\sqrt{(B-D_i)^2+4D_iA_i^t}>B+D_i
\nonumber\\
&\Longleftrightarrow
(B-D_i)^2+4D_iA_i^t>(B+D_i)^2
\nonumber\\
&\Longleftrightarrow
A_i^t>B .
\label{eq:positive_part_active}
\end{align}
Therefore, the positive part is active exactly when \(B<A_i^t\). If
\(B\ge A_i^t\), then \([\chi_i^t(B)]^+=0\).

It remains to show that \([\chi_i^t(B)]^+\) is nonincreasing in \(B\).
On the region \(B<A_i^t\), since \(B\ge0\), we have \(A_i^t>0\). Therefore,
\(4D_iA_i^t>0\). Differentiating \eqref{eq:chi_def} with respect to \(B\)
gives
\begin{equation}
\frac{d\chi_i^t(B)}{dB}
=
\frac{
-1+
\dfrac{B-D_i}{
\sqrt{(B-D_i)^2+4D_iA_i^t}
}
}{
2D_i
}.
\label{eq:chi_derivative}
\end{equation}
Since \(4D_iA_i^t>0\), we have
\(
(B-D_i)^2+4D_iA_i^t>(B-D_i)^2.
\)
Taking square roots gives
\(
\sqrt{(B-D_i)^2+4D_iA_i^t}>|B-D_i|.
\)
Consequently,
\[
\frac{B-D_i}{
\sqrt{(B-D_i)^2+4D_iA_i^t}
}
<1.
\]
Indeed, if \(B-D_i\ge0\), then the denominator is strictly larger than the
numerator. If \(B-D_i<0\), then the fraction is negative and is therefore
strictly smaller than \(1\). Hence,
\[
-1+
\frac{B-D_i}{
\sqrt{(B-D_i)^2+4D_iA_i^t}
}
<0.
\]
Since \(2D_i>0\), we obtain, by \eqref{eq:chi_derivative},
\(
\frac{d\chi_i^t(B)}{dB}<0
\quad \text{for } B<A_i^t .
\)
Thus, \([\chi_i^t(B)]^+\) is strictly decreasing on the region
\(B<A_i^t\).

At the threshold \(B=A_i^t\), the left branch equals zero because
\[
(A_i^t-D_i)^2+4D_iA_i^t=(A_i^t+D_i)^2.
\]
Hence
\[
\chi_i^t(A_i^t)
=
\frac{
-(A_i^t+D_i)
+
\sqrt{(A_i^t+D_i)^2}
}{
2D_i
}
=0.
\]
For \(B\ge A_i^t\), we have \([\chi_i^t(B)]^+=0\). Therefore,
\([\chi_i^t(B)]^+\) is strictly decreasing before the threshold, reaches
zero at the threshold, and remains equal to zero after the threshold. Hence
\([\chi_i^t(B)]^+\) is nonincreasing in \(B\).

Finally, \(B_i^t(s,\lambda)\) is nondecreasing in \(s\) and increasing in
\(\lambda\) by \eqref{eq:b_monotone}. Since
\(
\varphi_i^t(s,\lambda)
=
\left[\chi_i^t\!\left(B_i^t(s,\lambda)\right)\right]^+,
\)
and \([\chi_i^t(B)]^+\) is nonincreasing in \(B\), it follows that
\(\varphi_i^t(s,\lambda)\) is nonincreasing in \(s\) and nonincreasing in
\(\lambda\). This proves the lemma.
\end{proof}

The next Lemma establishes the monotonicity of the residuals.

\begin{lemma}
\label{lemma:R_Q_monotone}
For fixed \(\lambda\), the residual \(R_t(s;\lambda)\)  \eqref{eq:R_residual} is strictly increasing
in~\(s\). Moreover, \(Q_t^C(\lambda)\) \eqref{eq:Q_residual} is nondecreasing in \(\lambda\).
\end{lemma}

\begin{proof}
By Lemma~\ref{lemma:phi_monotone}, for fixed \(\lambda\), the map
\(
    s\longmapsto \sum_{i\in\pazocal N}\varphi_i^t(s,\lambda)
\)
is nonincreasing. Therefore,
\(
    R_t(s;\lambda)
    =
    s-\sum_{i\in\pazocal N}\varphi_i^t(s,\lambda)
\)
is strictly increasing in~\(s\).

Similarly, by Lemma~\ref{lemma:phi_monotone}, the map
\(
    \lambda\longmapsto \sum_{i\in\pazocal N}\varphi_i^t(C,\lambda)
\)
is nonincreasing. Hence
\(
    Q_t^C(\lambda)
    =
    C-\sum_{i\in\pazocal N}\varphi_i^t(C,\lambda)
\)
is nondecreasing in \(\lambda\).
\end{proof}


\begin{lemma}
\label{lemma:follower_bisection}
The scalar equation \(R_t(s;0)=0\) can be solved by bisection on any interval
\([0,\overline s^t]\) containing \(\widehat s^t\). If \(\widehat s^t>C\), then
the equation \(Q_t^C(\lambda)=0\) can be solved by bisection on any interval
\([0,\overline\lambda_t]\) containing $\lambda^{t*}$.
\end{lemma}

\begin{proof}
By Lemma~\ref{lemma:R_Q_monotone}, \(R_t(\cdot;0)\) is strictly increasing.
Moreover, \(R_t(\cdot;0)\) is continuous because each
\(\varphi_i^t(\cdot,0)\) is continuous.

Let \(\widehat s^t\) be a solution of
\(
    R_t(s;0)=0.
\)
Assume that \(\widehat s^t\in[0,\overline s^t]\). Then
\(
    0\le \widehat s^t\le \overline s^t.
\)
Since \(R_t(\cdot;0)\) is strictly increasing, we obtain
\[
    R_t(0;0)
    \le
    R_t(\widehat s^t;0)
    \le
    R_t(\overline s^t;0).
\]
Using \(R_t(\widehat s^t;0)=0\), this gives
\(
    R_t(0;0)\le 0\le R_t(\overline s^t;0).
\)
If one endpoint value is zero, then the solution has already been found.
Otherwise,
\(
    R_t(0;0)<0<R_t(\overline s^t;0),
\)
and hence
\(
    R_t(0;0)R_t(\overline s^t;0)<0.
\)
Therefore, by the classical bisection theorem for continuous functions with
opposite signs at the endpoints  \citep[Theorem 2.1]{burden2015numerical} and \citep{sikorski1982bisection}, the equation \(R_t(s;0)=0\) can be solved by
bisection on \([0,\overline s^t]\). Since \(R_t(\cdot;0)\) is strictly
increasing, this solution is unique.


By Lemma~\ref{lemma:R_Q_monotone}, \(Q_t^C(\lambda)\) is nondecreasing in
\(\lambda\). Moreover, \(Q_t^C\) is continuous because each
\(\varphi_i^t(C,\cdot)\) is continuous. Let
\([0,\overline\lambda_t]\) be an interval satisfying
\[
    Q_t^C(0)\le 0\le Q_t^C(\overline\lambda_t).
\]
If one endpoint value is zero, then the solution has already been found.
Otherwise,
\(
    Q_t^C(0)<0<Q_t^C(\overline\lambda_t),
\)
and hence
\(
    Q_t^C(0)Q_t^C(\overline\lambda_t)<0.
\)
Therefore, by the classical bisection theorem \citep[Theorem 2.1]{burden2015numerical}, the equation
\(
    Q_t^C(\lambda)=0
\)
can be solved by bisection on \([0,\overline\lambda_t]\).
\end{proof}


\begin{lemma}
\label{lemma:follower_correctness}
Algorithm~\ref{alg:follower_computation} recovers the unique follower VE
\(\mathbf h^{t*}(C,\theta)\) and its associated multiplier
\(\lambda^{t*}(C,\theta)\).
\end{lemma}

\begin{proof}
First compute \(\widehat s^t\) from
\(
    R_t(s;0)=0.
\)
Then
\(
    \widehat s^t
    =
    \sum_{i\in\pazocal N}\varphi_i^t(\widehat s^t,0).
\)
If \(\widehat s^t\le C\), the unconstrained candidate is feasible for the
shared-capacity constraint. Therefore the capacity constraint is not binding,
and the multiplier is
\(
    \lambda^{t*}=0.
\)
Thus
\(
    s^t=\widehat s^t,
    \quad
    h_i^{t*}=\varphi_i^t(s^t,0),
    \quad i\in\pazocal N .
\)

If \(\widehat s^t>C\), the unconstrained candidate violates the capacity
constraint. Therefore the shared-capacity constraint binds, and
\(
    s^t=C.
\)
The multiplier is then chosen to satisfy
\(
    Q_t^C(\lambda^{t*})=0,
\)
that is,
\(
    C=\sum_{i\in\pazocal N}\varphi_i^t(C,\lambda^{t*}).
\)
Then
\(
    h_i^{t*}=\varphi_i^t(C,\lambda^{t*}),
    \quad i\in\pazocal N .
\)

In both cases, \(h_i^{t*}\) is recovered from
\(
    h_i^{t*}=\varphi_i^t(s^t,\lambda^{t*}),
    \quad i\in\pazocal N,
\)
and the scalar pair \((s^t,\lambda^{t*})\) satisfies the corresponding
feasibility and complementarity conditions \eqref{eq:mu_non_neg} \eqref{eq:complementary}. Since the follower VE is unique, the vector returned by the algorithm
is the unique VE.
\end{proof}

The next Lemma gives the computational cost.

\begin{lemma}
\label{lemma:follower_complexity}
Assume that \(\overline s^t\) and \(\overline\lambda_t\) are fixed
problem-dependent bounds independent of the accuracy tolerance
\(\varepsilon\). Then Algorithm~\ref{alg:follower_computation} has worst-case
complexity
\(
    O\!\left(
    |\pazocal N|\log\frac{1}{\varepsilon}
    \right).
\)
\end{lemma}

\begin{proof}
Each evaluation of \(R_t(\cdot;\lambda)\) or \(Q_t^C(\cdot)\) requires
computing \(\varphi_i^t(\cdot,\cdot)\) for every follower
\(i\in\pazocal N\). Since each \(\varphi_i^t\) is evaluated in constant time,
one residual evaluation costs
\(
    O(|\pazocal N|).
\)

By the bisection error bound recalled in \citep[Theorem 2.1]{burden2015numerical}, after
\(n\) bisection evaluations on an interval \([a,b]\), the error is of order
\(
    \frac{b-a}{2^{n+1}}.
\)
Hence, to achieve accuracy \(\varepsilon\), it is sufficient to take
\(
    n
    =
    O\!\left(
    \log\frac{b-a}{\varepsilon}
    \right).
\)

For the equation \(R_t(s;0)=0\), the bisection interval is
\([0,\overline s^t]\), whose length is \(\overline s^t\). Since
\(\overline s^t\) is fixed with respect to \(\varepsilon\),
\(
    \log\frac{\overline s^t}{\varepsilon}
    =
    \log(\overline s^t)+\log\frac{1}{\varepsilon}
    =
    O\!\left(\log\frac{1}{\varepsilon}\right).
\)
Thus solving \(R_t(s;0)=0\) requires
\(
    O\!\left(\log\frac{1}{\varepsilon}\right)
\)
bisection iterations, and therefore costs
\(
    O\!\left(
    |\pazocal N|\log\frac{1}{\varepsilon}
    \right).
\)

In the binding case, the equation \(Q_t^C(\lambda)=0\) is solved on
\([0,\overline\lambda_t]\). Since \(\overline\lambda_t\) is also fixed with
respect to \(\varepsilon\),
\(
    \log\frac{\overline\lambda_t}{\varepsilon}
    =
    O\!\left(\log\frac{1}{\varepsilon}\right).
\)
Therefore the multiplier step also costs
\(
    O\!\left(
    |\pazocal N|\log\frac{1}{\varepsilon}
    \right).
\)

Thus, the
worst-case complexity of the algorithm is
\(
    O\!\left(
    |\pazocal N|\log\frac{1}{\varepsilon}
    \right).
\)
\end{proof}

Combining Lemma~\ref{lemma:follower_bisection},
Lemma~\ref{lemma:follower_correctness}, and
Lemma~\ref{lemma:follower_complexity}, we obtain that the unique follower VE
can be computed by bisection on at most two scalar monotone equations: first
\(R_t(s;0)=0\), and, only if the shared-capacity constraint is binding,
\(Q_t^C(\lambda)=0\). Then the overall computational cost is
\(
    O\!\left(
    |\pazocal N|\log\frac{1}{\varepsilon}
    \right).
\)
\subsection{Leader capacity computation}
\label{subsec:leader_capacity_computation}
Fix a price
\(\theta\in[0,\overline \theta]\). The leader problem~\eqref{eq:leader_problem}
is
\begin{equation}
\max_{0\le C\le \overline C}\ \Phi(C,\theta),
\qquad
\Phi(C,\theta)
:=
\sum_{t\in\pazocal T}\sum_{i\in\pazocal N}
\theta\,h_i^{t*}(C,\theta)-Cost(C),
\label{eq:leader_fixed_theta_problem}
\end{equation}
where \(h_i^{t*}(C,\theta)\) denotes the unique VE of follower \(i\) at
time~\(t\).

For each \(t\in\pazocal T\), let
\(
\widehat{\mathbf h}^{\,t}(\theta)
=
\bigl(
\widehat h_1^t(\theta),\dots,
\widehat h_{|\pazocal N|}^t(\theta)
\bigr)
\)
denote the unconstrained follower resource commitment obtained from the follower-side
computation with multiplier equal to zero ($\lambda^t=0$) (see \eqref{eq:explicit_hi_star}). Define the corresponding
unconstrained total resource by
\begin{equation}
\widehat s^t(\theta)
:=
\sum_{i\in\pazocal N}\widehat h_i^{\,t}(\theta).
\label{eq:def_Dhat}
\end{equation}
When the dependence on the fixed price must be emphasized, we write
$\hat{s}^t(\theta)$ for the quantity denoted by $\hat{s}^t$ in Algorithm~\ref{alg:follower_computation}.

The leader capacity is computed by Algorithm~\ref{alg:leader_capacity_computation}.

\begin{algorithm}[htbp]
\caption{Computation of the optimal leader capacity for fixed \(\theta\)}
\label{alg:leader_capacity_computation}
\begin{algorithmic}[1]
\REQUIRE Fixed price \(\theta\in[0,\overline\theta]\), capacity bound \(\overline C\)
\ENSURE Optimal capacity \(C^*(\theta)\)

\STATE For every \(t\in\pazocal T\), compute \(\widehat s^t(\theta)\) \eqref{eq:def_Dhat}.

\STATE Form the ordered list
\(
0=c_0<c_1<\cdots<c_M=\overline C
\).

\STATE For every \(\ell=1,\dots,M\), compute
\(
k_\ell
=
\left|
\left\{
t\in\pazocal T:
\widehat s^t(\theta)\ge c_\ell
\right\}
\right|.
\)

\STATE For each interval \(J_\ell=[c_{\ell-1},c_\ell]\), compute
\(
C_\ell^*
\in
\arg\max_{C\in J_\ell}
\left\{
\theta k_\ell C-Cost(C)
\right\}.
\)

\STATE Define the finite candidate set
\[
\mathcal C(\theta)
=
\{C_\ell^*:\ell=1,\dots,M\}
\cup
\left(
\{\widehat s^t(\theta):t\in\pazocal T\}
\cap[0,\overline C]
\right)
\cup
\{0,\overline C\}.
\]

\STATE Compute
\(
C^*(\theta)
\in
\arg\max_{C\in\mathcal C(\theta)}
\left\{
\theta\sum_{t\in\pazocal T}
\min\{\widehat s^t(\theta),C\}
-
Cost(C)
\right\}.
\)

\RETURN \(C^*(\theta)\).
\end{algorithmic}
\end{algorithm}

We now justify the steps of Algorithm~\ref{alg:leader_capacity_computation}.
The first Lemma reduces the total follower resource commitment in each slot to a
scalar minimum.

\begin{lemma}
\label{lemma:leader_total_min}
For every \(t\in\pazocal T\) and every \(C\in[0,\overline C]\), the total
VE resource commitment satisfies
\begin{equation}
\sum_{i\in\pazocal N}h_i^{t*}(C,\theta)
=
\min\{\widehat s^t(\theta),C\}.
\label{eq:total_min_formula}
\end{equation}
\end{lemma}

\begin{proof}
Fix \(t\in\pazocal T\) and \(C\in[0,\overline C]\). By Algorithm \ref{alg:follower_computation} and \eqref{eq:ca},
if
\(
\widehat s^t(\theta)\le C,
\)
where \(\widehat s^t(\theta)\) is defined in \eqref{eq:def_Dhat}. Then the unconstrained follower candidate $\widehat{\mathbf h}^{\,t}(\theta)$ is feasible for the shared-capacity
constraint. Hence the capacity multiplier $\lambda^{t*}$ is zero, and
\begin{equation}
\sum_{i\in\pazocal N}h_i^{t*}(C,\theta)
=
\widehat s^t(\theta).
\label{eq:leader_nonbinding_total}
\end{equation}
If instead
\(
\widehat s^t(\theta)>C,
\)
then the unconstrained follower candidate $\widehat{\mathbf h}^{\,t}(\theta)$ violates the shared-capacity
constraint. Hence the capacity constraint binds, and
\begin{equation}
\sum_{i\in\pazocal N}h_i^{t*}(C,\theta)
=
C.
\label{eq:leader_binding_total}
\end{equation}
Combining \eqref{eq:leader_nonbinding_total} and
\eqref{eq:leader_binding_total} gives \eqref{eq:total_min_formula}.
\end{proof}

For fixed \(\theta\in[0,\overline\theta]\), applying \eqref{eq:total_min_formula} to each \(t\in\pazocal T\) gives

\begin{equation}
\Phi(C,\theta)
=
\theta\sum_{t\in\pazocal T}
\min\{\widehat s^t(\theta),C\}
-
Cost(C),
\qquad
0\le C\le \overline C .
\label{eq:leader_obj_min_form}
\end{equation}

We now introduce the breakpoints of \eqref{eq:leader_obj_min_form}. Let
\begin{equation}
0=c_0<c_1<\cdots<c_M=\overline C
\label{eq:leader_ordered_breakpoints}
\end{equation}
be the ordered list obtained from the boundary points \(0,\overline C\) and
the distinct breakpoints \(\widehat s^t(\theta)\) that belong to
\([0,\overline C]\). Since there are at most \(|\pazocal T|\) such breakpoints,
\begin{equation}
M\le |\pazocal T|+1.
\label{eq:leader_number_intervals}
\end{equation}
The points in \eqref{eq:leader_ordered_breakpoints} define the regime
intervals
\begin{equation}
J_\ell:=[c_{\ell-1},c_\ell],
\qquad
\ell=1,\dots,M.
\label{eq:regime_intervals}
\end{equation}
For each interval \(J_\ell\), define
\begin{equation}
\pazocal T_{>}^\ell(\theta)
:=
\left\{
t\in\pazocal T:
\widehat s^t(\theta)\ge c_\ell
\right\},
\label{eq:T_greater_ell}
\end{equation}
and
\begin{equation}
k_\ell
:=
|\pazocal T_{>}^\ell(\theta)|.
\label{eq:k_ell_def}
\end{equation}

\begin{lemma}
\label{lemma:leader_branch_formula}
For every \(\ell\in\{1,\dots,M\}\), every \(C\in J_\ell\), and every
\(t\in\pazocal T\),
\begin{equation}
\min\{\widehat s^t(\theta),C\}
=
\begin{cases}
C, & t\in \pazocal T_{>}^\ell(\theta),\\[1mm]
\widehat s^t(\theta), & t\notin \pazocal T_{>}^\ell(\theta).
\end{cases}
\label{eq:min_branch_on_interval}
\end{equation}
\end{lemma}

\begin{proof}
Fix \(\ell\in\{1,\dots,M\}\) and \(C\in J_\ell\).

If \(t\in\pazocal T_{>}^\ell(\theta)\), then by
\eqref{eq:T_greater_ell},
\(
\widehat s^t(\theta)\ge c_\ell.
\)
Since \(C\in J_\ell=[c_{\ell-1},c_\ell]\), we have \(C\le c_\ell\). Hence
\(
C\le c_\ell\le \widehat s^t(\theta),
\)
and therefore
\(
\min\{\widehat s^t(\theta),C\}=C.
\)

If \(t\notin\pazocal T_{>}^\ell(\theta)\), then by
\eqref{eq:T_greater_ell},
\(
\widehat s^t(\theta)<c_\ell.
\)
Because the ordered list in \eqref{eq:leader_ordered_breakpoints} contains all
breakpoints, no breakpoint lies in the interior of \(J_\ell\). Therefore,
\(
\widehat s^t(\theta)\le c_{\ell-1}.
\)
Since \(C\in J_\ell\), we have \(c_{\ell-1}\le C\). Hence
\(
\widehat s^t(\theta)\le c_{\ell-1}\le C,
\)
and therefore
\(
\min\{\widehat s^t(\theta),C\}
=
\widehat s^t(\theta).
\)
This proves \eqref{eq:min_branch_on_interval}.
\end{proof}

The next Lemma gives the explicit expression of the leader objective on each
regime interval.

\begin{lemma}
\label{lemma:leader_piecewise_objective}
For every \(\ell\in\{1,\dots,M\}\) and every \(C\in J_\ell\),
\begin{equation}
\Phi(C,\theta)
=
\theta k_\ell C
+
\theta
\sum_{t\in\pazocal T\setminus \pazocal T_{>}^\ell(\theta)}
\widehat s^t(\theta)
-
Cost(C).
\label{eq:leader_piecewise_interval}
\end{equation}
Consequently, maximizing \(\Phi(C,\theta)\) over \(J_\ell\) is equivalent to
maximizing
\begin{equation}
C\longmapsto
\theta k_\ell C-Cost(C)
\label{eq:reduced_interval_function}
\end{equation}
over the same interval \(J_\ell\).
\end{lemma}

\begin{proof}
By \eqref{eq:min_branch_on_interval}, for every \(C\in J_\ell\),
\begin{align}
\sum_{t\in\pazocal T}
\min\{\widehat s^t(\theta),C\}
&=
|\pazocal T_{>}^\ell(\theta)|\,C
+
\sum_{t\in\pazocal T\setminus \pazocal T_{>}^\ell(\theta)}
\widehat s^t(\theta)  \notag\\
&=
k_\ell C
+
\sum_{t\in\pazocal T\setminus \pazocal T_{>}^\ell(\theta)}
\widehat s^t(\theta),
\label{eq:sum_min_interval}
\end{align}
where the second equality uses \eqref{eq:k_ell_def}. Substituting
\eqref{eq:sum_min_interval} into \eqref{eq:leader_obj_min_form} gives
\eqref{eq:leader_piecewise_interval}.
The second term in \eqref{eq:leader_piecewise_interval} is constant with
respect to \(C\) on \(J_\ell\). Hence maximizing \(\Phi(C,\theta)\) over
\(J_\ell\) is equivalent to maximizing the reduced function
\eqref{eq:reduced_interval_function} over \(J_\ell\).
\end{proof}

The next Lemma constructs a finite candidate set that contains a global
maximizer.

\begin{lemma}
\label{lemma:leader_candidate_set}
Assume that for every integer \(k\in\{0,\dots,|\pazocal T|\}\) and every
interval \(J\subseteq[0,\overline C]\), a maximizer of
\(
C\longmapsto \theta kC-Cost(C)
\)
over \(J\) can be computed. For each interval \(J_\ell\), let
\begin{equation}
C_\ell^*
\in
\arg\max_{C\in J_\ell}
\left\{
\theta k_\ell C-Cost(C)
\right\}.
\label{eq:C_ell_star}
\end{equation}
Define the finite candidate set
\begin{equation}
\mathcal C(\theta)
:=
\underbrace{\{C_\ell^*:\ell=1,\dots,M\}}_{\text{interval maximizers}}
\;\cup\;
\underbrace{
\left(
\{\widehat s^t(\theta):t\in\pazocal T\}
\cap[0,\overline C]
\right)
}_{\text{breakpoints}}
\;\cup\;
\underbrace{\{0,\overline C\}}_{\text{boundary points}} .
\label{eq:candidate_set_C}
\end{equation}
Then the leader problem admits an optimal solution in \(\mathcal C(\theta)\).
Equivalently,
\begin{equation}
C^*(\theta)
\in
\arg\max_{C\in\mathcal C(\theta)}
\Phi(C,\theta)
\label{eq:C_star_candidate_problem}
\end{equation}
solves the fixed-price leader problem.
\end{lemma}

\begin{proof}
Let \(\widetilde C\) be a global maximizer of
\(\Phi(\cdot,\theta)\) over \([0,\overline C]\). Since the intervals in
\eqref{eq:regime_intervals} cover \([0,\overline C]\), there exists
\(\ell\in\{1,\dots,M\}\) such that
\begin{equation}
\widetilde C\in J_\ell.
\label{eq:Ctilde_in_interval}
\end{equation}
By Lemma~\ref{lemma:leader_piecewise_objective}, maximizing
\(\Phi(C,\theta)\) over \(J_\ell\) is equivalent to maximizing
\eqref{eq:reduced_interval_function} over \(J_\ell\). Hence the interval
maximizer \(C_\ell^*\) in \eqref{eq:C_ell_star} satisfies
\begin{equation}
\Phi(C_\ell^*,\theta)
\ge
\Phi(\widetilde C,\theta).
\label{eq:Cell_dominates_Ctilde}
\end{equation}
Since \(\widetilde C\) is globally optimal, \eqref{eq:Cell_dominates_Ctilde}
implies that \(C_\ell^*\) is also globally optimal. Moreover, by
\eqref{eq:candidate_set_C},
\(
C_\ell^*\in\mathcal C(\theta).
\)
Therefore, \(\mathcal C(\theta)\) contains at least one global maximizer of
the original leader problem. Hence \eqref{eq:C_star_candidate_problem} solves
the fixed-price leader problem.
\end{proof}

The next Lemma gives the arithmetic complexity of the leader capacity
computation.

\begin{lemma}
\label{lemma:leader_complexity}
Assume that each interval maximizer in \eqref{eq:C_ell_star} can be computed
in time \(\Gamma_C\). Then Algorithm~\ref{alg:leader_capacity_computation}
has arithmetic complexity
\begin{equation}
O\!\left(
|\pazocal T|
\left(
|\pazocal N|\log\frac{1}{\varepsilon}
+
\log |\pazocal T|
+
\Gamma_C
\right)
\right).
\label{eq:leader_complexity_bound}
\end{equation}
\end{lemma}

\begin{proof}
First, by Lemma~\ref{lemma:follower_complexity}, each
\(\widehat s^t(\theta)\) is obtained by solving one scalar monotone equation,
with cost
\(
O\!\left(
|\pazocal N|\log(1/\varepsilon)
\right).
\)
Therefore, computing all values \(\widehat s^t(\theta)\),
\(t\in\pazocal T\), costs
\begin{equation}
O\!\left(
\sum_{t\in\pazocal T}
|\pazocal N|\,\log(1/\varepsilon)
\right) = O\!\left(|\pazocal T|
|\pazocal N|\,\log(1/\varepsilon)
\right) 
\label{eq:complexity_Dhat}
\end{equation}

Second, constructing the ordered list in
\eqref{eq:leader_ordered_breakpoints} requires sorting the breakpoints
\(\widehat s^t(\theta)\), \(t\in\pazocal T\). This costs
\begin{equation}
O(|\pazocal T|\log|\pazocal T|).
\label{eq:complexity_sort}
\end{equation}

Third, by \eqref{eq:leader_number_intervals}, there are at most
\(|\pazocal T|+1\) regime intervals. Since each interval maximizer
\(C_\ell^*\) in \eqref{eq:C_ell_star} can be computed in time \(\Gamma_C\),
computing all interval maximizers costs
\begin{equation}
O(|\pazocal T|\,\Gamma_C).
\label{eq:complexity_interval_max}
\end{equation}

Finally, the candidate set \(\mathcal C(\theta)\) in
\eqref{eq:candidate_set_C} has cardinality \(O(|\pazocal T|)\). Once the
values \(\widehat s^t(\theta)\) are available, evaluating
\(
\theta\sum_{t\in\pazocal T}
\min\{\widehat s^t(\theta),C\}
-
Cost(C)
\)
at all candidates and selecting the best one is linear in the number of
candidates. This cost is dominated by
\eqref{eq:complexity_Dhat}, \eqref{eq:complexity_sort}, and
\eqref{eq:complexity_interval_max}. Combining these bounds gives
\eqref{eq:leader_complexity_bound}.
\end{proof}

\subsection{Leader price computation}
\label{subsec:leader_price_computation}

Define the leader objective in \eqref{eq:leader_problem} by
\begin{equation}
\phi(\theta)
:=
\max_{0\le C\le \overline C}
\left\{
\theta\sum_{t\in\pazocal T}\min\{\widehat s^t(\theta),C\}
-
Cost(C)
\right\},
\qquad
\theta\in[0,\overline\theta],
\label{eq:reduced_price_value}
\end{equation}
where \(\widehat s^t(\theta)\) is defined in \eqref{eq:def_Dhat}.

The price computation is summarized in Algorithm~\ref{alg:leader_price_grid}.

\begin{algorithm}[htbp]
\caption{Price-grid computation}
\label{alg:leader_price_grid}
\begin{algorithmic}[1]
\REQUIRE Price bound \(\overline\theta\), optimality gap \(\eta>0\), Lipschitz constant \(L_\phi\)
\ENSURE Grid price \(\theta_\Delta\) and capacity \(C^*(\theta_\Delta)\)

\STATE Choose a uniform price grid \(\Theta_\Delta\subset[0,\overline\theta]\)
with mesh size
\(
\Delta_\theta\le \frac{2\eta}{L_\phi}.
\)

\FOR{each \(\theta\in\Theta_\Delta\)}
    \STATE Compute \(C^*(\theta)\) using Algorithm~\ref{alg:leader_capacity_computation}.
    \STATE Evaluate
    \(
    \phi(\theta)
    =
    \theta\sum_{t\in\pazocal T}
    \min\{\widehat s^t(\theta),C^*(\theta)\}
    -
    Cost(C^*(\theta)).
    \)
\ENDFOR

\STATE Select
\(\theta_\Delta
\in
\arg\max_{\theta\in\Theta_\Delta}\phi(\theta).
\)

\RETURN \(\theta_\Delta\) and \(C^*(\theta_\Delta)\).
\end{algorithmic}
\end{algorithm}

We first prove that
\(\widehat s^t(\theta)\) is Lipschitz in \(\theta\).

\begin{lemma}
\label{lemma:Dhat_lipschitz_price}
For every \(t\in\pazocal T\),
\(\widehat s^t(\theta)\) is nonincreasing and Lipschitz continuous on
\([0,\overline\theta]\). In particular, for all
\(\theta_1,\theta_2\in[0,\overline\theta]\),
\begin{equation}
\left|
\widehat s^t(\theta_1)-\widehat s^t(\theta_2)
\right|
\le
\frac{|\pazocal N|}{\delta-\gamma}
|\theta_1-\theta_2|.
\label{eq:Dhat_lipschitz}
\end{equation}
\end{lemma}

\begin{proof}
Fix \(t\in\pazocal T\). By definition,
\(
\widehat s^t(\theta)
=
\sum_{i\in\pazocal N}\widehat h_i^t(\theta)
\)
is the unconstrained total resource commitment defined in \eqref{eq:def_Dhat}.
Since the unconstrained problem has no capacity multiplier, the stationarity
condition \eqref{eq:stationarity_positive_hi} for an active follower is
\(
\frac{A_i^t}{1+\widehat h_i^t}
=
b_i\theta
+
b_i\gamma\sum_{j\neq i}\widehat h_j^t
+
b_i\delta \widehat h_i^t .
\)
Using
\(
\sum_{j\neq i}\widehat h_j^t
=
\widehat s^t(\theta)-\widehat h_i^t,
\)
we obtain
\begin{equation}
\frac{A_i^t}{1+\widehat h_i^t}
=
b_i\theta
+
b_i\gamma\widehat s^t(\theta)
+
b_i(\delta-\gamma)\widehat h_i^t .
\label{eq:a_to_z}
\end{equation}
Defining
\begin{equation}
z:=\theta+\gamma\widehat s^t(\theta),
\label{eq:z_price_load}
\end{equation}
\eqref{eq:a_to_z} becomes
\begin{equation}
\frac{A_i^t}{1+\widehat h_i^t}
=
b_i z+b_i(\delta-\gamma)\widehat h_i^t .
\label{eq:unconstrained_z_foc}
\end{equation}

For each fixed \(z\ge0\), let \(\rho_i^t(z)\) denote the unique nonnegative
solution of \eqref{eq:unconstrained_z_foc}. Equivalently, from
\eqref{eq:explicit_hi_star},
\begin{equation}
\rho_i^t(z)
=
\max\left\{
0,\,
\frac{
-\left(b_i z+b_i(\delta-\gamma)\right)
+
\sqrt{
\left(b_i z-b_i(\delta-\gamma)\right)^2
+
4b_i(\delta-\gamma)A_i^t
}
}{
2b_i(\delta-\gamma)
}
\right\}.
\label{eq:rho_closed_form}
\end{equation}
Then
\begin{equation}
\widehat h_i^t(\theta)=\rho_i^t(z),
\qquad
\widehat s^t(\theta)=\sum_{i\in\pazocal N}\rho_i^t(z).
\label{eq:hhat_s_rho}
\end{equation}

For any active follower, implicit differentiation of
\eqref{eq:unconstrained_z_foc} with respect to \(z\) gives
\[
-\frac{A_i^t}{(1+\rho_i^t(z))^2}\frac{d\rho_i^t(z)}{dz}
=
b_i+b_i(\delta-\gamma)\frac{d\rho_i^t(z)}{dz}.
\]
Hence,
\begin{equation}
\frac{d\rho_i^t(z)}{dz}
=
-
\frac{b_i}
{
\frac{A_i^t}{(1+\rho_i^t(z))^2}
+
b_i(\delta-\gamma)
}.
\label{eq:rho_derivative}
\end{equation}
Since \(A_i^t\ge0\), \(\rho_i^t(z)\ge0\), and
\(\delta>\gamma\) by Lemma~\ref{lemma:delta_greater_gamma}, the denominator
in \eqref{eq:rho_derivative} is bounded below by \(b_i(\delta-\gamma)\). Hence,
\begin{equation}
\left|
\frac{d\rho_i^t(z)}{dz}
\right|
=
\frac{b_i}
{
\frac{A_i^t}{(1+\rho_i^t(z))^2}
+
b_i(\delta-\gamma)
}
\le
\frac{b_i}{b_i(\delta-\gamma)}
=
\frac{1}{\delta-\gamma}.
\label{eq:rho_lipschitz}
\end{equation}
If follower \(i\) is inactive, then \(\rho_i^t(z)=0\) on the corresponding
inactive region. Hence \(\rho_i^t\) is constant there and
\(
\frac{d\rho_i^t(z)}{dz}=0.
\)
Therefore,
\(
\left|
\frac{d\rho_i^t(z)}{dz}
\right|
=0
\le
\frac{1}{\delta-\gamma}.
\)
At the activation threshold, \(\rho_i^t\) may fail to be differentiable because
it switches from the zero branch to the positive branch. However, this does not
affect Lipschitz continuity. On the zero branch the slope is \(0\), and on the
positive branch, the derivative is bounded in absolute value by
\(1/(\delta-\gamma)\). Since the two branches agree at the threshold,
\(\rho_i^t\) is continuous there. Therefore, applying \citep[ 3.Mean Value Theorem ]{pugh2002real}
on each branch and using continuity at the switching point, \(\rho_i^t\) is
globally Lipschitz with constant \(1/(\delta-\gamma)\); see
\citep[5.Corollary]{pugh2002real}.

Define
\(
X^t(z):=\sum_{i\in\pazocal N}\rho_i^t(z)\).
Then \(X^t\) is nonincreasing and Lipschitz continuous with constant
\(
L_X:=\frac{|\pazocal N|}{\delta-\gamma}.
\)
Moreover, \(\widehat s^t(\theta)\) satisfies
\begin{equation}
\widehat s^t(\theta)
=
X^t\!\left(\theta+\gamma\widehat s^t(\theta)\right).
\label{eq:Dhat_fixed_point}
\end{equation}

Let \(0\le\theta_1<\theta_2\le\overline\theta\), and write
\(
S_\ell:=\widehat s^t(\theta_\ell),
\quad
z_\ell:=\theta_\ell+\gamma S_\ell,
\quad \ell=1,2.
\label{eq:D_z_price_def}
\)
Since \(X^t\) is nonincreasing, \eqref{eq:Dhat_fixed_point} implies
\begin{equation}
S_2\le S_1.
\label{eq:D_monotone_price}
\end{equation}
Also, \(z_2\ge z_1\); otherwise \(z_2<z_1\) would imply
\[
S_2=X^t(z_2)\ge X^t(z_1)=S_1,
\]
contradicting \eqref{eq:D_monotone_price} unless \(S_1=S_2\), in which case
the desired bound is immediate. Therefore,
\(
S_1-S_2
=
X^t(z_1)-X^t(z_2)
\le
L_X(z_2-z_1).
\)
Using
\(
z_2-z_1=(\theta_2-\theta_1)-\gamma(S_1-S_2),
\)
we obtain
\(
S_1-S_2
\le
L_X\bigl((\theta_2-\theta_1)-\gamma(S_1-S_2)\bigr)
\le
L_X(\theta_2-\theta_1).
\)
Thus, for all \(\theta_1,\theta_2\in[0,\overline\theta]\),
\[
\left|
\widehat s^t(\theta_1)-\widehat s^t(\theta_2)
\right|
\le
\frac{|\pazocal N|}{\delta-\gamma}
|\theta_1-\theta_2|.
\]
This proves \eqref{eq:Dhat_lipschitz}.
\end{proof}

We now prove that \(\phi\) in \eqref{eq:reduced_price_value} is Lipschitz.

\begin{lemma}
\label{lemma:phi_lipschitz}
\(\phi\) in \eqref{eq:reduced_price_value} is
Lipschitz continuous on \([0,\overline\theta]\). In particular, one may take
\begin{equation}
L_\phi
:=
|\pazocal T|
\left(
\overline C
+
\frac{\overline\theta |\pazocal N|}{\delta-\gamma}
\right)
\label{eq:L_phi}
\end{equation}
as a Lipschitz constant. That is, for all
\(\theta_1,\theta_2\in[0,\overline\theta]\),
\begin{equation}
\left|
\phi(\theta_1)-\phi(\theta_2)
\right|
\le
L_\phi
\left|
\theta_1-\theta_2
\right|.
\label{eq:phi_lipschitz_property}
\end{equation}
\end{lemma}

\begin{proof}
For any fixed \(C\in[0,\overline C]\), define
\begin{equation}
\Phi(C,\theta)
:=
\theta\sum_{t\in\pazocal T}\min\{\widehat s^t(\theta),C\}
-
Cost(C),
\label{eq:Phi_fixed_C_price}
\end{equation}
as in \eqref{eq:leader_obj_min_form}. Since \(0\le C\le\overline C\), we have
\begin{equation}
\sum_{t\in\pazocal T}\min\{\widehat s^t(\theta),C\}
\le
|\pazocal T|\overline C .
\label{eq:bound_t_c_bar}
\end{equation}
Since \(\theta_2\in[0,\overline\theta]\), we have
\(
\theta_2\le\overline\theta.
\)
Moreover, taking the minimum with the fixed capacity \(C\) cannot increase
differences between two values. Indeed, for any \(x,y\ge0\),
\(
\left|\min\{x,C\}-\min\{y,C\}\right|\le |x-y|\)
Therefore, for each \(t\in\pazocal T\),
\[
\left|
\min\{\widehat s^t(\theta_1),C\}
-
\min\{\widehat s^t(\theta_2),C\}
\right|
\le
\left|
\widehat s^t(\theta_1)-\widehat s^t(\theta_2)
\right|.
\]
Hence,
\begin{equation}
\theta_2
\sum_{t\in\pazocal T}
\left|
\min\{\widehat s^t(\theta_1),C\}
-
\min\{\widehat s^t(\theta_2),C\}
\right|
\le
\overline\theta
\sum_{t\in\pazocal T}
\left|
\widehat s^t(\theta_1)-\widehat s^t(\theta_2)
\right|.
\label{eq:theta_two_dhat}
\end{equation}

Therefore, using \eqref{eq:Dhat_lipschitz}, for any
\(\theta_1,\theta_2\in[0,\overline\theta]\),
\begin{align}
|\Phi(C,\theta_1)-\Phi(C,\theta_2)|
&\le
|\theta_1-\theta_2|
\sum_{t\in\pazocal T}\min\{\widehat s^t(\theta_1),C\}
\nonumber\\
&\quad
+
\theta_2
\sum_{t\in\pazocal T}
\left|
\min\{\widehat s^t(\theta_1),C\}
-
\min\{\widehat s^t(\theta_2),C\}
\right|
\nonumber\\
&\overset{\eqref{eq:bound_t_c_bar},\,\eqref{eq:theta_two_dhat}}{\le}
|\pazocal T|\overline C|\theta_1-\theta_2|
+
\overline\theta
\sum_{t\in\pazocal T}
\left|
\widehat s^t(\theta_1)-\widehat s^t(\theta_2)
\right|
\nonumber\\
&\overset{\eqref{eq:Dhat_lipschitz}}{\le}
|\pazocal T|
\left(
\overline C
+
\frac{\overline\theta|\pazocal N|}{\delta-\gamma}
\right)
|\theta_1-\theta_2|.
\label{eq:Phi_lipschitz}
\end{align}
The bound in \eqref{eq:Phi_lipschitz} is uniform in
\(C\in[0,\overline C]\).

Therefore, taking the maximum of the functions
\(C\mapsto \Phi(C,\theta)\) over the compact set \([0,\overline C]\), it
follows that \(\phi\) is Lipschitz with the same constant \(L_\phi\) in
\eqref{eq:L_phi}. Indeed, for any \(\theta_1,\theta_2\),
\[
\phi(\theta_1)-\phi(\theta_2)
\le
\max_{0\le C\le\overline C}
\left\{\Phi(C,\theta_1)-\Phi(C,\theta_2)\right\}
\le
L_\phi|\theta_1-\theta_2|.
\]
Interchanging \(\theta_1\) and \(\theta_2\) gives
\[
|\phi(\theta_1)-\phi(\theta_2)|
\le
L_\phi|\theta_1-\theta_2|.
\]
This proves \eqref{eq:phi_lipschitz_property}.
\end{proof}

It remains to prove the grid approximation guarantee.

\begin{lemma}
\label{lemma:price_grid_eta}
Let \(\Theta_\Delta\subset[0,\overline\theta]\) be a uniform price grid with
mesh size \(\Delta_\theta\), and let
\(\theta_\Delta\in\arg\max_{\theta\in\Theta_\Delta}\phi(\theta)\). If
\begin{equation}
\Delta_\theta\le \frac{2\eta}{L_\phi},
\label{eq:price_grid_spacing}
\end{equation}
then
\begin{equation}
\max_{\theta\in[0,\overline\theta]}\phi(\theta)-\phi(\theta_\Delta)
\le \eta.
\label{eq:eta_opt_price}
\end{equation}
\end{lemma}

\begin{proof}
Let
\(
\theta^*\in\arg\max_{\theta\in[0,\overline\theta]}\phi(\theta).
\)
Choose a grid point \(\widehat\theta\in\Theta_\Delta\) such that
\(
|\theta^*-\widehat\theta|\le \frac{\Delta_\theta}{2}.
\)
Since \(\theta_\Delta\) maximizes \(\phi\) over \(\Theta_\Delta\),
\(
\phi(\theta_\Delta)\ge \phi(\widehat\theta).
\)
Therefore, by the Lipschitz continuity of \(\phi\) (Lemma \ref{lemma:phi_lipschitz}),
\[
\phi(\theta^*)-\phi(\theta_\Delta)
\le
\phi(\theta^*)-\phi(\widehat\theta)
\le
L_\phi|\theta^*-\widehat\theta|
\le
\frac{L_\phi\Delta_\theta}{2}
\]
Using \eqref{eq:price_grid_spacing}, we obtain
\(
\phi(\theta^*)-\phi(\theta_\Delta)\le \eta,
\)
which proves \eqref{eq:eta_opt_price}. Hence the price-grid solution is an approximate price solution with a bounded optimality gap \(\eta\).
\end{proof}

Combining Lemma \ref{lemma:price_grid_eta} with
Lemma \ref{lemma:leader_candidate_set} which computes \(C^*(\theta)\) for each
fixed grid price, gives an approximate SE with a bounded optimality gap $\eta$.

\subsection{Stackelberg Equilibrium computation}
\label{subsec:stackelberg_computation}

By Lemma~\ref{lemma:leader_complexity}, for each fixed price \(\theta\), the
capacity \(C^*(\theta)\) can be computed with arithmetic complexity
\begin{equation}
O\!\left(
|\pazocal T|(|\pazocal N|\log\!\left(\frac{1}{\varepsilon}\right)+\log |\pazocal T|+\Gamma_C)
\right).
\label{eq:fixed_price_oracle_complexity}
\end{equation}

By Lemma \ref{lemma:price_grid_eta}, an approximate price with bounded optimality gap $\eta$
can be obtained using a uniform grid with mesh size
\(
\Delta_\theta\le \frac{2\eta}{L_\phi}.
\)
Hence the number of grid points can be chosen as
\begin{equation}
n_\theta
=
O\!\left(
\frac{L_\phi\overline\theta}{\eta}
\right).
\label{eq:number_price_grid_points}
\end{equation}
Combining \eqref{eq:fixed_price_oracle_complexity} and
\eqref{eq:number_price_grid_points}, the overall arithmetic complexity for
computing an approximate SE with bounded optimality gap $\eta$
\begin{equation}
O\!\left(
\frac{L_\phi\overline\theta}{\eta}
\,
|\pazocal T|(|\pazocal N|L+\log |\pazocal T|+\Gamma_C)
\right).
\label{eq:overall_eta_complexity}
\end{equation}
Using the Lipschitz constant in \eqref{eq:L_phi}, this becomes
\begin{equation}
O\!\left(
\frac{\overline\theta}{\eta}
\,
|\pazocal T|^2
\left(
\overline C+
\frac{\overline\theta |\pazocal N|}{\delta-\gamma}
\right)
(|\pazocal N|L+\log |\pazocal T|+\Gamma_C)
\right).
\label{eq:overall_eta_complexity_expanded}
\end{equation}
In particular, for the quadratic investment cost in \eqref{eq:cost_parameters}, the maximizer of
\(
C\mapsto \theta k C-Cost(C)
\)
over any interval can be computed in closed form by projecting the solution
of
\(
Cost'(C)=\theta k
\)
onto that interval. Hence \(\Gamma_C=O(1)\), and
\eqref{eq:overall_eta_complexity_expanded} reduces to
\begin{equation}
O\!\left(
\frac{\overline\theta}{\eta}
\,
|\pazocal T|^2
\left(
\overline C+
\frac{\overline\theta |\pazocal N|}{\delta-\gamma}
\right)
(|\pazocal N| L+\log |\pazocal T|)
\right).
\label{eq:overall_eta_complexity_quadratic}
\end{equation}
Therefore, the proposed algorithm computes an approximate Stackelberg
equilibrium in polynomial time in \(|\pazocal N|\), \(|\pazocal T|\),
\(1/\eta\), and \(1/(\delta-\gamma)\) with only logarithmic dependence on the
scalar bisection tolerance $\epsilon$ through
\(
L=O(\log(1/\varepsilon)).
\)
\section{Computational runtime}
\label{appendix:equilibrium-computation}

All experiments are implemented in Python on a machine equipped with a
13th Gen Intel Core i9-13950HX CPU at 2.20 GHz.
In the numerical study, we set $\overline{C}=5000$, $\overline{\theta}=2000$, \(\varepsilon=10^{-8}\) and
at most \(1000\) iterations for the scalar searches (see Proposition~\ref{prop:computational_tractability}). 
Table~\ref{tab:runtime} reports the average computational time for different
problem sizes. Runtime increases with both the investment horizon and the
number of SPs. Although individual runs vary across risk and uncertainty
configurations, the average runtime remains below one hour even for the largest
instance with ten SPs and a five-year horizon.

\begin{table}[H]
\centering
\caption{Runtime per experiment.}
\label{tab:runtime}
\begin{tabular}{ccc}
\hline
\(N\) (SPs) & \(I\) (years) & Runtime (min) \\
\hline
3  & 1 & 14 \\
3  & 5 & 30 \\
10 & 5 & 49 \\
\hline
\end{tabular}
\end{table}

\section{Tightness of the Lower Bound on the PoP}
\label{app:pop_gap}
Fig.~\ref{fig:max_guarantee_gap} shows the maximum gap, across SPs, between the empirical PoP estimated over 1000 revenue realizations and the theoretical lower bound~$\hat{\nu}_i$. The gap measures the conservativeness of the guarantee. It is small at low $CV$, where revenue outcomes are stable, and increases with $CV$ because the bound must protect against more severe adverse realizations. The gap is largest under risk neutrality, where SPs keep higher resource exposure. Risk aversion reduces the gap by raising the lower bound through more conservative resource choices. The high-$CV$ gap is mainly driven by SP~1, whose bound is most sensitive to severe uncertainty.
\begin{figure}[!ht]
\centering
\includegraphics[width=0.45\columnwidth]{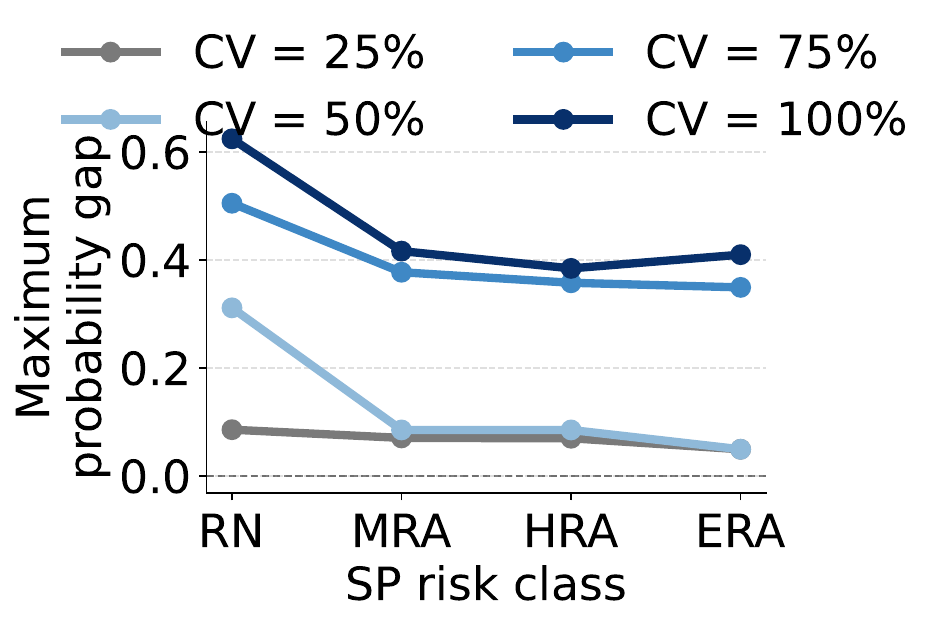}
\caption{Maximum gap across SPs between the empirical PoP and the theoretical lower bound on the PoP.}
\label{fig:max_guarantee_gap}
\end{figure}

\section{Sensitivity to SP-Side and Investment Costs}\label{appendix:sensitivity}

\subsection{Sensitivity to Operational and Congestion Parameters}
\label{subsec:sensitivity_congestion}

We next examine the sensitivity of the equilibrium to the operational and congestion parameters
$\delta$ and $\gamma$ \eqref{eq:psi}. The baseline case corresponds to the congestion setting
used in the main experiments (see Table \ref{tab:model_parameters}), where $\gamma$ is kept at its nominal value and
$\delta$ is chosen just above the uniqueness threshold, i.e.,
$\delta = 1.01\,\gamma \max_i \xi_i$. We compare this baseline with cases in
which either the operational cost parameter $\delta$ or the
congestion cost parameter $\gamma$ is increased.
\begin{figure}[h]
    \centering
    \includegraphics[width=0.65\linewidth]{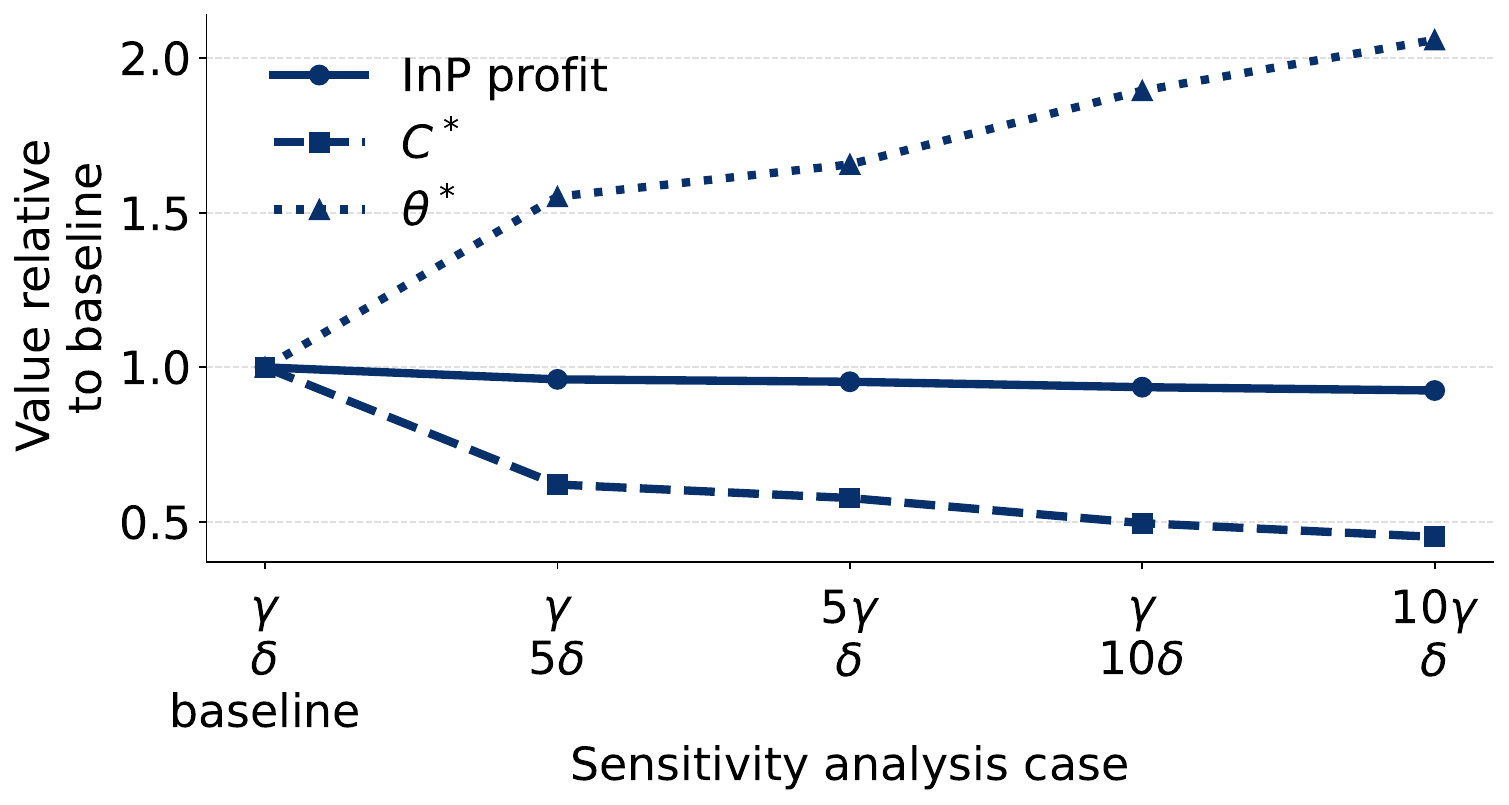}
    \caption{Sensitivity of InP profit, optimal capacity $C^*$, and optimal
    price $\theta^*$ to SP-side cost parameter changes. Values are normalized by
    their baseline values.}
    \label{fig:sensitivity_leader_curves}
\end{figure}
Figure~\ref{fig:sensitivity_leader_curves} reports the InP outcomes
normalized by their corresponding baseline values. Hence, the curves show
relative changes in InP profit, optimal capacity $C^*$, and optimal price
$\theta^*$. The figure shows that stronger
SP-side cost parameters mainly affect the InP through capacity and price. The optimal
capacity decreases as operational and congestion cost parameters become stronger, because SPs reduce their
resource commitments when SP-side costs increase. The InP, therefore, scales
down capacity investment to avoid over-provisioning. At the same time, the
optimal price increases, indicating that the InP compensates for the smaller
resource base by charging a higher access price. The InP's profit decreases
only mildly, since the joint adjustment of capacity and price partly offsets
the negative effect of operational and congestion costs.

\begin{figure}[t]
    \centering
    \includegraphics[width=0.8\linewidth]{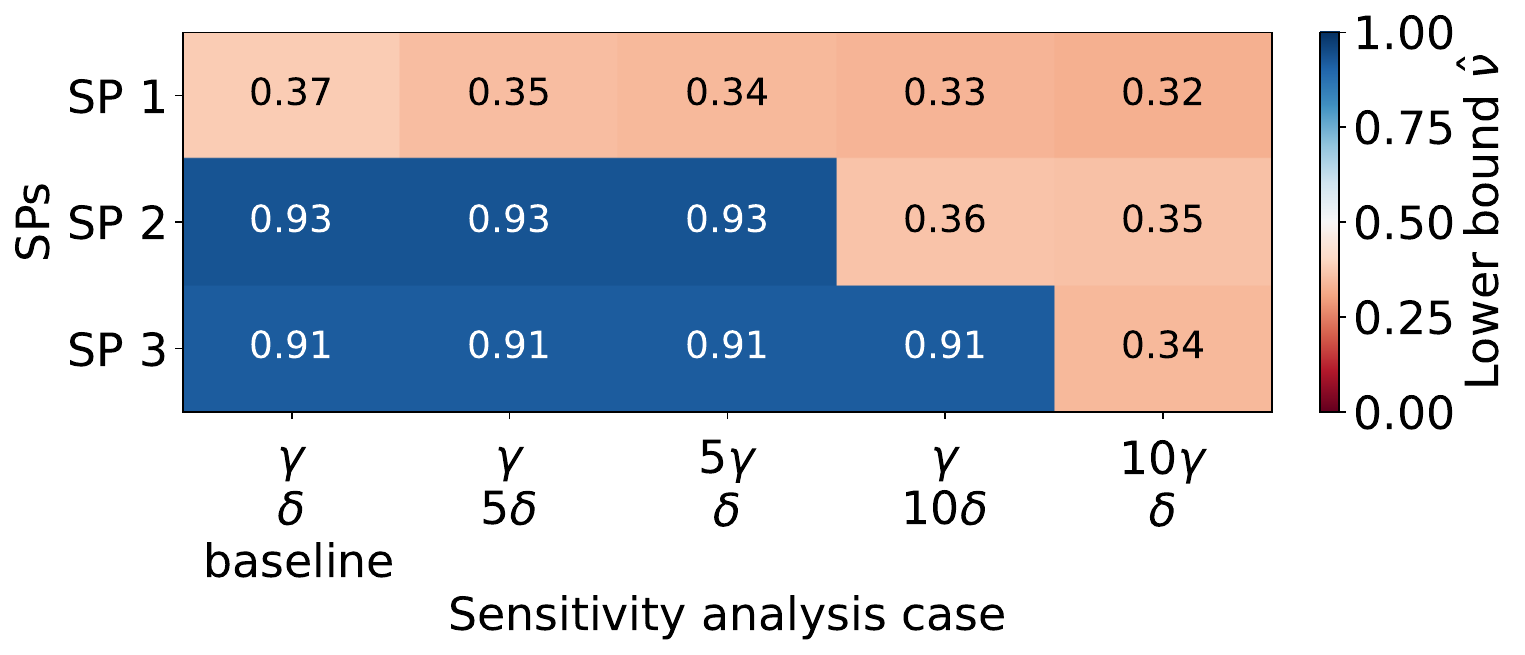}
    \caption{Sensitivity of the lower bound $\hat{\nu}_i$ on the Probability of Profit to SP-side cost parameters changes.}
    \label{fig:nu_hat_by_case_heatmap}
\end{figure}

Figure~\ref{fig:nu_hat_by_case_heatmap} reports the lower bound on SPs’ PoP and shows that the effect on this lower bound is heterogeneous. The smallest SP (SP 1) is affected across all cases, reflecting its limited revenue buffer. The medium SP (SP 2) remains
protected under moderate change in $\delta$ and $\gamma$ but becomes vulnerable under
stronger values of $\delta$ and $\gamma$. The largest SP (SP 3) is the most robust and is mainly affected
only when shared-infrastructure congestion becomes severe. Overall, increasing $\gamma$
has the strongest impact because it amplifies the coupling among SPs sharing
the same infrastructure.

\subsection{Sensitivity to Investment Costs}
\label{subsec:sensitivity_investment_costs}

We also examine the sensitivity of the equilibrium to the investment-cost
parameters. The baseline case corresponds to the cost setting used in the main
experiments (see Table \ref{tab:model_parameters}). We then multiply all investment-cost parameters, $F_0$, $d_{\text{dep}}$, $d_{\text{maint}}$,
and $d_{\text{exp}}$, by a factor of ten and compare the resulting equilibrium with
the baseline case.

Table~\ref{tab:sensitivity_investment_costs} shows that multiplying all
investment-cost parameters by ten has only a limited effect on the equilibrium.
The InP slightly reduces capacity and increases the access price, while
InP profit decreases slightly. The lower bounds on the PoP $\hat{\nu}$ remain unchanged for all SPs because the investment
cost parameters affect the InP's cost function, but do not directly enter
the SP profit, CVaR, or resource commitment decisions. Their effect on SPs is only
indirect through the small changes in $C^*$ and $\theta^*$.
\begin{table}[H]
\centering
\caption{Sensitivity to investment-cost parameters at $CV=100\%$.}
\label{tab:sensitivity_investment_costs}
\begin{tabular}{l|c|c}
\hline
Metric & Baseline & $10\times$ cost parameters \\
\hline
InP profit (\$) & $192.54$M & $191.77$M \\
$C^*$(vCore) & $86$ & $84$ \\
$\theta^*$(\$/vCore-hour) & $67$ & $69$ \\
$\hat{\nu}_1$ & $0.37$ & $0.37$ \\
$\hat{\nu}_2$ & $0.93$ & $0.93$ \\
$\hat{\nu}_3$ & $0.91$ & $0.91$ \\
\hline
\end{tabular}
\end{table}

\section{Validation with Real-World Dataset}
\label{app:dataset_aiwt}
We further strengthen the realism of our results by incorporating a real-world dataset. Since large-scale MEC deployments by network operators at their 5G base stations remain limited, possibly due to the absence of effective investment-pricing mechanisms, our study relies on the closest available empirical representation for mobile-demand. Specifically, we use the publicly available Telecom Italia Call Detail Records dataset,\footnote{\url{https://doi.org/10.7910/DVN/EGZHFV}} which contains real mobile traffic measurements collected in Milan, Italy. This dataset has been widely adopted in MEC-related studies \citep{bouet2018mobile, hussain2019mobile} as a representation for real-world user demand. We therefore use it to capture user demand variations across time and across SPs.

The dataset contains Call Detail Records collected in Milan, Italy, over November and December 2013. The records are organized by geographical grid cell and by 10-minute time interval, which makes the dataset suitable for representing time dependent mobile activity.
It also includes several categories of mobile-network usage: sent and received SMS, incoming and outgoing voice calls, and Internet activity. Each entry contains a timestamp, a square identifier, and the corresponding traffic measurements for these activity categories. Since the MEC services considered in this work are mainly associated with data-user demand, we use the Internet-usage category.
We first aggregate the Internet-usage measurements from 10-minute intervals to hourly slots. For each hour, the traffic load is computed at the grid-cell level. The grid cells are then divided among the SPs, so that each SP receives a fixed set of cells. The traffic observed in these assigned cells provides the empirical load realizations used in the dataset-based experiment. 

Let $\ell^t_{i,\omega}$ be the hourly request load associated with SP $i$ at time slot $t$ for empirical realization $\omega$. We convert this load into the revenue coefficient through
\(
a^t_{i,\omega}=e_i\ell^t_{i,\omega} \Delta,
\)
where $e_i$ denotes the benefit per request in dollars and $\Delta$ is the time slot length. Hence, $a^t_{i,\omega}$ captures the monetary revenue opportunity generated by the observed traffic assigned to SP $i$. The realizations are obtained from the empirical traffic traces. We set $e_i = 6 \times 10^{-6} \, \$/\text{req}$ according to \citep{AWSLambdaPricing}.

Table~\ref{tab:dataset_setup} summarizes the dataset-based simulation setup,
including the number of SPs, the number of time slots, the length of the time slot, and the simulation investment period.

\begin{table}[H]
\centering
\caption{Dataset-based simulation setup.}
\label{tab:dataset_setup}
\begin{tabular}{lll}
\toprule
Quantity & Symbol & Value \\
\midrule
Number of SPs & \(N\) & 5 \\
Number of time slots & \(|\pazocal T|\) & 1440 \\
Time slot length & $\Delta$ & 1 hour\\
Investment period & \(I\) & 2 months \\
\bottomrule
\end{tabular}
\end{table}

Table~\ref{tab:dataset_equilibrium} reports the equilibrium outcomes obtained
from the dataset-based experiment, including the InP's optimal capacity, access
price, profit, and computational runtime.

\begin{table}[H]
\centering
\caption{InP equilibrium outcomes for the dataset-based experiment.}
\label{tab:dataset_equilibrium}
\begin{tabular}{lll}
\toprule
Quantity & Symbol & Value \\
\midrule
Optimal capacity & \(C^*\) & 109 vCores \\
Optimal access price & \(\theta^*\) & \$27 per vCore-hour \\
InP profit & -- & \$1.37 million \\
Runtime & -- & 5 minutes \\
\bottomrule
\end{tabular}
\end{table}

Figure~\ref{fig:dataset_profit_guarantee_heatmap} reports the SP profit guarantees obtained in the dataset-based experiment. The empirical PoP is equal to one for all SPs, indicating that every SP remains profitable across the tested realizations. The theoretical lower bound on the PoP $\hat{\nu}_i$ is also high for all SPs, ranging from $0.94$ to $1.00$. This confirms that the equilibrium resource commitment decisions provide a strong profit guarantee even under real mobile-traffic data.
The figure also shows that the guarantee is not identical across SPs. SP~2 and SP 3 obtain the strongest lower bounds while SP~1 has the lowest lower bound. This difference reflects the heterogeneity in risk preferences and revenue levels across SPs. 
The small gap between the empirical PoP and the lower bound also indicates that the theoretical guarantee is relatively tight in this data-driven setting, especially for the more risk-averse SPs.
Overall, the real-data results are consistent with the findings obtained from the synthetic experiments.
\begin{figure}[t]
    \centering
    \includegraphics[width=0.82\linewidth]{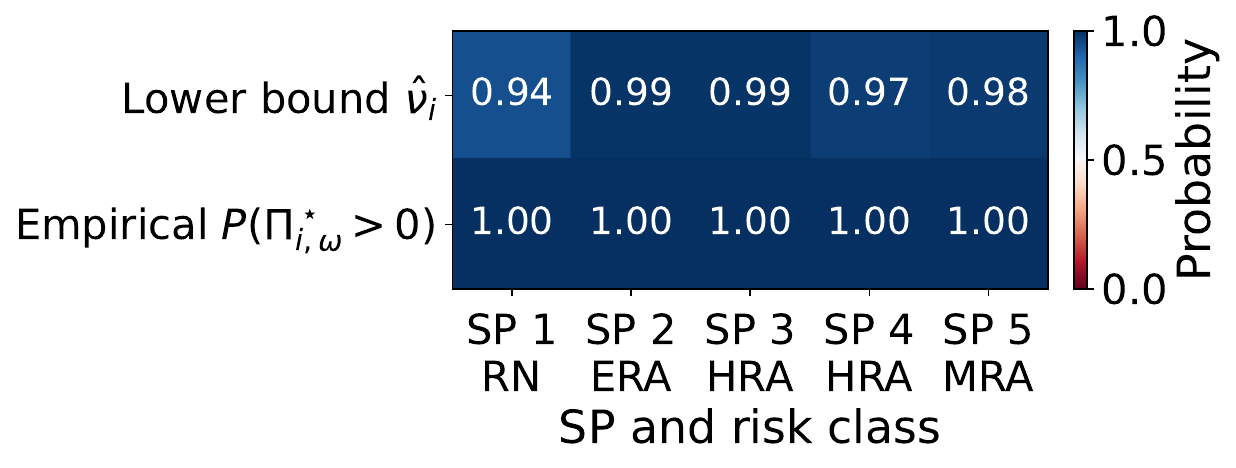}
    \caption{Dataset-based profit guarantees across SPs. Each column corresponds to one SP, with its risk class shown below the SP label. The first row reports the theoretical lower bound $\hat{\nu}_i$, while the second row reports the empirical PoP $\mathbb P(\Pi_{i,\omega}^*>0)$.}
    \label{fig:dataset_profit_guarantee_heatmap}
\end{figure}

\section{Benchmark Comparison}
\label{app:benchmark_comparison}
We evaluate the proposed model through three representative benchmarks that reflect the main modeling directions in the related literature. Because those studies are built on different assumptions, objectives, and decision structures, we do not attempt a direct numerical reproduction. Instead, we use a common experimental setting and vary one modeling layer at a time, which allows a fair comparison of the economic impact of capacity dimensioning, access pricing, and endogenous SP resource choices.

The capacity-only benchmark is motivated by capacity-investment models under uncertainty, which optimize infrastructure dimensioning without joint access-price design \citep{zhao2019flexible,sakr2026coanor}. The price-only benchmark is motivated by access-pricing and infrastructure-sharing models, which optimize price for a given infrastructure capacity \citep{huang2022edge,datar2022strategic}. The third benchmark is motivated by joint capacity--pricing models \citep{fu2018private,huang2025strategic,dong2021capacity,maglaras2003pricing}. These models capture important InP-level capacity and pricing decisions, but they do not include the SP-side structure considered here: revenue-generating SPs committing to resource levels under revenue uncertainty, risk-aware take-or-pay contracts, operational and congestion costs, and a shared capacity constraint. To obtain a benchmark that is comparable within our setting, we therefore use a joint capacity--pricing benchmark with exogenous utilization, where the InP optimizes both capacity and access price, while SP resource usage is given by a forecast rather than by the endogenous SP equilibrium.
Together, these benchmarks provide a controlled ablation of the proposed framework: the capacity-only benchmark isolates the value of access-price optimization, the price-only benchmark isolates the value of capacity dimensioning, and the exogenous-utilization benchmark isolates the value of modeling strategic risk-aware SP responses under uncertainty.
A key feature of the proposed model is that SPs can adjust their resource commitments to control downside profit exposure, which allows us to derive probabilistic lower bounds on realized positive profit.

Fig.~\ref{fig:benchmark_comparison} compares the proposed risk-aware Stackelberg model with the benchmarks under the HRA setting and $CV=100\%$. Overall, the proposed model provides the best balance across the main performance dimensions. It maintains the highest InP profit, high utilization, and strong lower bounds on the SPs' PoP at equilibrium.
The comparison with the capacity-only and price-only benchmarks demonstrates the advantage of jointly optimizing the InP's investment and pricing decisions. As shown in Fig.~\ref{fig:benchmark_utilization}, the capacity-only benchmark maintains high utilization, but fixing the access price restricts the InP's ability to balance access-revenue extraction against SP resource commitments. As a result, it yields lower InP profit in Fig.~\ref{fig:benchmark_profit} and a weaker lower bound for SP~1 and SP~3 in Fig.~\ref{fig:benchmark_nuhat}. The price-only benchmark achieves a profit close to the proposed model in Fig.~\ref{fig:benchmark_profit}, but it does so with a larger capacity investment in Fig.~\ref{fig:benchmark_capacity} and lower average utilization in Fig.~\ref{fig:benchmark_utilization}. This indicates that optimizing price alone may lead to costly over-provisioning. The exogenous-utilization benchmark performs worst across the main dimensions. Since it treats utilization as an exogenous forecast, it does not account for the fact that utilization is generated by strategic SPs whose resource commitments depend on capacity, access price, congestion, revenue uncertainty, and risk preferences. Consequently, it selects an excessively high access price in Fig.~\ref{fig:benchmark_price}, which reduces SP resource commitments and leads to severe under-utilization in Fig.~\ref{fig:benchmark_utilization}, lower InP profit in Fig.~\ref{fig:benchmark_profit}, and weaker lower-bound guarantees in Fig.~\ref{fig:benchmark_nuhat}. This confirms that ignoring the equilibrium response can weaken SP profit and lead to poor investment and pricing decisions.

Thus, the proposed model combines the main advantages of the benchmarks while avoiding their limitations. It preserves high utilization as in the capacity-only benchmark, achieves profit close to or above the best benchmark outcome, avoids the over-investment pattern observed under price-only optimization, and substantially improves over the exogenous-utilization benchmark by endogenizing utilization through the SP equilibrium response. 

\begin{figure}[!t]
    \centering

    \begin{subfigure}[t]{0.48\textwidth}
        \centering
        \includegraphics[width=\linewidth]{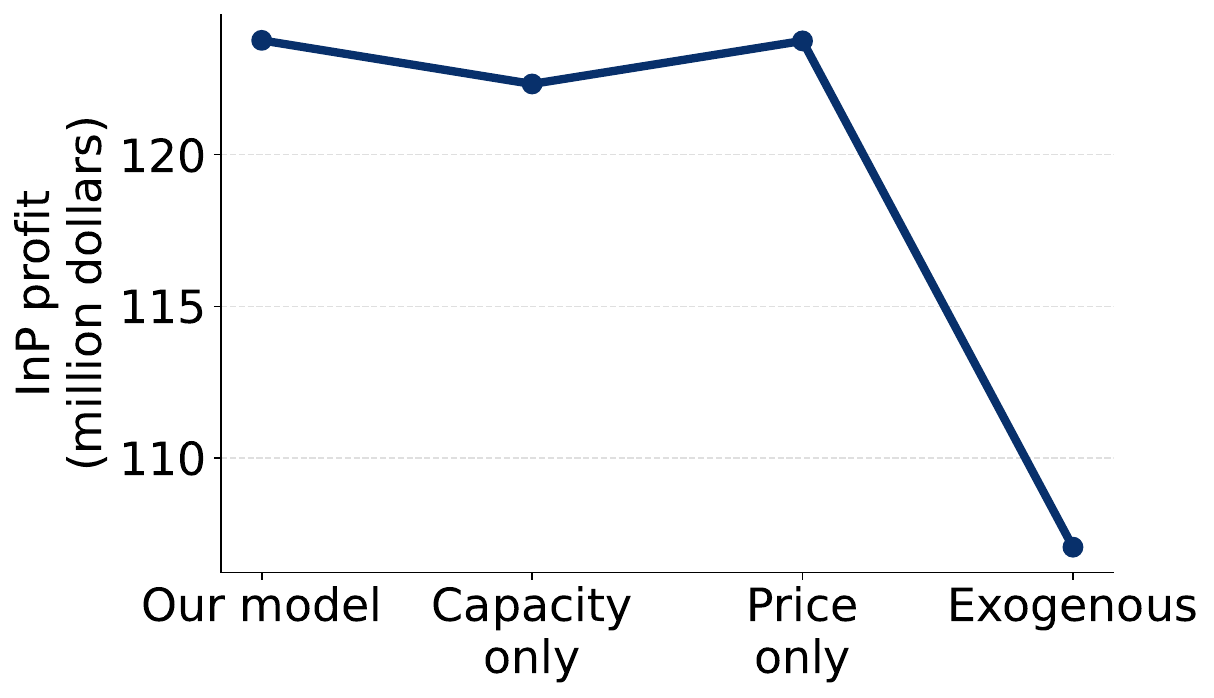}
        \caption{InP profit}
        \label{fig:benchmark_profit}
    \end{subfigure}
    \hfill
    \begin{subfigure}[t]{0.48\textwidth}
        \centering
        \includegraphics[width=\linewidth]{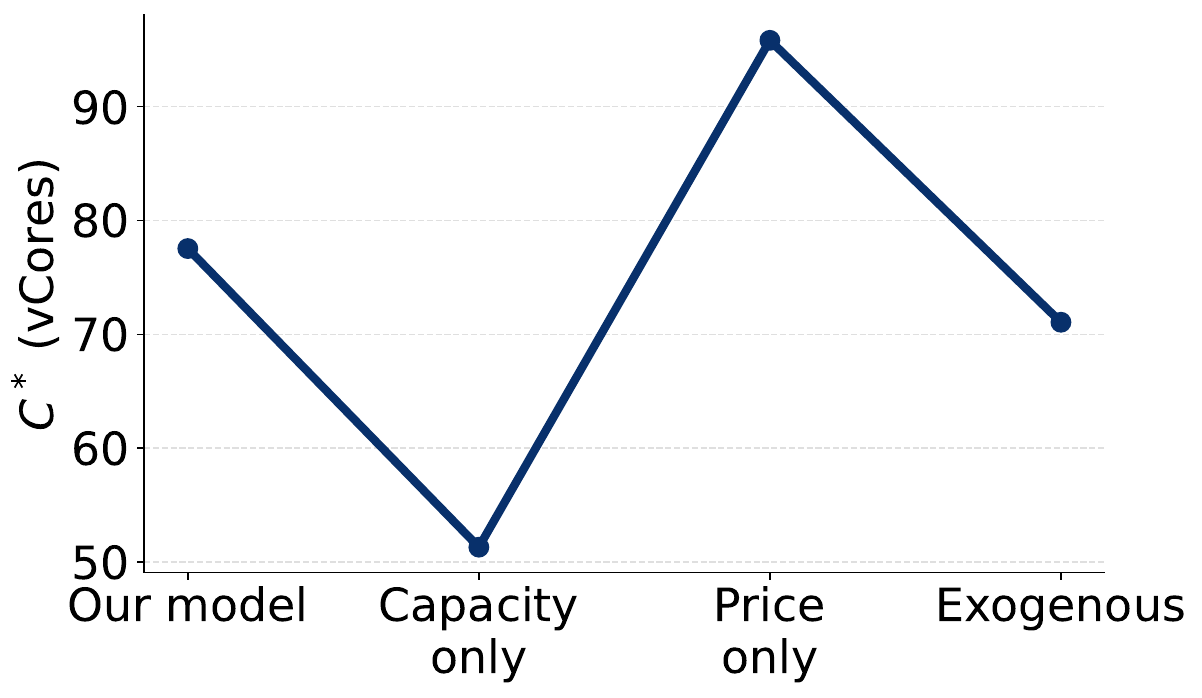}
        \caption{Optimal capacity}
        \label{fig:benchmark_capacity}
    \end{subfigure}

    \vspace{0.35cm}

    \begin{subfigure}[t]{0.48\textwidth}
        \centering
        \includegraphics[width=\linewidth]{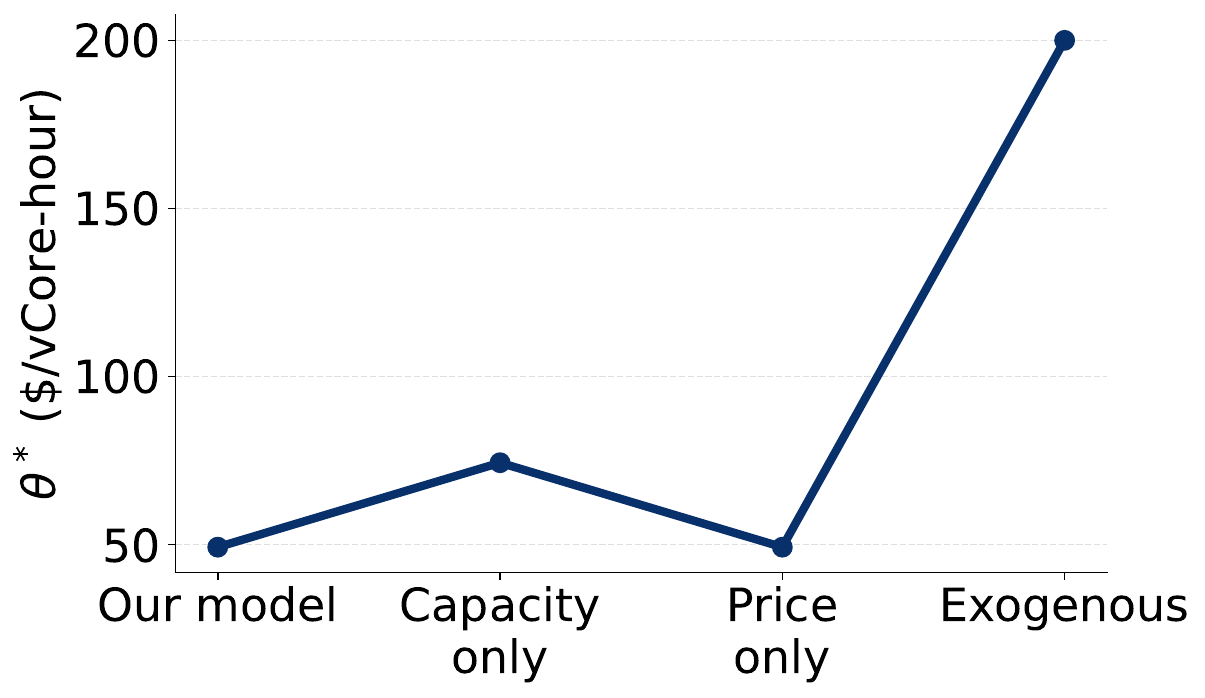}
        \caption{Optimal access price}
        \label{fig:benchmark_price}
    \end{subfigure}
    \hfill
    \begin{subfigure}[t]{0.48\textwidth}
        \centering
        \includegraphics[width=\linewidth]{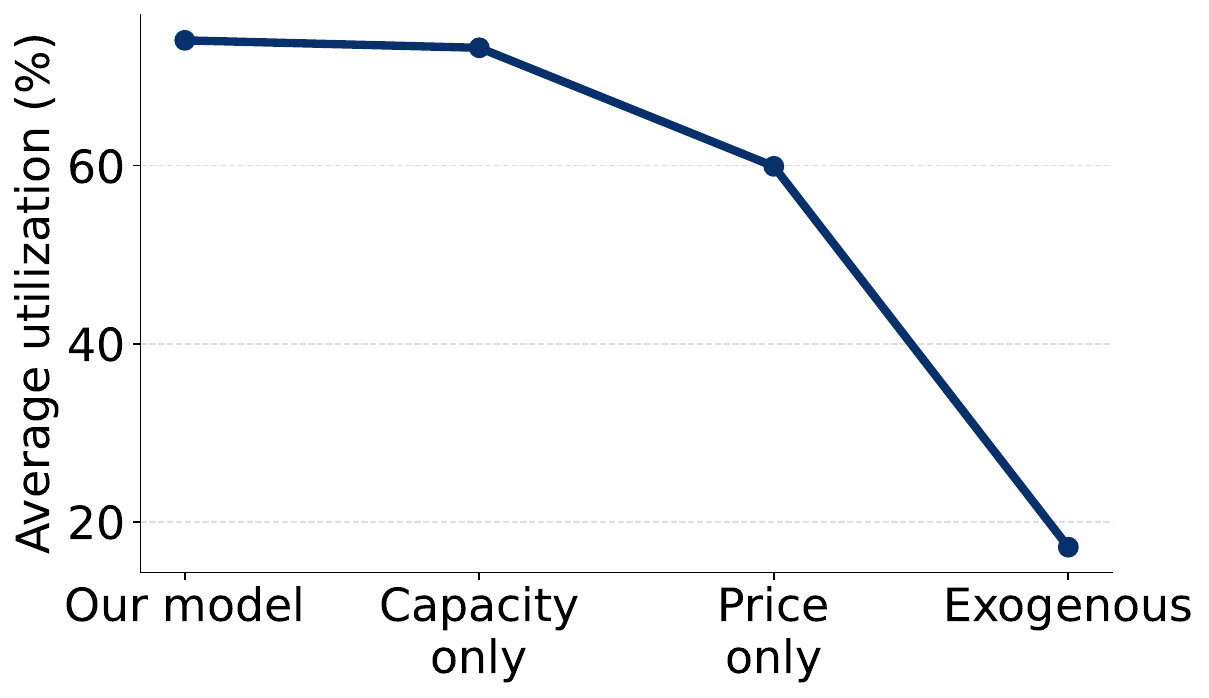}
        \caption{Average utilization}
        \label{fig:benchmark_utilization}
    \end{subfigure}

    \vspace{0.35cm}

    \begin{subfigure}[t]{0.6\textwidth}
        \centering
        \includegraphics[width=\linewidth]{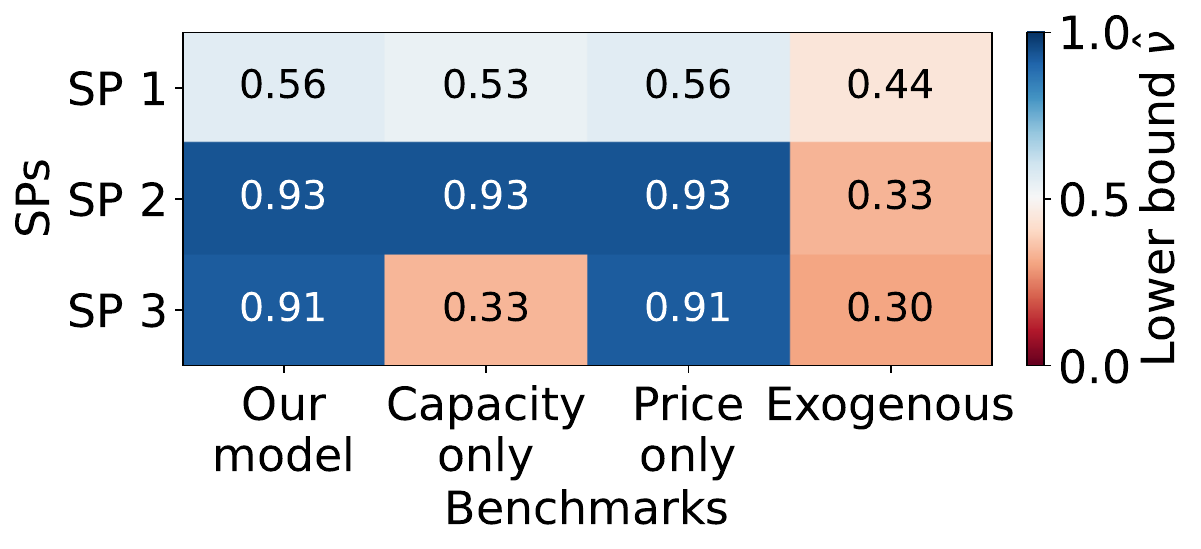}
        \caption{Lower bound on the Probability of Profit}
        \label{fig:benchmark_nuhat}
    \end{subfigure}

    \caption{Comparison of the proposed risk-aware Stackelberg model with three benchmarks under the HRA setting and $CV=100\%$.}
    \label{fig:benchmark_comparison}
\end{figure}





\bibliographystyle{elsarticle-harv}
\bibliography{ref}

\end{document}